\documentclass[pra,twocolumn,superscriptaddress,longbibliography,preprintnumbers]{revtex4-2}
\usepackage[latin1]{inputenc}
\usepackage{graphicx}
\usepackage{amsmath}
\usepackage{amsthm}
\usepackage{bm}
\usepackage{layout}
\usepackage{float}
\usepackage{amsfonts}
\usepackage{amssymb}%
\usepackage[margin=0.75in]{geometry}
\usepackage{color}
\usepackage{soul}
\usepackage{mathtools}
\usepackage[colorlinks=true,citecolor=blue]{hyperref}
\theoremstyle{definition}
\newtheorem{theorem}{Theorem}

\newtheorem{lemma}{Lemma}
\newtheorem{definition}{Definition}

\usepackage[capitalise,compress]{cleveref}

\usepackage{physics}
\newcommand{\avg}[1]{\left \langle #1 \right\rangle}
\usepackage{ulem}

\newcommand{\oket}[1]{ | #1 )}
\newcommand{\obra}[1]{ ( #1 |}
\newcommand{\obraket}[2]{\big ( #1 \big | #2 \big)}

\renewcommand{\epsilon}{\varepsilon}
\renewcommand{\O}[1]{\mathcal{O}\left(#1\right)}

\renewcommand{\emph}[1]{\textit{#1}}

\newcounter{para}

\makeatletter
\newcommand*\bigcdot{\mathpalette\bigcdot@{.5}}
\newcommand*\bigcdot@[2]{\mathbin{\vcenter{\hbox{\scalebox{#2}{$\m@th#1\bullet$}}}}}

\newcommand{\llangle}[1][]{\savebox{\@brx}{\(\m@th{#1\langle}\)}%
  \mathopen{\copy\@brx\kern-0.5\wd\@brx\usebox{\@brx}}}
\newcommand{\rrangle}[1][]{\savebox{\@brx}{\(\m@th{#1\rangle}\)}%
  \mathclose{\copy\@brx\kern-0.5\wd\@brx\usebox{\@brx}}}

\makeatother

\usepackage{array}
\newcolumntype{L}{>{$}l<{$}} % math-mode version of "l" column type
\newcolumntype{C}{>{$}c<{$}} % math-mode version of "c" column type
\newcolumntype{R}{>{$}r<{$}} % math-mode version of "r" column type

\usepackage{xcolor}

\usepackage{pifont}% http://ctan.org/pkg/pifont

\usepackage{mathtools}

\usepackage{xr}
\makeatletter
\newcommand*{\addFileDependency}[1]{% argument=file name and extension
  \typeout{(#1)}
  \@addtofilelist{#1}
  \IfFileExists{#1}{}{\typeout{No file #1.}}
}
\makeatother

\usepackage{tikz}
\renewcommand{\L}{\mathcal L}

\usepackage{mdframed}
\newmdenv[topline=false,rightline=false,bottomline=false,linewidth=2pt,linecolor=white!60!black,]{leftborder}

\newcommand{\Var}{\text{Var}}

\usepackage{physics}

\renewcommand{\L}{\left}
\newcommand{\R}{\right}

\newcommand{\sys}{{\text{sys}}}
\newcommand{\ext}{{\text{ext}}}
\newcommand{\anc}{{\text{anc}}}
\newcommand{\ancs}{{}}

\newcommand{\scrambler}{scrambling map}
\newcommand{\can}{\text{can}}

\newcommand{\ryd}{\text{Ryd}}
\newcommand{\FH}{\text{FH}}

\newcommand{\Haar}{\text{Haar}}
\usepackage{array}
\newcolumntype{P}[1]{>{\centering\arraybackslash}p{#1}}
\newcommand{\est}{\text{bound}}
\usepackage{bbm}

\usepackage{booktabs}

\makeatletter
\DeclareFontFamily{OMX}{MnSymbolE}{}
\DeclareSymbolFont{MnLargeSymbols}{OMX}{MnSymbolE}{m}{n}
\SetSymbolFont{MnLargeSymbols}{bold}{OMX}{MnSymbolE}{b}{n}
\DeclareFontShape{OMX}{MnSymbolE}{m}{n}{
    <-6>  MnSymbolE5
   <6-7>  MnSymbolE6
   <7-8>  MnSymbolE7
   <8-9>  MnSymbolE8
   <9-10> MnSymbolE9
  <10-12> MnSymbolE10
  <12->   MnSymbolE12
}{}
\DeclareFontShape{OMX}{MnSymbolE}{b}{n}{
    <-6>  MnSymbolE-Bold5
   <6-7>  MnSymbolE-Bold6
   <7-8>  MnSymbolE-Bold7
   <8-9>  MnSymbolE-Bold8
   <9-10> MnSymbolE-Bold9
  <10-12> MnSymbolE-Bold10
  <12->   MnSymbolE-Bold12
}{}

\let\llangle\@undefined
\let\rrangle\@undefined
\DeclareMathDelimiter{\llangle}{\mathopen}%
                     {MnLargeSymbols}{'164}{MnLargeSymbols}{'164}
\DeclareMathDelimiter{\rrangle}{\mathclose}%
                     {MnLargeSymbols}{'171}{MnLargeSymbols}{'171}
\makeatother

\usepackage{algpseudocode}
\usepackage{algorithm}
\makeatletter
\renewcommand{\ALG@name}{Algorithm}
\makeatother\usepackage{comment}

\begin{document}

\title{Measuring Arbitrary Physical Properties in Analog Quantum Simulation}

\author{Minh C. Tran}
\altaffiliation{These authors contributed equally to this work.}
\affiliation{Center for Theoretical Physics, Massachusetts Institute of Technology, Cambridge, MA 02139, USA}
\affiliation{Department of Physics, Harvard University, Cambridge, MA 02138, USA}

\author{Daniel K. Mark}
\altaffiliation{These authors contributed equally to this work.}
\affiliation{Center for Theoretical Physics, Massachusetts Institute of Technology, Cambridge, MA 02139, USA}

\author{Wen Wei Ho}
\altaffiliation{Corresponding author 1: \href{mailto:wenweiho@nus.edu.sg}{wenweiho@nus.edu.sg}}
\affiliation{Department of Physics, Stanford University, Stanford, CA 94305, USA}
\affiliation{Department of Physics, National University of Singapore, Singapore 117542}

\author{Soonwon Choi}
\altaffiliation{Corresponding author  2: \href{mailto:soonwon@mit.edu}{soonwon@mit.edu}}
\affiliation{Center for Theoretical Physics, Massachusetts Institute of Technology, Cambridge, MA 02139, USA}

\date{\today}

\preprint{MIT-CTP/5503}

\begin{abstract}
A central challenge in analog quantum simulation is to characterize desirable physical properties of quantum states produced in experiments. However, in conventional approaches, the extraction of arbitrary information requires performing measurements in many different bases, which necessitates a high level of control that present-day quantum devices may not have. 
Here, we propose and analyze a scalable protocol that leverages the ergodic nature of generic quantum dynamics, enabling the efficient extraction of many physical properties. 
The protocol does not require sophisticated controls and can be generically implemented in analog quantum simulation platforms today. 
Our protocol involves introducing ancillary degrees of freedom in a predetermined state to a system of interest, quenching the joint system under Hamiltonian dynamics native to the particular experimental platform, and then measuring globally in a single, fixed basis.
We show that arbitrary information of the original quantum state is contained within such measurement data, and can be extracted using a classical data-processing procedure. 
We numerically demonstrate our approach with a number of examples, including the measurements of entanglement entropy, many-body Chern number, and various superconducting orders in systems of neutral atom arrays, bosonic and fermionic particles on optical lattices, respectively, only assuming existing technological capabilities.
Our protocol excitingly promises to overcome limited controllability and, thus, enhance the versatility and utility of near-term quantum technologies.
\end{abstract}
\maketitle

\section{Introduction}\label{sec:intro}
One of the most promising applications of near-term quantum technologies is analog quantum simulation: by coherently manipulating a system of many particles,
a myriad of complex quantum phenomena can be controllably simulated.
Ranging from platforms of atoms, molecules, optical elements to solid-state systems, quantum simulators enable the  study of physics across many domains and scales, bringing fresh insights to unsolved, fundamental problems.
Examples include understanding high temperature superconductivity~\cite{m:hartObservationAntiferromagneticCorrelations2015,m:chiuStringPatternsDoped2019,m:hartkeDoublonHoleCorrelationsFluctuation2020,m:jiCouplingMobileHole2021}, probing new physics in quantum matter coupled to gauge fields~\cite{m:aidelsburgerRealizationHofstadterHamiltonian2013,m:aidelsburgerMeasuringChernNumber2015,m:yangObservationGaugeInvariance2020,m:leonardRealizationFractionalQuantum2022}, realizing exotic topological quantum matter~\cite{m:schollQuantumSimulation2D2021,m:ebadiQuantumPhasesMatter2021}, and studies of non-equilibrium phenomena~\cite{m:neyenhuisObservationPrethermalizationLongrange2017,m:choiProbingQuantumThermalization2019,m:pengFloquetPrethermalizationDipolar2021}.  
Indeed, early experiments have reported several discoveries such as novel dynamical phases of matter~\cite{m:choiObservationDiscreteTimecrystalline2017,m:zhangObservationDiscreteTime2017,m:auttiObservationTimeQuasicrystal2018,m:rovnyObservationDiscreteTimeCrystalSignatures2018,m:smitsObservationSpaceTimeCrystal2018,m:osullivanDissipativeDiscreteTime2019,m:kyprianidisObservationPrethermalDiscrete2021,m:randallObservationManybodylocalizedDiscrete2021} and the breaking of ergodicity in the form of quantum many-body scars~\cite{m:bernienProbingManybodyDynamics2017}. 

\begin{figure*}[t]
\includegraphics[width=0.85\textwidth]{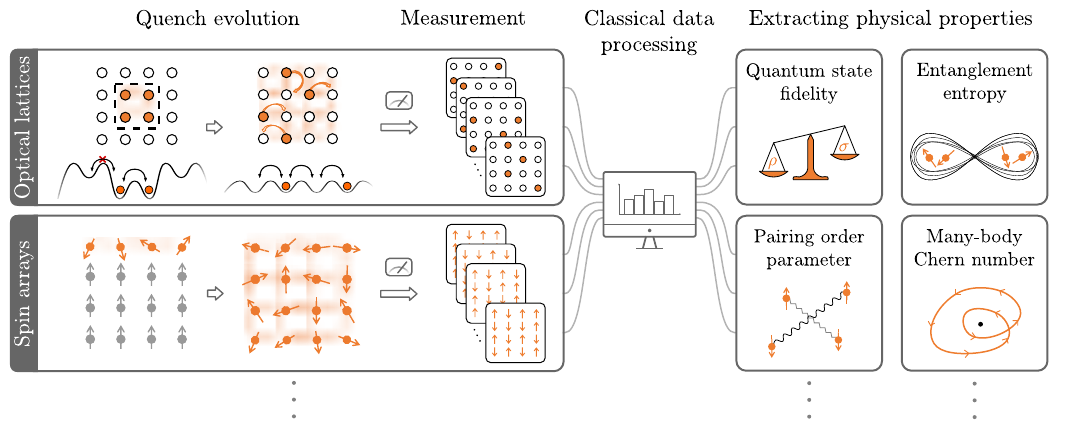}
\caption{
Schematic illustration of our protocol applied to two representative quantum simulators: itinerant particles in optical lattices (top) and systems of spins (bottom) realized by, for example, arrays of Rydberg atoms.
In both cases, we aim to extract properties of an unknown  state $\rho$ characterizing the system (orange circles or spins prior to the quench), initially prepared via independent quantum simulation methods. 
We assume the presence of nearby ancillary sites (empty circles) or spins (grey arrows) which are initially in a known, definite state $\ket{\phi_\ancs}$, 
decoupled with the system of interest, e.g. via large potential barriers (black dashed box).
The first step of our protocol is to allow the systems to interact with the ancillae: the time evolution of the global system generated by the natural Hamiltonian of the analog simulators. 
This step can be done in many ways; examples include lowering the potential barrier (itinerant particles in optical lattices) or by shuttling the ancilla spins closer to the system using optical tweezers (Rydberg atoms).
This quench evolution scrambles information from the system to the ancillary degrees of freedom and typically entangles them in the process.
The second step is to take snapshots of the extended system in some fixed basis, e.g.,  the particle occupation number or spin-polarization bases.
Arbitrary desired observables of the initial state $\rho$ can then be extracted through a final classical data processing step.
}
\label{fig:scheme}
\end{figure*}

Despite their exciting capabilities, analog quantum simulators face limitations. 
A particularly pressing challenge is the extraction of physical information in such platforms: even if a desired complex quantum system can be simulated and an important target quantum state is realized, it is often not obvious how to measure physical properties such as long-range correlators, entanglement entropies, or topological signatures of the system. 
This challenge stems from the fact that measurements are typically performed only in one or a few particular bases, such as the occupation number basis for quantum gas microscopes 
or the atomic level basis for neutral atoms in a Rydberg atom simulator.
Adopting quantum information science parlance, we call these natural measurement bases the ``standard basis." 
In contrast, measurements in more complicated, possibly non-local bases are difficult to achieve, owing to the need to employ basis rotations which lie beyond the reach of dynamical controls in present-day analog quantum simulators. 
Consequently, observables diagonal in the standard basis (e.g.~densities and their correlation functions) are readily accessible but off-diagonal observables (e.g.~current densities, off-diagonal correlation functions, or Wilson loops) are not. 
While there do exist schemes to measure some off-diagonal observables~\cite{m:ohligerEfficientFeasibleState2013,m:islamMeasuringEntanglementEntropy2015,m:pichlerMeasurementProtocolEntanglement2016,m:brydgesProbingRenyiEntanglement2019}, 
these often involve fine-tuned, ad-hoc schemes for specific target observables and are not easily generalizable.

In this work, we propose to overcome this challenge by introducing a universal---hardware and observable-independent---method to extract {\it arbitrary} physical properties in analog quantum simulators. We only assume the experimental capability to (i) ``expand'' the system of interest into a larger state space, (ii) coherently time-evolve the entire system under 
many-body dynamics native to the quantum simulator, and (iii) perform measurements globally in a single, fixed basis (see \cref{fig:scheme}). 
We will show that as long as such dynamics is ergodic and scrambling in nature---as expected for evolution by generic interacting systems, it is 
possible to recover {\it any} information about the prepared state upon appropriate classical processing of the resulting measurement data.
Notably, the classical computation required for this data processing step can be performed independently of the experimental data acquisition.
Furthermore, the experimental steps of the protocol are independent of the target observables.
Therefore, our protocol enables the adoption of a \textit{``measure first, ask questions later"} philosophy espoused in the related approach of \textit{randomized measurements}~\cite{m:ohligerEfficientFeasibleState2013,m:huangPredictingManyProperties2020,m:elbenRandomizedMeasurementToolbox2022}: 
one can imagine first collecting measurement data of a given experimental system with our protocol, only later deciding which quantities to extract via classical post-processing. This feature desirably alleviates the need to redesign an experiment to target any given specific observable. 

In our approach, the ergodicity of quantum dynamics, aided by classical computation, is harnessed as a resource for useful quantum information science applications.
Recent works employing such a principle include Ref.~\cite{m:boixoCharacterizingQuantumSupremacy2018,m:choiEmergentQuantumRandomness2022,m:markBenchmarkingQuantumSimulators2022}, 
wherein certain universal statistical properties arising from ergodic quantum dynamics are used for estimating the fidelity between a target pure quantum state and an experimentally prepared mixed state.
Here, we consider additionally introducing ancillary degrees of freedom in a controlled fashion, enabling the extraction of arbitrary physical properties while balancing required experimental and computational resources. Hence, our protocol is versatile and scalable, and thus promises to greatly expand the utility of current and near-term quantum simulators in characterizing quantum states which realize complex and interesting physical phenomena.

The paper is organized as follows. 
In \cref{sec:overview}, we start by first explaining the underlying working principles involved in our protocol. 
We state our protocol in \cref{sec:framework} and provide in \cref{sec:analysis} its performance analysis from various aspects, including the sample complexity, the classical computational overhead, and the robustness in the presence of noise.
Readers who are more interested in practical implementations of our protocol may skip \cref{sec:analysis} and refer to proof-of-principle numerical examples in \cref{sec:rydberg,sec:bcs,sec:bose-hubbard}, where we apply the protocol to extract interesting properties from a Rydberg array and itinerant particles on optical lattices.
In the first example (\cref{sec:rydberg}), we consider a Rydberg atom array experiment and show how various observables, including the quantum state fidelity, the entanglement entropy, and arbitrary local observables, can be extracted, with a modest number of measurement snapshots.
% from a Rydberg array in the blockade regime. 
We also use this example to illustrate how different quench arrangements can be used to minimize the sample complexity for different target observables as well as classical computational requirements.
In the second example (\cref{sec:bcs}), we consider a quantum gas experiment with itinerant fermions in an optical lattice with single-site readout resolution. 
We extract the pairing order correlations from superconducting states of fermions, and show that our protocol can reliably distinguish between $s$-wave and $d$-wave superconducting orders, which are phenomena long-sought after in such systems.
In the last example (\cref{sec:bose-hubbard}), we consider an experiment of itinerant bosons in an optical lattice, and use the protocol to extract the many-body Chern number and measure local currents in a %Hofstadter-Bose-Hubbard model: a model of bosonic itinerant particles in the presence of an artificial gauge field. 
%system of bosons hopping on an optical lattice. 
topological state, realized by engineering an artificial gauge field.
This model contains non-trivial phases of matter, illustrating the bosonic fractional quantum Hall effect, and has been investigated both theoretically and experimentally~\cite{m:cooperFractionalQuantumHall2020,m:aidelsburgerMeasuringChernNumber2015,m:taiMicroscopyInteractingHarper2017,m:leonardRealizationFractionalQuantum2022}.
This last example demonstrates the power of our protocol for extracting observables that are otherwise extremely difficult to measure, overcoming the limited controllability of current experiments.
Finally, we conclude and discuss several open questions in \cref{sec:conclusion}.

\section{Overview of main ideas and key results} \label{sec:overview}

Before presenting and analyzing the technical details of our protocol in \cref{sec:framework,sec:analysis}, we first explain at a high level the key physical ideas and describe several important metrics for accessing its performance. 

\subsection{Ancillary system as a resource to perform randomized measurements}\label{sec:protocol-high-level}

Our aim is to characterize an unknown state $\rho$ of a system of interest, assuming  the ability to only perform measurements in a single, fixed direction. Na\"ively, this  precludes the extraction of observables which are off-diagonal in the measurement basis. 

However, suppose that instead of having access only to $\rho$, we also have access to an ancillary system prepared in some fiducial state $|\phi\rangle$, and the ability to couple them through a single, fixed, generic, but known unitary $U$. The extended system is therefore described by the density matrix
\begin{align}
\rho_{\ext} = U(\rho \otimes |\phi\rangle \langle \phi| )U^\dagger.
\end{align}
We claim that upon measuring the extended system in the same fixed basis as before, it is now generically possible to recover {\it any} information about $\rho$, including observables off-diagonal in the original measurement basis, solely from the probability distribution $P_z = \langle z| \rho_{\ext} |z\rangle$ of the measurement outcomes $z$.
In other words, by letting a system of interest ``expand'' into a larger space, one can infer initially ``inaccessible'' information about it. 
This mechanism is reminiscent of the celebrated time-of-flight (TOF) measurements performed in Bose-Einstein condensate (BEC) experiments~\cite{m:hallDynamicsComponentSeparation1998}, where upon releasing a  BEC from its confining trap such that it undergoes free expansion, its initial unknown momentum  distribution can be inferred by measuring density distributions of the cloud at later times.
Our approach can be considered a generalization of TOF measurements for strongly interacting quantum dynamics.

\begin{figure}[t]
\includegraphics[width=0.48\textwidth]{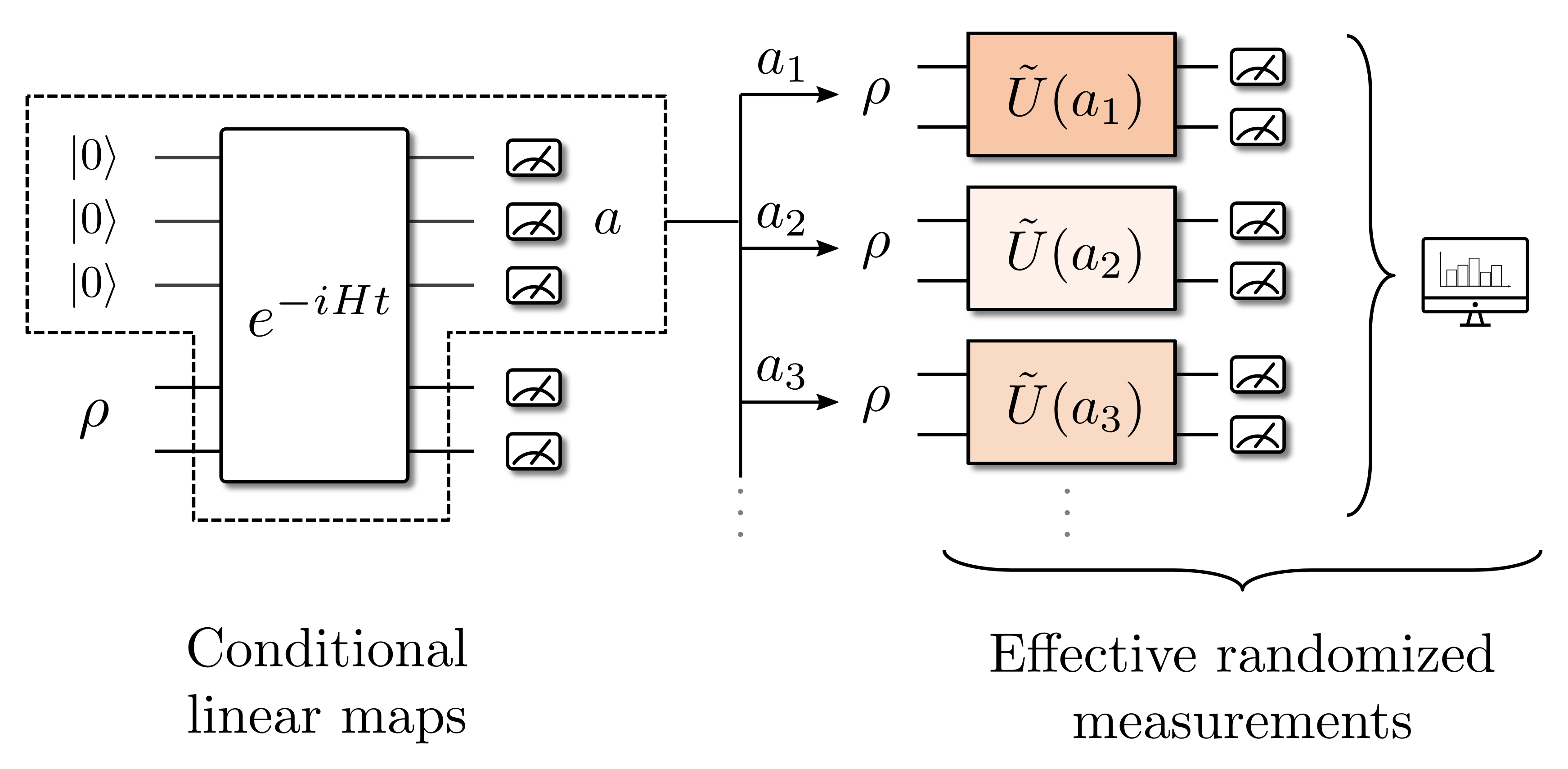}
\caption{
Illustration of the basic working principle of our protocol. 
Quench evolving the extended system by natural time-evolution $e^{-iHt}$, followed by measuring the ancillae is effectively equivalent to applying a random linear map $\tilde U(a)$ on the system described by the state $\rho$.
The linear map $\tilde U(a)$ is determined by the random outcome ${a}$ that the ancillae collapse to, which occurs with probability $P_a$.
This protocol therefore implements randomized rotations on $\rho$ and realizes effective randomized measurements.
}
\label{fig:why-it-work}
\end{figure}

To better understand why the distribution $P_z$ can contain all information about $\rho$, imagine for the sake of simplicity that the extended system consists of $n_\ext$ spin-$\frac{1}{2}$ particles (qubits), and that a measurement outcome yields a bit-string $z \in \{0,1\}^{n_\ext}$, which pertains to a particular classical configuration of spin-ups ($0$) and spin-downs ($1$).
We can imagine dividing the bit-string $z$ into two substrings $z = sa$, where $s$ $(a)$ is a bit-string describing the classical configuration on the system of interest (ancillae), which allows us to rewrite  the probability $P_z$ in a more suggestive way:
\begin{align}
P_{s,a} & = P_{s\vert a} P_a,
\end{align}
where $P_a$ is the probability to measure $a$ from the ancillae and $P_{s\vert a} \equiv P_{s,a}/P_a$ is the conditional probability to measure $s$ from the system given an outcome $a$ from the ancillae. 
The latter can be expressed as $P_{s\vert a} = \langle s |\rho_a |s\rangle $, where the  state $\rho_a = \tilde U(a) \rho \tilde U^\dag(a)/P_a$ is defined through the conditional linear map $\tilde{U}(a) = (\mathbb{I} \otimes \langle a|){U} (\mathbb{I} \otimes |\phi\rangle)$ acting on the system (\cref{fig:why-it-work}) and the normalization factor $P_a$\,$=$\,$\Tr(\tilde{U}(a) \rho \tilde{U}(a)^\dagger)$.
The new expression gives a useful interpretation of $P_z$ as follows: we first measure the ancillae to obtain a random outcome $a$ with probability $P_a$, which transforms the remaining system $\rho \mapsto \rho_a$   according to the Born rule, 
and then measuring $\rho_a$ to yield  outcome $s$ with conditional probability $P_{s|a}$.
Equivalently, we can think of it as arising from effectively measuring the original density matrix $\rho$ in a ``rotated'' basis  $\sim \tilde{U}^\dagger(a)|s\rangle$, where the choice of  ``rotation'' $\tilde{U}^\dagger(a)$ is sampled with probability $P_a$.
Formally, $\{P_a, \tilde{U}^\dagger(a)\}$ forms an {\it ensemble} of random (non-trace preserving) quantum operations. The size of this ensemble is $d_\text{anc}$, the dimension of the ancillary space.
Thus, we see how the ancillary system can serve as a randomizer of measurement bases, and hence allow for matrix elements of $\rho$ which are off-diagonal in the original measurement basis to be probed. 
Note that this is a generalization of the concept of the projected ensemble recently considered in Refs.~\cite{mm:ippolitiSolvableModelDeep2022,m:cotlerEmergentQuantumState2021,m:hoExactEmergentQuantum2022,m:hoExactEmergentQuantum2022,m:wilmingHightemperatureThermalizationImplies2022,m:claeysEmergentQuantumState2022,mm:ippolitiDynamicalPurificationEmergence2022},
which is a distribution of quantum states generated from partial measurements of a single parent quantum state; here, we have a distribution of processes generated from partial measurements of a single unitary operator describing quantum dynamics.

When the density matrix $\rho$ can be fully reconstructed from the probability distribution $P_z$, we say that the protocol is {\it tomographically complete}. 
Obviously, we cannot expect tomographic completeness for every choice of coupling unitary $U$ or without any restrictions on the ancillae.
Indeed, the above discussion already highlights two important features that the coupling $U$ and ancillary system should have.
First, in order to achieve nontrivial basis changes, $U$ needs to be ergodic and, in a certain sense, be a sufficient ``scrambler'' of quantum information. For example, the trivial identity map~$U = \mathbb{I}$ will clearly not work because measurements outcomes on the ancillae do not depend on the state of the system, the two systems being always decoupled. 
We argue in this work that, with the coupling unitary $U$ generated by natural Hamiltonian dynamics $U = e^{-iHt}$ with reasonable times $t$, our protocol generically implies tomographic completeness (\cref{sec:recoverability,sec:quench_time}).
We also explain how the required evolution time $t$ is affected by the locality of the Hamiltonian and how to modify the protocol to account for exceptional cases such as the presence of symmetries that restrict the ergodicity of quantum dynamics (\cref{sec:analysis}).
Second, the dimension of the ancillae must be sufficiently large.
To fully characterize a density matrix of a system with dimension $d_\text{sys}$, it is a well-known fact in quantum state tomography that one has to perform at least $d_\text{sys}^2$ generalized measurements~\cite{m:buschInformationallyCompleteSets1991,m:debrotaInformationallyCompleteMeasurements2020}.
This requirement sets a lower bound on the number of effective rotations $\tilde{U}(a)$ and, consequently, a lower bound on the dimension of the ancillary space $d_\text{anc}$: $d_\text{anc} \geq d_\text{sys}$~\footnote{A set of generalized measurements is specified by a set of a positive, semi-definite operators $\{ E_i \}_{i=1}^N$ which sum to the identity: $\sum_{i=1}^N E_i = 
\mathbb{I}$, such that outcome $i$ occurs with probability $p_i = \Tr(E_i \rho)$. This set is also known as a positive operator-valued measure (POVM).
It is a fact in quantum state tomography that a POVM requires at least $N = d_\text{sys}^2$ elements for $\rho$ to be reconstructible from the statistics $p_i$. 
When $\rho$ is reconstructible, the POVM is called {\it informationally complete} ({\it minimally informationally-complete} if the number of elements $N$ is exactly $d_\text{sys}^2$). Our protocol can be equivalently cast in this language upon identifying $E_{s,a} = \tilde U(a)^\dag \ketbra{s} \tilde{U}(a)$, immediately yielding the claimed requirement $d_\text{anc} \geq d_\text{sys}$.}.

The basic working principle behind our protocol (\cref{fig:why-it-work}) 
 also immediately highlights a
connection to a recently-introduced quantum state-learning protocol called \emph{classical shadow tomography}~\cite{m:huangPredictingManyProperties2020}.
Indeed, the main idea behind both protocols is that of performing measurements in randomized bases, but the key difference between them is the source of this randomness. 
Classical shadow tomography  assumes the application of random unitary rotations $U$ (drawn from  ensembles with known statistics) to the initial state $\rho \mapsto U \rho U^\dagger$, using explicit dynamical control.
In our protocol, these effective random ``rotations'' $\rho \mapsto \rho_a  \propto \tilde{U}(a) \rho \tilde{U}(a)^\dagger$ 
are instead  induced by measurements on an ancillary system. 
For this reason,  our protocol may be termed {\it ancilla-assisted shadow tomography}.
This difference is also the reason for the comparative advantage of our protocol over classical shadow tomography in terms of the ease of experimental implementation: the level of dynamical control required in the former is arguably much less than in the latter. 
 We refer to \cref{app:classical_shadow_tomography} for an elaboration of the connection of our protocol to classical shadow tomography.

\subsection{Scrambling and recovery maps}
\label{sec:data_processing}

We now explain schematically the classical data processing steps involved in recovering information about the system of interest $\rho$.
For a given coupling $U$ and measurement basis, one can construct 
a map $S$ that takes the initial state $\rho$ to the probability distribution of the measurement outcomes $P_z$ of the extended system (given by the Born rule):
\begin{align}
  \oket{\rho} = 
  \begin{pmatrix}
    \rho_{1,1}\\
    \rho_{1,2}\\
    \vdots\\
    \rho_{d_\sys,d_\sys}
  \end{pmatrix}
  \xrightarrow{\ \ S\ \ } 
  \ket{P} = 
  \begin{pmatrix}
    P_1\\
    P_2\\
    \vdots\\
    P_{d_\ext}
  \end{pmatrix},\label{eq:S-high-level}
\end{align}
where we have rewritten both the density matrix $\rho$ and the probability $P_z$ as column vectors denoted by $\oket{\rho}$ and $\ket{P}$. 
Here, $d_\ext = d_\anc d_\sys$ is the dimension of the extended system.
We illustrate the construction of $S$ diagrammatically in \cref{fig:Pz-diagram} and define it precisely in \cref{sec:framework}, but the salient point is that it can be obtained solely from knowledge of the coupling unitary $U$ and initial state of the ancillae. 
As the map $S$ is precisely the agent responsible for scrambling information from the system into the larger space, we  refer to it as the  {\it scrambling map}. 
Note that since $S$ is linear, it can be represented by a matrix, which has dimension $d_\ext \times d_\sys^2$.

\begin{figure}[t]
\includegraphics[width=0.3\textwidth]{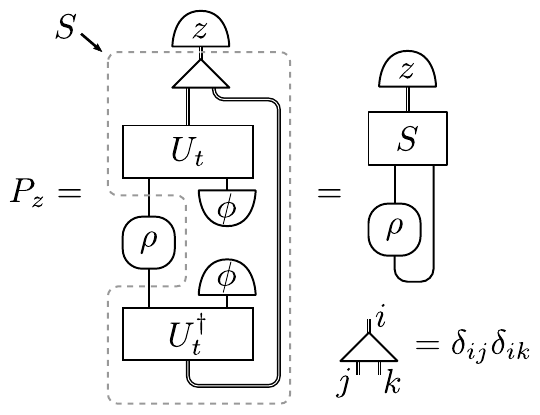}
\caption{Tensor diagram illustrating the construction of the \scrambler~$S$ in \cref{eq:S-high-level}. The thick double lines indicate bonds of dimension $d_\ext$. The three-legged triangle tensor takes unit values if the three indices are equal and vanishes otherwise.}
\label{fig:Pz-diagram}
\end{figure}

Tomographic completeness of our protocol is equivalent to the fact that the map $S$ is left-invertible.
Indeed, if there exists a linear map $R$ such that $RS = \mathbb I$, the map $R$ will take the outcome probability vector $\ket{P}$ back to the initial state $\rho$:
\begin{align}
R\ket{P} = \oket{\rho}. 
\label{eq:recovery_map}
\end{align}
Naturally, we call $R$ the {\it recovery map}.
Since $S$ is assumed known,   $R$ can also be computed as we show below.
Note that generally, if $d_\ext>d_\sys^2$, $S$ is not a square matrix and the recovery map $R$ is not unique. 
We will explain in \cref{sec:frame} the relative advantages and disadvantages of different constructions of $R$ and how they affect the experimental and computational resources required for our protocol.
We also remark that in certain scenarios (such as in the limit of large ancillae prepared at infinite effective temperature), $S$ and $R$ may assume universal forms with known analytic expressions---arising from approximate designs---as uncovered by the recent related works on projected ensembles in Refs.~\cite{mm:ippolitiSolvableModelDeep2022,m:cotlerEmergentQuantumState2021,m:hoExactEmergentQuantum2022,m:hoExactEmergentQuantum2022,m:wilmingHightemperatureThermalizationImplies2022,m:claeysEmergentQuantumState2022,mm:ippolitiDynamicalPurificationEmergence2022}. However, our protocol does not require the emergence of such universal behavior.

In practice, measurements in experiments yield bit-strings $\{z_1,\dots,z_m\}$ sampled from the distribution $P_z$.
Each $z_j$ is associated with an indicator vector $\ket{z_j}$ whose entries are all zero except for one element corresponding to the configuration $z_j$. 
By (numerically) applying $R$ onto each observed bit-string and averaging,  one can obtain an estimate of the initial state:
\begin{align}
  \frac{1}{m}\sum_{j=1}^m R\ket{z_j} \xrightarrow{\ m\rightarrow \infty\ } R \ket{P} = \oket{\rho}.
 \label{eqn:recovery_rho}
\end{align}
In other words, it is in principle possible to tomographically reconstruct the entire density matrix in the limit of large number of samples $m$.

However, while tomographic completeness is an important theory concept in this work, we emphasize that our primary motivation is often not to fully reconstruct $\rho$. 
Instead, our focus is to efficiently extract certain (we stress: not all) desired physical properties of $\rho$, such as the expectation values of a small subset of observables, many-body fidelities, or entanglement entropies etc. This task can be distinguished from that of quantum state tomography by the term {\it quantum state learning}, and has important practical differences. Indeed, it is well-known that the determination of an entire quantum state to within fixed precision requires a number of measurements that is exponential in system size, rendering recovery of the density matrix practically infeasible for a system with a large number of particles. In contrast, the latter task can place  significantly fewer demands on the experimental resources required (see Sec.~\ref{subsec:sample_complexity} and \ref{sec:quench-setups}).

Without fully reconstructing $\rho$, an estimate of the expectation value $\langle O \rangle$ of an  observable $O$ can instead be directly obtained from the measurement data $\{z_j\}$.
We present a way to construct a \emph{single-shot estimator} $o_{z}$ as a function of $z$ such that averaging $o_{z}$ over experimentally measured $\{z_1, \dots,  z_m \}$ amounts to estimating the desired quantity:
\begin{align}
\langle  O \rangle &\approx \frac{1}{m} \sum_{j=1}^m o_{z_j}, \qquad o_z = \obra{O^\dagger} R|z\rangle,\label{eq:estimate_O}
\end{align}
where $\oket{O}$ is the vectorized version of the operator $O$ (similarly to how we rewrote the density matrix $\rho$ as a vector $\oket{\rho}$ earlier) and $\obra{O} = \oket{O}^\dag$.
With a good estimator, the sample averaging of $o_z$ may converge much faster than that of $\rho$ in \cref{eqn:recovery_rho}, implying $\langle O \rangle$ can be learned much more efficiently with fewer samples without explicit quantum state tomography. 
Finally, we note that \cref{eq:estimate_O} can be generalized to extract nonlinear observables on $\rho$, such as the R\'enyi-2 entropy, 
which requires two copies of $\rho$.

\subsection{Experimental implementation}
\label{sec:expansion}

The crux of our proposed protocol lies in the ability to ``expand'' the state space.  Here, we elaborate on how this can be concretely realized in the context of present-day experimental quantum simulator platforms. 

Importantly, the expansion of the state space can be achieved in many different ways and is dependent on the experimental system at hand.
For example, for cold atoms on an optical lattice, one possibility is for the system of interest to be a block of sites residing in the bulk of the lattice,  and the surrounding sites  to be the ancillary space (\cref{fig:scheme}, top row) initially prepared in a known state. For example, they can be empty (the vacuum) or they can have one atom each (the Mott insulating state with unity filling).  By imposing a sufficiently high potential barrier, one can keep the system and ancillae well-separated throughout the course of a (separate) experiment, at the end of which the system is described by the state $\rho$. For instance, $\rho$ could be the result of preparing the ground state of a simulated model in some parameter regime, or it could be the state achieved after quench dynamics in experiments probing transport.
Our protocol enters when we want to characterize $\rho$. In this set-up, a natural way of ``expanding'' the state space  would be to lower the barrier to allow mixing between the two subsystems, i.e. quench the global system for some short time, before measuring.

As another example, in arrays of trapped Rydberg atoms, we can imagine expanding the state space by using optical tweezers to physically shuttle ancillary atoms from an initially isolated, non-interacting reservoir to be near the atoms of interest and allow them to interact, an ability that has been demonstrated in Ref.~\cite{m:bluvsteinQuantumProcessorBased2022}.
We stress again, though, that this is but only one possibility of ``expansion'' in this platform. 
In fact, introducing ancillary degrees of freedom does not even necessitate introducing physically distinct particles as in the previous two examples: the Hilbert space could also be expanded via allowing mixing to different internal or motional levels---beyond those normally utilized for a qubit encoding---of an atom or a molecule. 
We note that such capabilities are an exciting direction of current experimental development~\cite{m:allcockOmgBlueprintTrapped2021,m:chenAnalyzingRydbergbasedOpticalmetastableground2022,m:wuErasureConversionFaulttolerant2022,m:strickerExperimentalSinglesettingQuantum2022}.

\begin{figure}[t]
\includegraphics[width=0.48\textwidth]{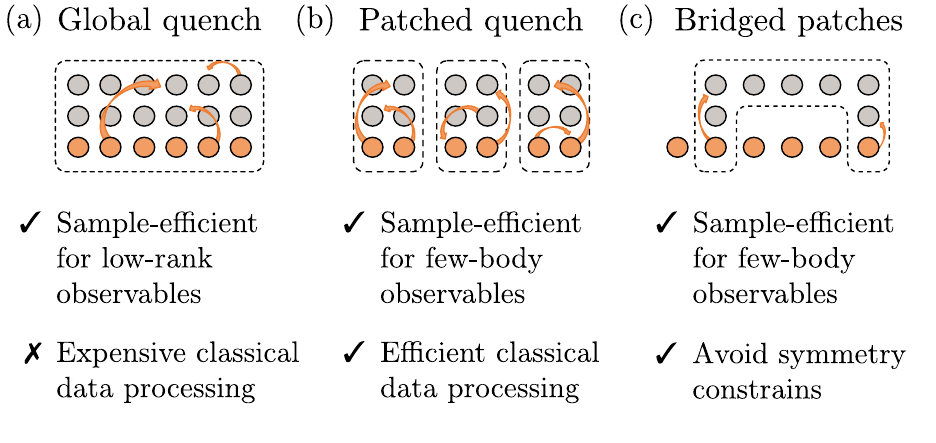}
\caption{
Different quench setups featuring different interactivity between the system and ancillae, which change the sample complexity of observables. 
In the figure, orange circles denote sites of system of interest while gray circles denote ancillary sites, and only particles within the same dashed boxes are assumed to interact during the quench.
The global quench yields a low sample complexity for global, low-rank observables, whereas the patched quench is sample-efficient for few-body (i.e., local) observables.
The bridged setup is a modification to the patched quench to circumvent symmetry constraints that may be present in the natural dynamics of the extended system, e.g., each patch conserving their individual particle number.
}
\label{fig:schemes}
\end{figure}

Besides the different choices of what the physical constituents of the extended space are, there is also a great deal of flexibility for the connectivity between the ancillae and the system. 
For example, we may allow the system to interact with a common set of ancillae (``global quench"), or divide the system into smaller disjoint patches, such that each patch couples to their own set of ancillae (``patched quench"), see \cref{fig:schemes}(a,b).
While tomographic completeness is largely independent of what the connectivity between the ancillae and system is
[an exception will be dynamics with symmetries, in which case we need a careful arrangement of the ancillae (``bridged quench"), \cref{fig:schemes}(c)]
, we will show that
different connectivity arrangements in relation to target observables have important, practical differences in terms of their performance captured by the protocol's sample complexity and computational complexity.
We briefly explain these factors below and analyze them more carefully in \cref{sec:quench-setups}.

\subsection{Sample complexity and Computational complexity}
\label{subsec:sample_complexity}
The performance and experimental feasibility of our protocol is assessed by two key metrics: sample complexity and computational complexity.
The former metric---the number of measurement snapshots in experiments required to produce a good approximation of the target quantity, called the {\it sample complexity}---depends on the choice of observable, and on the interactivity of the ancillae with the system.  
In particular, we expect sample complexity to be independent (or at worst, mildly-dependent) on system size for low-rank observables, e.g. the fidelity of $\rho$ to a pure reference state, upon using a ``global quench'' [\cref{fig:schemes}(a)]; or for few-body observables, upon using a ``patched quench'' [\cref{fig:schemes}(b)]. 
We elaborate on this point in \cref{sec:quench-setups} by drawing from insights provided by classical shadow tomography~\cite{m:huangPredictingManyProperties2020,m:huClassicalShadowTomography2021}.
In ideal cases, the sample complexities of our approach with global and patched quench is expected to be comparable to those of classical shadow tomography enabled by global and local random unitary circuits.

The second metric---the resources required of a classical computer for post-processing, called the {\it computational complexity}---is dominated by the cost of computing the scrambling map $S$ and the recovery map $R$.
Except for special cases where they have closed-form, analytic expressions, the cost of computing $S$ and $R$ generally scales exponentially with the size of the extended system.
This impediment can be overcome by imposing a local structure onto the scrambling map---for example, using a patched quench [\cref{fig:schemes}(b)]. 
Then, the quench unitary naturally factorizes into tensor products of local quench unitaries, each involving only degrees of freedom in an individual patch.
By limiting the size of the largest patch, we can efficiently control the computational overhead of our protocol.
Consequently, our protocol is both sample efficient and computationally tractable when extracting local observables.

For global observables such as state fidelities or R\'enyi-2 entropies, there is generally a trade-off between sample complexity and computational complexity: using the patch quench with a smaller patch size 
lowers the computational complexity of the protocol but increases its sample complexity.
Ultimately, the optimal patch size for extracting these observables is determined by carefully balancing the experimental and computational resources at hand~\cite{m:huClassicalShadowTomography2021}.

\subsection{Robustness against noise}
Another important and practical aspect we have to contend with is noise, which is ubiquitous in current-day quantum simulators. 
In \cref{sec:noise}, we discuss two strategies for dealing with noisy quench dynamics. To summarize our analysis: on the one hand, if the noise rate is sufficiently small, we argue that we can process the measurement snapshots as if there were no noise in the experiment. This approach results in systematic errors that cannot be reduced by taking more samples. We estimate the magnitude of such systematic errors and show that they depend linearly on the noise rate.
On the other hand, if the noise rate is sufficiently high, and if one can describe the noisy dynamics with high accuracy, we can invert the noisy linear map describing the quench evolution. While there is no systematic error in this case, we will argue that the presence of noise typically increases the sample complexity compared to the noiseless scenario.

\begin{table}[t]
\label{tab:protocol}
    \caption{Protocol for extracting the expectation value of an arbitrary observable $O$. The ``experiment" and ``pre-computation" steps can be done independently; the ``data processing" step uses both inputs to estimate $\Tr(\rho O)$.} 
    \label{tab:protocol}
    \begin{tabular}{p{0.48\textwidth}}
        \toprule
        \textbf{Input: } A quantum state $\rho$ and an observable $O$.\\
        \textbf{Output: } The observable expectation value $\Tr(\rho O)$.\\        
        \midrule
        \textbf{Experiment:}
        \begin{enumerate}
            \item Prepare an ancillary system with dimension at least as large as the system of interest, in a known state $\ketbra{\phi}$.
            \item Time-evolve the extended system $\rho \otimes \ketbra{\phi}$ under a joint quench evolution $U_t = \exp(-iHt)$.
            \item Measure the extended system in the standard basis, obtaining an outcome $z$.
            \item Repeat steps 1.~to 3.~$m$ times, to obtain $m$ samples $\{z_1,\dots,z_m\}$.
        \end{enumerate}
        \textbf{Pre-computation:}
        \begin{enumerate}
        \setcounter{enumi}{4}
            \item Compute the scrambling map $S$ using $\ketbra{\phi}$ and $U_t$ [\cref{eq:linear_map_Q}, \cref{fig:Pz-diagram}].
            \item Compute the (non-unique) recovery map $R$: \cref{eq:naive-QL} for a simpler version and \cref{eq:optimal-QL} for a sample-optimal version.
            \item Compute the estimator $o_{z}$, which depends on the bit-string $z$ [\cref{eq:ozdef}], recovery map $R$ and the choice of observable $O$.
        \end{enumerate}
        \textbf{Data processing:}
        \begin{enumerate}
        \setcounter{enumi}{7}
            \item Compute the sample average of the estimator $\frac{1}{m} \sum_j o_{z_j}$,
            which returns an unbiased estimate of the expected value of the observable:
            $\Tr(\rho O) \approx \frac{1}{m} \sum_j o_{z_j}$.
        \end{enumerate}\\
        \bottomrule
    \end{tabular}
    
\end{table}

\section{Protocol and Mathematical Framework}\label{sec:framework}

Having explained the key physical principles at play, we now present our protocol explicitly (\cref{tab:protocol}) and in the remainder of the section we set up the  mathematical framework to describe it, in anticipation of a detailed analysis to be performed in \cref{sec:analysis}, which is a technical, fleshed-out version of \cref{sec:overview}.
We note that readers who are more interested in first seeing our protocol in practice can skip \cref{sec:analysis} and proceed to the examples in \cref{sec:rydberg,sec:bcs,sec:bose-hubbard} before returning.

 \cref{tab:protocol} presents our concrete protocol. We see that there are three key steps: (i) ``expansion'' of the state space via quench evolution of a global system, (ii) measurement, and (iii) recovery of observables via such information and the classical computation of the scrambling map $S$ and its inverse $R$.

Mathematically, the expansion step is modeled as such:
 the system, described by a density matrix $\rho$ of dimension $d_\text{sys}$, and an ancillary system, of dimension $d_\text{anc} \geq d_\text{sys}$ and prepared in a known fiducial state $|\phi\rangle$, interact via a coupling unitary. This unitary is realized by quench evolution for some time $t$ under the native Hamiltonian $H$ of the experimental platform: $U_t = \exp(-i Ht)$, which is also assumed well-known. Therefore, prior to measurement, the state of the extended system---which has dimension $d_\text{ext} = d_\text{sys} d_\text{anc}$---is:
\begin{align} 
	\rho_\ext(t) = U_t \left(\rho\otimes \ketbra{\phi_\ancs} \right) U^\dag_t. 
\end{align}
Measurements of the extended system in the standard basis $|z\rangle$ sample from $\rho_{\text{ext}}(t)$. 
The measurement outcomes $z$ are typically configurations such as a bit-string (for spin-1/2s) or a real-space particle configuration (for itinerant particles). Generalizations to qudits or other configurations, e.g. spin-resolved Fock-space basis states, are straightforward.
Each experimental run $j$ gives an outcome $z_j$ sampled from the probability distribution
\begin{align} 
	P_z = \bra{z}U_t \left(\rho\otimes \ketbra{\phi_\ancs} \right) U^\dag_t\ket{z}.\label{eq:Pz_original}
\end{align}
By repeating the experiment $m$ times, one obtains $m$ snapshots $\mathcal{S}_m = \{z_j\}_{j=1}^m$.

Recovery of information is performed via a classically-computed recovery map $R$, derived from the scrambling map $S$, and subsequent processing of the measurement data. 
To formally define $S$, consider first collating the probability distribution into a vector $|P\rangle \in \mathbb{R}^{d_\text{ext}}$ such that $P_z = \langle z|P\rangle$. Then we can rewrite \cref{eq:Pz_original} as
\begin{align} 
    \ket{P} = S \oket{\rho} \label{eq:Pz-Q}~,
\end{align}
where $\oket{\rho} \in \mathbb{C}^{d^2_\sys}$ is the vectorized version of the density matrix~$\rho$~[\cref{eq:S-high-level}].
One sees that the scrambling map $S$ (\cref{fig:Pz-diagram}) has a representation as a rectangular  matrix of size $d_\ext \times d_\sys^2$, with entries
\begin{equation}
S_{z,(k,l)} 
= \bra{z}U_t\left(\ketbra{k}{l} \otimes \ketbra{\phi}{\phi} \right) U_t^\dagger \ket{z},
\label{eq:linear_map_Q}
\end{equation}
where $|k\rangle, |l\rangle$ constitute vectors from the orthonormal basis of the system that $|\rho)$ is written in. 

Because $d_\ext \geq d_\sys^2$, $S$ can possibly have a left-inverse (note this is not guaranteed, though we will argue it is generically so in \cref{sec:recoverability}), denoted by $R$, the recovery map. It satisfies $RS = \mathbb{I}_{d_\text{sys}^2}$, so that in particular,
\begin{align}
R|P\rangle = RS \oket{\rho} = \oket{\rho}.\label{eq:RPztorho}
\end{align}
Because of the non-squareness of the matrix $S$, the left-inverse $R$ is not unique: one choice is the so-called Moore-Penrose pseudo-inverse, given by
\begin{align}
R_\text{MP} =   \big(S^\dagger S \big)^{-1} S^\dagger.
\label{eq:naive-QL}
\end{align}
While this is a natural and often practical choice, surprisingly, this inverse is not optimal in terms of sample complexity. In \cref{sec:frame} we discuss the optimal recovery map which attains the lowest sample complexity. We will see that the optimal recovery map will be particularly useful when there is prior knowledge of the probability distribution $P_z$.

Averaging over all realizations of $m$ samples $\mathcal{S}_m = \{ z_j \}_{j=1}^m$ drawn from $P_z$, the $m$-sample reconstruction 
\begin{align}
    \rho^{(m)} = \frac{1}{m} \sum_{j=1}^m R|z_j\rangle
\end{align} is an unbiased estimator of $\rho$, i.e. averaging over all possible $m$-sample sets $\mathcal{S}_m$, $\mathbb{E}_{\mathcal{S}_m}[\rho^{(m)}] = \rho$. 
However, since $\rho$ has $d_\sys^2$ entries, the random fluctuations $\abs{\rho-\rho^{(m)}}$ will be large: any tomography scheme requires $\O{d_\sys^2/\epsilon^2}$ measurements to reconstruct a state $\rho$ up to precision~$\epsilon$~\cite{m:haahSampleOptimalTomographyQuantum2015}. 

Instead, as mentioned, the expectation value $\langle O \rangle = \Tr(O\rho)$ of an observable $O$ may be directly estimated without the full reconstruction of $\rho$. Here, we assume $O$ may or may not be Hermitian.
We can write
\begin{align}
    \Tr(O \rho) = (O^\dagger|\rho)
    = \sum_z \underbrace{(O^\dagger|R |z\rangle }_{\equiv o_z}\underbrace{\langle z|S\oket{\rho}}_{= P_z} ,\label{eq:ozdef}
\end{align}
where we have inserted the identity superoperator $\mathbb I = R S =  \sum_z R\ketbra{z} S$, and $(A|B)$ denotes the Hilbert-Schmidt inner product $(A|B) \equiv \Tr(A^\dagger B)$.
\cref{eq:ozdef} showcases that $\{o_z\}$ is a single-shot, unbiased estimator for the expectation value $\langle O\rangle$. That is, given $m$ snapshots $z_1,z_2,\dots,z_m$, we can use the mean of $\{o_{z_j}\}_{j=1}^m$ to estimate $\langle O\rangle$:\begin{align} 
	\bar o_{(m)} \equiv \frac{1}{m} \sum_{j = 1}^m o_{z_j} \overset{m \rightarrow \infty}{\longrightarrow} \sum_z P_z o_z = \Tr(O \rho).
\end{align}
Here we introduce the bar notation $\bar{f}_{(m)}$ to indicate the sample averaging $\frac{1}{m}\sum_{j=1}^m f(z_j)$ for a particular $m$ snapshots $\mathcal{S}_m$.
In the absence of noise, $\bar o_{(m)}$, on average over $m$-sample sets, equals $\Tr(O\rho)$.
The relevant figure of merit for our protocol is then the number of samples required to estimate $\Tr(O \rho)$ up to a certain precision. (Additional systematic errors in $\bar o_{(m)}$ may arise in the presence of noise; we study them in \cref{sec:noise}.)
Given different $m$-sample sets, the estimator $\bar{o}_{(m)}$ fluctuates around the average value $\Tr(O\rho)$. The magnitude of such fluctuations is given by the variance of $o_z$:
\begin{align} 
    &\left(\Delta \bar o_{(m)}\right)^2 = \Var [\bar o_{(m)}] = \frac{\Var[o_z]}{m}, \label{eq:sample_complexity}\\ 
	&\text{where }\Var[o_z] = \sum_z P_z \big\vert o_z {\big\vert}^2  - \Big\vert \sum_{z} P_z o_z \Big\vert^2. \label{eq:var_oz}
\end{align}
The quantity $\Var[o_z]$ quantitatively captures our previously-introduced notion of \textit{sample complexity} associated   with our protocol in estimating $O$.
Note that $\Var[o_z]$ implicitly depends on the choice of recovery map $R$, hence one aims to minimize $\Var[o_z]$ by carefully designing $R$.
Chebyshev's inequality allows us to bound how much the estimator $\bar o_{(m)}$ deviates from its average value $\Tr(O \rho)$. 
For example, for any $\epsilon>0$, the probability $\text{Pr}\left[\abs{\bar o_{(m)}-\Tr(O \rho)}> \epsilon \right]$ is less than 10\% as long as $m\geq 10 \Var[o_z]/\epsilon^2$.

Finally, we can generalize \cref{eq:ozdef} to extract nonlinear observables that are supported on $k\in \mathbb N$ copies of $\rho$:
\begin{align} 
    \Tr(O \rho^{\otimes k}) 
    = \sum_{z_1,\dots,z_k} \underbrace{(O^\dagger|R^{\otimes k} |z_1,\dots,z_k\rangle }_{\equiv o_{z_1,\dots,z_k}}P_{z_1}\dots P_{z_k},\label{eq:ozdef_k}
\end{align}
where $z_1,\dots,z_k$ are independent samples from the same distribution $P_z$ defined in \cref{eq:ozdef}.
As an example, the SWAP operator is a non-linear operator on two copies of $\rho$, and is related to the R\'enyi-2 entropy of $\rho$. 
Given a finite set of $m$ samples $\{z_1,\dots,z_m\}$, the so-called $U$-statistics offers a sample-efficient estimate of $\Tr(O\rho^{\otimes k})$~\cite{m:hoeffdingClassStatisticsAsymptotically1948,m:huangPredictingManyProperties2020}:
\begin{align} 
  \Tr(O\rho^{\otimes k}) \approx \binom{m}{k}^{-1} \sum_{1\leq i_1<\dots<i_k\leq m} o_{z_{i_1},\dots,z_{i_k}}. \label{eq:U-statistics-estimator}
\end{align}

\section{Analysis of protocol}\label{sec:analysis}
We now analyze the performance of our protocol in depth. In \cref{sec:recoverability} we discuss the conditions under which arbitrary observables of the target state can and cannot be estimated.
We discuss in \cref{sec:quench_time} the related matter of the required quench evolution time for the protocol to be tomographically complete.
In \cref{sec:quench-setups}, we explain how different quench setups affect the sample complexity and the computational complexity of our protocol.
In \cref{sec:frame}, we derive the optimal classical post-processing protocol that minimizes statistical fluctuations.
Finally, we discuss the performance of our protocol in the presence of noise in \cref{sec:noise}.

\subsection{Recoverability: Symmetry constraints}\label{sec:recoverability}

Tomographic completeness---the ability to recover arbitrary physical information---of our protocol requires the \scrambler~$S$ to be invertible. 
A pertinent question is therefore whether this is the case when the scrambling map is generated by quench evolution $\exp(-iHt)$ under many-body Hamiltonians native to the experimental platform.  
Indeed, we argue that tomographic completeness generically holds if the Hamiltonian is sufficiently ergodic, that is, as long as information initially localized on system degrees of freedom scrambles into ancillary degrees of freedom.

Before presenting a detailed analysis, let us first present an intuitive understanding of tomographic completeness in terms of \textit{operator scrambling}.
To begin,   consider 
two distinct quantum states $\rho$ and $\sigma = \rho + \delta \rho$, where $\delta \rho \neq 0$ is some traceless operator. We ask when they can be distinguished by standard-basis measurements following quench dynamics. A positive answer to this question is signaled by a non-zero difference in the measurement outcome probabilities $\delta P_z$, for some $z$.
Tomographic completeness is then equivalent to {\it every} pair of states being distinguishable, i.e., for any arbitrary difference $\delta \rho$. 
Equivalently, we can consider the dynamics of an operator $\delta \rho \otimes |\phi\rangle \langle \phi| $ on the global system, under the quench evolution
\begin{align}
    \delta \rho(t) \equiv  e^{-iHt} \left( \delta\rho \otimes \ketbra{\phi}{\phi} \right) e^{iHt}.
\end{align}
For a qubit system, if the dynamics is scrambling, over time this becomes generically a complicated linear combination of many Pauli string operators, i.e.
\begin{align}
\delta\rho(t) = \sum_\mu c_\mu (t) \sigma_\mu,    
\end{align}
 where $\mu$ enumerates over $4^{n_\text{sys} + n_\text{anc}} $ Pauli string operators $\sigma_\mu$, e.g., $\sigma^x \otimes \sigma^y \otimes \mathbb{I} \otimes \cdots $ and their corresponding coefficients $c_\mu (t)$. 
 In this formulation, distinguishability of $\rho$ and $\sigma$ (non-zero $\delta P_z$ for some $z$) is possible if the coefficients $c_\mu(t)$ are nonvanishing for some diagonal Pauli string operators, e.g., $\mathbb{I} \otimes \sigma^z \otimes \sigma^z \otimes \cdots$.
 Then, the condition for tomographic completeness is that the time-evolved operator $\delta \rho(t)$ has overlap with \emph{some} diagonal Pauli string  for any $\delta \rho$. 
 Now, consider the structure of these coefficients $c_\mu(t)$. 
 Barring any special circumstances (e.g., symmetries, or dynamical localization etc., discussed below), we expect from numerous previous studies on operator spreading~\cite{m:nahumOperatorSpreadingRandom2018,m:vonkeyserlingkOperatorHydrodynamicsOTOCs2018,m:khemaniOperatorSpreadingEmergence2018} that under  ergodic dynamics, a given operator $\delta \rho \otimes |\phi\rangle \langle \phi|$ generically spreads within and in fact fills in its light cone, thus making it very unlikely that   $\delta \rho(t)$ completely avoids spreading to any diagonal Pauli string in dynamics. Indeed,   
 tomographic {\it incompleteness} amounts to the presence of an operator $\delta \rho$ with vanishing coefficients $c_\mu(t)$ for {\it all}  diagonal Pauli strings. This is arguably a very unlikely scenario as it requires fine-tuning a linear combination of operators; a moment's thought shows that this problem can be cast as a set of $2^{n_\text{sys} + n_\text{anc}}$   simultaneous linear equations with $4^{n_\text{sys}}$ unknowns, which is highly over-constrained if $n_\text{anc} > n_\text{sys}$.
 That is to say, when the ancillary system is large enough, we can generically expect full recoverability of information if quantum dynamics is ergodic.

 The above discussion regarding tomographic completeness can be succinctly captured by a simple statement: it is that
\begin{align}
\label{eq:operatorOverlap}
    \sum_{\mu \in \textrm{diag} } 
    \abs{ \Tr(\delta \rho(t) \sigma_\mu)}^2 \neq 0
\end{align}
for all traceless linear operators $\delta \rho$ supported in the system degrees of freedom.
Interestingly, the left hand side can be  re-expressed as a sum of out-of-time-ordered correlators (OTOC) $ \sum_z 4 \Tr( \pi_z \delta \rho(t) \pi_z \delta \rho(t)) = \sum_z 4 \abs{\bra{z} \delta \rho(t) \ket{z}}^2$ with the projection operator $\pi_z = \ketbra{z}{z}$, so that tomographic completeness is equivalent to these particular OTOCs never vanishing for any $\delta \rho$.

While the above picture of operator scrambling in ergodic quantum dynamics is appealing and explains why our protocol should be expected to work in general, it is desirable to place it on firmer, rigorous footing.
However, proving \cref{eq:operatorOverlap} for a single instance of an arbitrary ergodic Hamiltonian dynamics is difficult, if not impossible.
Nevertheless, we are able to make progress and establish 
a rigorous result on the tomographic completeness of a slightly modified version of protocol, in which the evolution time $t$ is not fixed, but randomly chosen.
Our proof relies on two widely accepted assumptions:  
(i) that the Hamiltonian satisfies  the \textit{second no-resonance condition} (see below for its definition) and 
(ii) that the measurement basis  is distributed across all eigenstates of the global Hamiltonian.
Both assumptions concern the ergodicity of the Hamiltonian dynamics defined through its eigenvalues and eigenvectors, respectively.
Furthermore, by inspecting when the second condition is violated, we identify a failure-mode: when the Hamiltonian displays symmetries that restrict information scrambling. We provide ways to overcome this limitation by using different geometric arrangements of ancillae.

We begin by reiterating in a slightly different form  the conditions for which our protocol is tomographically incomplete for the scrambling map $S_t$, associated with the quench dynamics of duration $t$:
there exists two states $\rho^{(t)} \neq \sigma^{(t)}$ that give the same probability distribution $P_z$:
\begin{align}
    \forall z,~ \delta P_z(t) =  \bra{z} U_t \left(\delta \rho^{(t)} \otimes \ketbra{\phi_\ancs}{\phi_\ancs}   \right) U_t^\dagger \ket{z} = 0\quad \label{eq:noninvertibility_def}
\end{align}
where $\delta\rho^{(t)} \equiv \rho^{(t)}-\sigma^{(t)}$. Physically, \cref{eq:noninvertibility_def} states that the measurement data does not contain any information about the perturbation $\delta \rho^{(t)}$.
Mathematically, it states that the~\scrambler~$S_t$ has a non-vanishing null-space and hence is non-invertible.

In order to make headway, we consider a slightly more restrictive scenario: we assume that there exists a pair of density matrices $\rho, \sigma$ that are indistinguishable for \textit{all times} $t$: 
\begin{align}
    \forall t, z,~ \delta P_z(t) =  \bra{z} U_t \left(\delta \rho \otimes \ketbra{\phi_\ancs}{\phi_\ancs}   \right) U_t^\dagger \ket{z} = 0. \quad \label{eq:assumption1}
\end{align}
That is to say, for all times $t$, the scrambling map $S_t$ is non-invertible with common kernel that contains $\delta \rho$. In particular, this implies that the time-averaged map $\mathbb{E}_t[S_t]$ is also non-invertible. We now invoke the \textit{second no-resonance condition} (i), which states that the energy eigenvalues of $H$ obey
\begin{gather}
    E_i+E_j-E_k-E_l = 0 \text{ if and only if} \\
    i=k, j=l \text{ or }i=l, j=k~.\nonumber
\end{gather}
This can be viewed as a generalization of the no degeneracy condition, which states that $E_i = E_j$ if and only if $i=j$. Here, we have a  no-degenerate {\it gap} condition, which requires the gap $E_i-E_j$ between any pair of eigenvalues to be unique. This condition is a common assumption in literature on many-body thermalization and is considered a mild one~\cite{m:goldsteinDistributionWaveFunction2006,m:reimannFoundationStatisticalMechanics2008,m:lindenQuantumMechanicalEvolution2009,m:kanekoCharacterizingComplexityManybody2020,m:huangExtensiveEntropyUnitary2021}. Note this condition is notably violated in non-interacting systems. In this sense, the second no-resonance condition captures the (spectral) notion of ergodicity.

In \cref{app:no_resonance}, we demonstrate that applying the second no-resonance condition to \cref{eq:assumption1} gives two equations:
\begin{align}
 &\forall z~,~\sum_E \abs{\braket{z}{E}}^2  \bra{E}\left(\delta \rho \otimes \ketbra{\phi_\ancs}{\phi_\ancs} \right)\ket{E} = 0\quad \text{and } \label{eq:invertibility_secondary_eq1} \\
 &  \sum_{E\neq E'} \abs{\braket{z}{E}}^2 \abs{\braket{z}{E'}}^2\left|\bra{E}\left(\delta \rho \otimes \ketbra{\phi_\ancs}{\phi_\ancs} \right)\ket{E'}\right|^2 = 0  \label{eq:invertibility_main_eq1}
\end{align}
where $\ket{E}$ are the eigenstates of $H$. Since each term in the sum of \cref{eq:invertibility_main_eq1} is non-negative, they must all be zero in order to satisfy \cref{eq:invertibility_main_eq1}.

Now we invoke the second condition (ii) regarding the ergodicity of eigenvectors of $H$: $\braket{z}{E} \neq 0$ for every $\ket{z},\ket{E}$.
Under this assumption, \cref{eq:invertibility_main_eq1} implies that $\delta \rho = 0$. In short, if the measurement basis is distributed across all eigenstates and the Hamiltonian satisfies the second no-resonance condition, there is no pair of states that our protocol cannot resolve for all quench times. 

In other words, the ergodicity assumptions (i) and (ii) imply the tomographic completeness for a slightly modified version of our protocol. Instead of evolving the extended system for a fixed time, we consider evolutions with many different, long times $t$. In such cases, the measurement data contains temporal labels $\ket{z,t}$. We may consider a larger \scrambler~$\mathcal{S}$ with elements $\mathcal{S}_{zt,ij} \equiv S^{(t)}_{z,ij}$. Our results show that this \textit{temporally-enhanced}~\scrambler~is tomographically complete under conditions (i) and (ii). We expect that this requirement for evolution over all times is a technical limitation of our proof and is not necessary in practice.
Indeed, we will find that tomographic completeness holds true for generic quench times $t$ in all numerical examples studied in \cref{sec:rydberg,sec:bcs,sec:bose-hubbard}, where we quench the extended system under Hamiltonians native to the quantum simulator platforms.

Now we turn to a failure case, when  \cref{eq:invertibility_main_eq1} is nontrivially satisfied (that is, for $\delta \rho \neq 0$).
To identify such scenarios, we form, for every $\ket{z}$, a subspace of eigenstates that overlap with $\ket{z}$: $\{\ket{E}:\braket{z}{E} \neq 0\}$. We denote the projector onto this space as $\Pi_z \coloneqq \sum_{E:\braket{z}{E} \neq 0} \ketbra{E}{E}$. \Cref{eq:invertibility_secondary_eq1,eq:invertibility_main_eq1} imply that  $\left|\bra{E}\left(\ketbra{\phi} \otimes \delta \rho \right)\ket{E'}\right|^2 = 0$ for all $\ket{E},\ket{E'}$ in this subspace. 
Summing over all such $\ket{E}$ and $\ket{E'}$, we obtain 
\begin{equation}
\Tr[\Pi_z\left(\delta \rho \otimes \ketbra{\phi_\ancs} \right) \Pi_z\left(\delta \rho \otimes\ketbra{\phi_\ancs}\right)] =0.    \label{eq:trPizPiz1}
\end{equation}
Since $\left(\Pi_z\left(\delta \rho \otimes\ketbra{\phi_\ancs} \right) \Pi_z\right)^2$ is positive semi-definite, \cref{eq:trPizPiz1} implies
\begin{equation}
\Pi_z\left(\delta \rho \otimes\ketbra{\phi_\ancs} \right) \Pi_z = 0.\label{eq:PzdeltarhoPz1}
\end{equation}
Therefore, the difference in density matrices $\delta \rho \otimes \ketbra{\phi_\ancs}$, which is a linear operator in the extended space, takes states in $\Pi_z$ out of this subspace.

A prominent example where this can occur is when the quench Hamiltonian exhibits a symmetry (e.g., if there is particle or charge conservation), and the readout states $\ket{z}$, as well as $|\phi\rangle$ have well-defined quantum numbers of this symmetry. Then, the product $\abs{\braket{z}{E}}^2 \abs{\braket{z}{E'}}^2$ is non-zero only when $z, E$, and $E'$ have the same quantum numbers: $\Pi_z$ is correspondingly a projector acting only within the symmetry sector defined by $z$. Therefore, any observable on  system $\delta \rho$ that is block off-diagonal between symmetry sectors will satisfy \cref{eq:PzdeltarhoPz1}, and cannot be detected by our quench protocol. 

This condition naturally arises in itinerant particles hopping in optical lattices, which has a $U(1)$ charge associated with particle number conservation. Consider for example the patched quench configuration discussed above in \cref{sec:quench-setups} and illustrated in Fig.~\ref{fig:schemes}(b), where particles in well-separated patches undergo separate evolution. In this case, there is in fact a higher symmetry in such quench evolution: the number of particles in each patch is individually conserved, a $U(1)\times U(1)$ symmetry. 

Assuming that 
measurements collapse the extended system to states that possess a well-defined $U(1)\times U(1)$ charge (i.e. number of particles on each patch), such a patched quench setup cannot distinguish between states that are coherent superpositions of configurations in different symmetry sectors. In particular, they cannot measure observables that break this $U(1)\times U(1)$ symmetry, 
such as $c^\dagger_{i} c_{j} + \text{h.c.}$, where $c^\dagger$ ($c$) are the raising and annihilating operators for particles, located on sites $i,j$ belonging to different patches. 
To remedy this, we can imagine modifying our quench evolution in such a way that allows particles in these two patches to tunnel between each other
[``bridging the patches," illustrated in \cref{fig:schemes}(c)]. This breaks the $U(1)\times U(1)$ symmetry and allows the extraction of the pairing order parameter.

\subsection{Required quench evolution time: Constraints from locality} \label{sec:quench_time}

In the previous section, we have argued that in the absence of symmetry constraints and assuming certain natural notions of ergodicity of the quench Hamiltonian, our protocol is tomographically complete at all times.
However, intuitively, we also expect that our protocol does not work very well if the quench time is too small: information in the system does not have enough time to scramble into the ancillae. 
A natural question is then how long the quench evolution should be, in practice.

Here, we argue that constraints from the geometric propagation of quantum information---called Lieb-Robinson bounds~\cite{m:liebFiniteGroupVelocity2004,m:tranLiebRobinsonLightCone2021,m:yinFiniteSpeedQuantum2022,m:kuwaharaOptimalLightCone2022}---determine this time.
Intuitively, any information initially localized on the system must be able to ``flow'' into the ancillary system in order for it to be accessible and thus recoverable. 
Concretely, if the quench Hamiltonian is geometrically local, the Lieb-Robinson bound constrains the propagation of information to be  within a light cone. 
If this light cone is linear, we  show below that there is a threshold time $t_*$ depending on the furthest distance of a system site to an ancillary site, such that our protocol is tomographically incomplete for all $t < t_*$.
In reality, for local Hamiltonian dynamics, for any $t > 0$, there will generally be an exponentially small leakage of information  outside the light cone. 
Therefore, strictly speaking, our protocol would be tomographically complete for all $t > 0$; however, for times $0 < t < t_*$ we expect an exponentially large sample complexity. 

To best illustrate this argument,  we consider a suboptimal setup that has a large distance between a system particle and its nearest ancilla. 
This example demonstrates why it is desirable to minimize this distance, in order to minimize the quench time required for the protocol. 
Specifically, we assume that the system and the ancillae are contiguous regions that live on a one-dimensional lattice of spins, separated by a single boundary (\cref{fig:LR-proof}),  and we evolve the extended system under a nearest-neighbor Hamiltonian:
\begin{align} 
 	H = \sum_{\avg{i,j}} h_{ij}, \label{eq:local-H-def}
 \end{align} 
where the sum is over nearest-neighboring sites, for time $t$.
For convenience, we label the system and ancilla sites by $1,\dots,n_\sys,$ and $n_\sys+1,\dots,n_\sys+n_\anc$, respectively, and divide the system into three complementary sets: $A = \{1\}, B = \{2,\dots,n_\sys\}$, and $C = \{n_\sys+1,\dots,n_\sys+n_\anc\}$~(\cref{fig:LR-proof}).
With $r = n_\sys-1$, Ref.~\cite{m:haahQuantumAlgorithmSimulating2021} used the Lieb-Robinson bound~\cite{m:liebFiniteGroupVelocity2004} to show that
\begin{align} 
	\norm{e^{-iHt} - e^{-iH_{{BC}}t}e^{iH_Bt}e^{-iH_{AB} t}}\lesssim e^{v_\text{LR} t-r}, \label{eq:HHKL}
\end{align}
where $AB\equiv A\cup B$, $H_{AB}$ denotes the Hamiltonian constructed from the terms $h_{i,j}$ in \cref{eq:local-H-def} that are supported entirely in $AB$, $v_{\text{LR}}$ is the Lieb-Robinson velocity, and
 $\| \cdot \|$ is the  operator spectral norm.
Given that $H = H_{BC} - H_B + H_{AB}$, the decomposition in \cref{eq:HHKL} can be viewed as a first-order trotterization of $e^{-iHt}$, albeit with the approximation error decaying exponentially with the size of $B$. 

Let $S$ and $\tilde S$ be the scrambling maps generated by the quenches $e^{-iHt}$ and $e^{-iH_{{BC}}t}e^{iH_Bt}e^{-iH_{AB} t}$, respectively. 
It follows from the definition in \cref{eq:linear_map_Q} and \cref{eq:HHKL} that
\begin{align} 
	 \norm{S - \tilde S} \lesssim e^{v_{\text{LR}}t-r}.\label{eq:S-diff}
\end{align}
We claim that for short times $t \leq r/v_{\text{LR}}$, the scrambling map $\tilde S$ is singular, making the corresponding protocol tomographically incomplete and, by \cref{eq:S-diff}, $S$ must be nearly singular too.

\begin{figure}[t]
\includegraphics[width=0.45\textwidth]{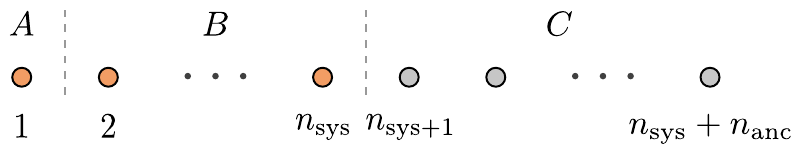}
\caption{An illustration of the decomposition of the extended system into three regions $A,B$, and $C$ as considered in \cref{sec:quench_time}.}
\label{fig:LR-proof}
\end{figure}

To prove this claim, recall that tomographic incompleteness is equivalent to the existence of an operator $\delta \rho$ such that 
\begin{equation}
    \forall z,~\bra{z}U_t(\delta\rho \otimes \ketbra{\phi}{\phi}) U_t^\dagger \ket{z} = 0~. \label{eq:deltarho2}
\end{equation}
This can equivalently be cast as 
\begin{equation}
    \Tr\left[U_t(\delta \rho \otimes \ketbra{\phi}{\phi}) U_t^\dagger Z_\mu \right] = 0~\label{eq:Lieb_Robinson_equation} 
\end{equation}
for all Pauli-strings $Z_\mu \in \{Z,\mathbb I_2\}^{\otimes (n_{\sys}+n_{\anc})}$ consisting only of $Z$ and $\mathbb I_2$, where $\mathbb{I}_2$ is the single-site ($2 \times 2$) identity operator.
Consider $U_t = e^{-iH_{{BC}}t}e^{iH_Bt}e^{-iH_{AB} t}$ and choose an operator $O$ such that $\delta \rho\propto O \equiv e^{iH_{AB} t} X^{(1)} e^{-iH_{AB} t}$, where $X^{(1)}$ is the Pauli-$X$ operator acting on site $1$ (region $A$).
We then have
\begin{align} 
	 &e^{-iH_{{BC}}t}e^{iH_Bt}e^{-iH_{AB} t}(O \otimes \ketbra{\phi}{\phi}) e^{iH_{{AB}}t}e^{-iH_Bt}e^{iH_{BC} t} \nonumber\\
	 &= X^{(1)}\otimes e^{-iH_{{BC}}t} (\mathbb I_{B}\otimes \ketbra{\phi}{\phi} )e^{iH_{{BC}}t},
\end{align}
which is trace-orthogonal to all Pauli-strings $Z_\mu$.
Therefore, \cref{eq:deltarho2} holds, implying that  
$\tilde S$ is singular and is tomographically incomplete as a~\scrambler. 

By \cref{eq:S-diff}, $S$ has at least one exponentially small singular value when $t < t_* = r/v_{\text{LR}}$.
It is natural to expect this to lead to an exponentially large sample complexity associated with the recovery map $R$.
In \cref{app:lieb_robinson}, we provide numerical simulations supporting this argument. 
Our numerical evidence also suggests that the Lieb-Robinson bound not only gives a lower bound for the requisite quench time, but also describes the optimal quench time: beyond this Lieb-Robinson time $t_*$, the sample complexity quickly plateaus and subsequent quench evolution brings little improvement.
We note that it is possible to generalize this lower bound on the quench time to $D$-dimensional lattices with power-law decaying Hamiltonians~\cite{m:tranLocalityDigitalQuantum2019} and to bosonic systems at finite particle density~\cite{m:kuwaharaOptimalLightCone2022}.

Finally, we note that, in practice, it is often easy to circumvent the constraints from the Lieb-Robinson light cone. 
For example, instead of the setup in \cref{fig:LR-proof} where the system and ancillae are connected only through a small bottleneck through which information has to flow, if we were to arrange the ancillae using the global quench setup with high connectivity like in \cref{fig:schemes}(a), the distance between a system site and its nearest ancilla would be independent of the linear length of the system.
Therefore, even with no leakage outside the light cone, a  quench time that is independent of system size is sufficient to ensure tomographic completeness of our protocol.

\subsection{Quench setups and sample-complexity} \label{sec:quench-setups}

As argued above, tomographic completeness is   relatively easy to satisfy, and as discussed in \cref{subsec:sample_complexity}, as our aim is not to fully tomographically reconstruct a many-body state but rather learn interesting physical properties of it, a more important figure of merit is the {\it sample complexity} of our protocol, i.e., the required number of measurements to well estimate    a target observable.
It turns out that the interactivity between the system and ancillae in relation to the choice of  observable, play a key role in determining its sample complexity, leading to different ways to implement our protocol that affect its performance.

To explain this, let us quickly recap the recently-introduced and related quantum state-learning protocol known as classical shadow tomography~\cite{m:ohligerEfficientFeasibleState2013,m:huangPredictingManyProperties2020,m:elbenRandomizedMeasurementToolbox2022}, which provides important insights into the design of our protocol.
The main idea of classical shadow tomography is to apply randomized measurements, realized by random unitary evolution from ensembles with known statistical properties. Information can be recovered of the system by post-processing the measurement outcomes in a manner similar to ours. 
In Ref.~\cite{m:huangPredictingManyProperties2020}, it was established that different random ensembles of unitaries are well suited to estimate different classes of observables. Specifically, low rank observables (which are necessarily non-local, i.e., they do not act on a small region in space) can be efficiently estimated through applying random unitaries supported on the full system. A concrete example is that of deep, random Clifford circuits, which mimic the behavior of global Haar-random unitaries. In contrast, few-body observables can be efficiently estimated with random spatially-local unitaries, concretely realized by products of random, on-site Clifford rotations. In the intermediate regime, it was argued that observables that are neither few-body nor low-rank require a large number of samples in either quench setup. 
Therefore, depending on the observables of interest one has to utilize different ensembles of unitary circuits to minimize sample complexity~\cite{m:huClassicalShadowTomography2021}.
We provide a more comprehensive review in \cref{app:classical_shadow_tomography}.

The above results find natural analogs in our setting of quench evolution with ancillae. We discuss several quench setups and the observables that they are well suited to estimate:

\emph{Global quench.}~[\cref{fig:schemes}(a)]---Here we couple the entire system of interest with a common set of ancillae, and quench evolve the joint system.
Intuitively, since there is high interactivity between the system and ancillae, this set-up mimics the behavior of scrambling of information from global Haar-random unitaries in classical shadow tomography.
Thus, we expect this configuration is well suited for estimating low-rank observables such as the many-body fidelity, though note it may also be used to estimate arbitrary non-local observables. However, a drawback is that it is  generally computationally costly to numerically compute the global scrambling and recovery maps, and hence cannot be applied to large systems, hindering scalability of this approach.
    
\emph{Patched quench.}~[\cref{fig:schemes}(b)]---Here, we divide the system into multiple disjoint subsystems and couple each subsystem with its own set of ancillae, before quench evolving them individually. 
Intuitively, since the interactivity between system and ancillae is limited to within local patches in space, this is akin to scrambling of information via random local unitaries in classical shadow tomography. Thus 
this configuration is expected to be well suited for estimating few-body observables: the subsystem size can be tuned to match the support size of the observable and thus minimize the sample complexity. Note this approach also has a low classical computational cost as this only depends  on the largest patch size considered, and is thus scalable. In particular, it can even be practically favorable to employ a patched quench to estimate observables which are global in nature, in order to overcome the computational overhead as described above in using a global quench. 
    
\emph{Bridged quench.}~[\cref{fig:schemes}(c)]---In certain analog quantum simulators, dynamics might be constrained by symmetries, preventing recoverability of information in certain quench setups. An example, as mentioned before in \cref{sec:recoverability}, is furnished by a system of itinerant particles hopping in optical lattices which conserves the total particle-number.
In particular, if 
we quench evolve two {\it disjoint} patches of itinerant particles on an optical lattice, the particle number in each patch is conserved, leading to an enhanced $U(1)\times U(1)$ symmetry. If we measure in the particle number basis, our protocol will not be able to detect observables that do not commute with this symmetry. In \cref{sec:bcs}, we discuss an example of such an observable: a superconducting pairing correlator that annihilates a Cooper pair in one patch, while creating one in the other. 

The interactivity of the system and ancillae must thus be engineered in a way to break this enhanced symmetry. For example, we can imagine  introducing a ``bridged" quench setup [see \cref{fig:schemes}(b)], which in the case of the example of itinerant particles hopping in a optical lattice allows particles to be exchanged between separate patches. 
\cref{sec:bcs} also demonstrates how the use of such a configuration now allows for the successful estimation of the superconducting pairing order parameter.

\subsection{Optimal classical data processing}\label{sec:frame}
We now discuss the optimal classical data processing strategy for estimating a given observable.
As discussed in \cref{sec:data_processing}, in general, a given~\scrambler~$S$ does not have a unique recovery map $R$. When we have knowledge of the state of interest $\rho$, we find a recovery map which provably minimizes the sample complexity of estimating any observable $O$. The key idea is to use results from \textit{frame theory}, a mathematical theory relevant to signal processing~\cite{m:scottTightInformationallyComplete2006,Daubechiestenlectures}.

As mentioned in \cref{sec:overview}, the projective measurements on the extended system induce randomized measurements on the system. 
Formally, these randomized measurements constitute a positive operator-valued measure (POVM) $\{\ketbra{S_z}\}$, where $\sum_z \ketbra{S_z} = \mathbb I_{d_\sys}$ and  
\begin{equation}
        \ket{S_z} = \left(\mathbb{I}_{d_\sys} \otimes \bra{\phi_\ancs} \right) U_t^\dagger |z\rangle. 
\end{equation}
The POVM $\{\ketbra{S_z}{S_z}\}\equiv \{\oket{S_z}\}$
can be identified with an overcomplete basis $\{\oket{S_z}\}$ over linear operators of the system. Intuitively, this overcompleteness gives redundant information in its measurement outcomes and therefore a redundancy in ways of extracting desired quantities. It turns out that a na{\"i}ve way of processing measurement outcomes (based on the Moore-Penrose pseudoinverse) overweights outcomes $z$ that are more frequently observed; one has to correct for this overweighting in order to minimize the statistical error. 

The above intuition can be formalized by recognizing that the POVM $\{\oket{S_z}\}$ is an object known as an \textit{operator frame}. \cref{app:frame-theory} formally defines a frame and discusses its properties. In quantum information theory, frames have been studied in the context of informationally complete POVMs~\cite{m:renesSymmetricInformationallyComplete2004,m:klappeneckerMutuallyUnbiasedBases2005,m:scottTightInformationallyComplete2006}.

Every frame has so-called \textit{dual frames} that allow for their inversion as in \cref{eq:RPztorho}. Crucially, such dual frames are not unique; this corresponds to the rectangular matrix $S$ not having a unique left-inverse $R$; the Moore-Penrose pseudo-inverse
\begin{equation}
R_\text{MP} = (S^\dagger S)^{-1} S^\dagger~,    
\end{equation}
is one such left-inverse, but it is easy to check that so are matrices of the form:
\begin{equation}
    R= (S^\dagger \Gamma S)^{-1} S^\dagger \Gamma~,\label{eq:inverse_maps_Gamma}
\end{equation}
for $\Gamma$ positive-definite (and Hermitian). Each choice of $R$ corresponds to a {\it different} observable estimator $o_z = (O^\dagger|R\ket{z}$ for the {\it same} observable $O$ and therefore has   different sample complexities $\text{Var}[o_z]$.

In \cref{app:frame-theory} we show that a result from frame theory gives the left-inverse $R$ that provably minimizes the sample complexity for a given state $\rho$, independent of the observable $O$. This left-inverse is given by \cref{eq:inverse_maps_Gamma} with the choice of $\Gamma$ being a diagonal matrix with entries 
\begin{equation}
\Gamma_{z,z} = 1/P_z = \bra{S_z}\rho\ket{S_z}^{-1} \label{eq:optimal-QL}
\end{equation}
(note it is explicitly $\rho$-dependent).
Explicitly, we invert the $d_\sys^2 \times d_\sys^2$ matrix $A\equiv S^\dagger \Gamma S$, which has matrix elements:
\begin{align}
    A_{(i,j),(k,l)} &= \sum_z \frac{1}{P_z} \braket{i}{S_z}\!\braket{S_z}{j} \!\braket{k}{S_z}\!\braket{S_z}{l}~, 
\end{align}
giving the corresponding $\{o_z\}$ as
\begin{align}
    o_z &= \frac{1}{P_z} \sum_{ijkl} O_{j,i} \big(A^{-1}\big)_{(i,j),(k,l)} S_{z,(k,l)},
    \label{eq:optimal-oz}
\end{align}
which has the smallest possible $\text{Var}(o_z)$ while also satisfying $\sum_z P_z o_z = \tr(O\rho)$.

In \cref{fig:rydberg-all}(a) below we demonstrate that the optimal recovery map halves the number of required samples compared to the Moore-Penrose version. 
In \cref{app:frame-nonlinear}, we show that the same dual frame in \cref{eq:optimal-QL} is also optimal for extracting information involving nonlinear observables.

The construction of the optimal recovery map [\cref{eq:optimal-QL}] is seemingly self-referential: our aim is to recover $\rho$ (thus $P_z$ is unknown), yet it explicitly requires knowledge of $P_z$.
Indeed, we expect this approach to be useful when one uses their quantum device to prepare a known target state.
In practice, with an unknown state, we may not be able to construct the optimal frame exactly. However, our result in fact accords a way to construct a recovery map if one does have a prior model for a \textit{distribution} $p(\rho)$ of initial states: we simply replace $P_z = \bra{S_z}\rho\ket{S_z}$ in  \cref{eq:optimal-QL}   by $\overline{P}_z \equiv \int d\rho~p(\rho) \bra{S_z}\rho\ket{S_z}$ in \cref{eq:optimal-QL}. This minimizes the \textit{expected} sample complexity over the distribution $p(\rho)$. For example, if one has no knowledge of the initial state, a reasonable guess might be the uniform distribution over $\mathcal{H}_\sys$, and we have $\overline{P}_z = \bra{z}U_t(\mathbb{I}_{d_\sys} \otimes \ketbra{\phi})U_t^\dagger \ket{z}$. As long as this distribution $\overline{P}_z$ is different from the uniform distribution $P_z = 1/d_{\ext}$, the optimal recovery map will outperform the na{\"i}ve Moore-Penrose pseudoinverse [\cref{eq:naive-QL}].

As a matter of practice, we note that when we only have a few observables $O$ to estimate, it is preferable to directly solve the linear equation $\oket{O} = S^\dagger \ket{o}$ to obtain $\ket{o} \equiv \{o_z\}$. 
That is, one can perform a QR decomposition on $S^\dagger$ and solve for $|o\rangle$ through Gaussian elimination. Compared to finding the inverse of $S$, this method is numerically more stable and has faster computational runtimes. One may verify that the standard computational method of solving linear equations based on such a QR decomposition yields the same solution as the Moore-Penrose pseudoinverse of $S$, $\ket{o} = R_\text{MP}^\dagger \oket{O}$ (\cref{app:QR_decomp_linear_equations}). To obtain the estimator $\{o_z\}$ arising from the optimal left-inverse $R^\dagger \oket{O}$, one may instead apply the QR decomposition algorithm to solve the linear equation $S^\dagger \Gamma^{1/2} \left[\Gamma^{-1/2} \ket{o}\right] = \oket{O}$.

Finally, our result of the optimal recovery map construction can in fact also be applied to conventional randomized measurement schemes such as classical shadow tomography, and may be of independent interest. 

\subsection{Effect of noise}\label{sec:noise}

Thus far, we have analyzed the performance of our protocol assuming the quench evolution of the extended system is an ideal unitary and measurements are perfectly implemented. As mentioned in \cref{sec:overview}, one deleterious effect is the presence of noise, which perturbs around such limits. In this section we discuss the effects of noise during the quench evolution on our protocol.
We consider two scenarios corresponding to our knowledge of the noise process.

\emph{Using the noiseless recovery map.---}First, we consider the scenario where we cannot fully characterize the noise process, and the noise rate is sufficiently small. Specifically, let $S^{(\gamma)}$ be the noisy \scrambler~under a global noise rate $\gamma$.
Note that $\gamma$ may depend on the system size $N$, e.g. $\gamma = N\gamma_{\text{loc}}$ with a local error rate $\gamma_{\text{loc}}$.  
Since we cannot compute the left inverse of $S^{(\gamma)}$, we use the recovery map $R^{(0)}$ of the noiseless channel to recover the initial state.
Evidently, since $R^{(0)} S^{(\gamma)}\neq \mathbb I$, this approach introduces errors in the recovered state. This error is systematic and cannot be suppressed by acquiring more experimental samples.

Intuitively, if $\gamma t \sim O(1)$, we expect one or more errors to occur in every experimental run and severely distort the measurement outcome probabilities $P_z$. In this case, without knowing the error channel, it is informationally impossible to recover the initial state. Therefore, we restrict our attention to the case where $\gamma t\ll 1$ and estimate the systematic error to leading order in $\gamma t$ using reasonable assumptions about the distribution of the measurement outcomes.

Recall that without noise, the measurement outcomes $z$ are sampled from the probability distribution $P_z$. Conditioned on the presence of at least one error with probability $\gamma t$, we instead sample $z$ from a different probability distribution $P'_z$.
This gives an incorrect estimate of $\avg{O}$:
\begin{align} 
	\avg{O}' = (1-\gamma t)\sum_z P_z o_z + \gamma t\sum_z P'_z o_z + \O{\gamma^2t^2},
\end{align}
where $o_z$ is the estimator of $O$ defined in \cref{eq:ozdef}.
The systematic deviation from the correct value $\avg{O}$ is
\begin{align} 
	\Delta_O = o - o' = \gamma t \sum_z \underbrace{(P_z - P'_z)}_{\equiv \delta P_z} o_z + \O{\gamma^2 t^2}. \label{eq:sys-error-def}
\end{align}
In general, $P_z$ and $P'_z$ can be arbitrary distributions and $\Delta_O$ can be a positive or negative offset. In a crude estimate, we expect $\abs{\delta P_z}\sim 1/d_\ext$ and, therefore, we estimate the magnitude $\abs{\Delta_O}$ as:
\begin{align} 
	\abs{\Delta_O} \lesssim \frac{\gamma t}{d_\ext} \sum_z  \abs{o_z} \leq \gamma t \sqrt{\frac{1}{d_\ext}{\sum_z o_z^2}} \nonumber\\
	\approx \gamma t \sqrt{\Var(o_z)} \approx N \gamma_\text{loc}t \sqrt{\Var(o_z)}, \label{eq:back-of-envelope-bound}
\end{align}
where the second inequality is due to the Cauchy-Schwarz inequality.  
The last term $\sqrt{\text{Var}(o_z)}$ is the only term in the bound that depends on the operator $O$.
\Cref{eq:back-of-envelope-bound} has an intuitive interpretation:
When the noise rate and the sample complexity are both small, most of the collected samples would have no error, resulting in a small total systematic error.
Consequently, observables that have low sample complexity in our protocol are also robust against noise. Intuitively, the sample complexity is proportional to the distribution of values of $o_z$. If this spread is large, measuring the incorrect bitstrings $z$ will lead to a large error in the estimated $\langle O \rangle$.

The back-of-the-envelope bound in \cref{eq:back-of-envelope-bound} is a conservative
estimate that assumes the summands in \cref{eq:sys-error-def} add up coherently. 
Nevertheless, we demonstrate with an example in \cref{app:noise} that \cref{eq:back-of-envelope-bound} closely captures the behavior of the systematic error in the presence of noise. 

\emph{Using the noisy recovery map.---}Next, we consider the case when the noisy evolution is well characterized. In particular, we assume that we know exactly the noisy evolution channel $\mathcal S$. In this case, we can simply generalize the definition of the \scrambler~$S$ in \cref{sec:framework} to account for the non-unitary quench channel, i.e.
\begin{align} 
	S_{z,kl} \equiv \bra{z}\mathcal S\Big(\ketbra{k}{l}\otimes \ketbra{\phi_\ancs}{\phi_\ancs}\Big)\ket{z}, 
\end{align}
which remains a linear map and is generically invertible, and the rest of the protocol would remain the same. 
However, because the scrambling quench is now different, the sample complexity may also be different from the noiseless case.
In fact, because noise 
can only reduce the distinguishability between quantum states, we 
expect to need more samples to determine observables up to the same precision, in the presence of noise.

As an example, we assume that the scrambling quench is affected by global depolarizing noise at a constant rate $\gamma = N \gamma_\text{loc}$. We argue that the sample complexity for estimating an observable increases exponentially with $\gamma t$, where $t$ is the scrambling quench time. Indeed, under this noise model, the state of extended system gradually flows towards the maximally mixed state during the scrambling quench. At the end of the quench, the extended state can be written as
\begin{align} 
	\rho_\ext^{{(\gamma)}}(t) = e^{-\gamma t} \rho_\ext^{(0)}(t) + (1-e^{-\gamma t}) \frac{\mathbb I}{d_\ext},  \label{eq:noisy_state}
\end{align}
where $\rho_\ext^{(0)}(t)$ is the extended state after a noiseless quench evolution.
Accordingly, the probability of obtaining an outcome $z$ in the presence of noise $P^{(\gamma)}_z$ is
\begin{align} 
	 P^{(\gamma)}_z = e^{-\gamma t} P^{(0)}_z + (1-e^{-\gamma t})/d_\ext, \label{eq:Pznoisy}
\end{align}
where $P^{(0)}_z$ is the probability without noise.
As we increase the noise rate $\gamma$, the distribution of outcomes $z$ approaches the uniform distribution, which contains no information about the initial state of the system.
For this reason, we will need more samples to recover the initial state to the same precision. 

Mathematically, \cref{eq:Pznoisy} implies that we can replace the \scrambler~$S^{(0)}$ by 
\begin{align} 
	S^{(\gamma)} = \left[e^{-\gamma t} \mathbb I + (1-e^{-\gamma t}) \frac{\mathbbm{1} }{d_\ext} \right] S^{(0)}, 
\end{align}
where $\mathbbm{1}$ is a $d_\ext\times d_\ext$ matrix with entries $\mathbbm{1}_{ij,kl} = \delta_{ij} \delta_{kl}$. 
It is then easy to verify that the recovery map in the presence of noise is
\begin{align} 
	 R^{(\gamma)} =  e^{\gamma t} R^{(0)} \left[ \mathbb I - (1-e^{-\gamma t}) \frac{\mathbbm{1} }{d_\ext} \right].
\end{align}
In the limit of large $\gamma t$ and large $d_\ext$, $R^{(\gamma)} \approx  e^{\gamma t} R^{(0)}$ gains an exponential factor $e^{\gamma t}$ compared to the noiseless case.
Given an observable $O$, the estimator $o_z^{(\gamma t)} \approx e^{\gamma t} o_z^{(0)}$ also gains the same factor. Therefore, the 
sample complexity of $o_z$ increases by $e^{2\gamma t}$ due to the global depolarizing noise.

Exactly how the sample complexity increases with noise depends on the the details of the noise model. While the above discussion was valid for a simple toy model of depolarizing noise, 
in typical many-body systems and noise models, we expect that the maximally mixed state $\mathbb{I}/d_\text{ext}$ in \cref{eq:noisy_state} can simply be replaced with an equilibrium, thermal density matrix $\rho_\text{eq}$. The inversion map can then be obtained by simply subtracting $\rho_\text{eq}$ and gives the same qualitative behavior discussed above.

For example, in \cref{app:noise}, we demonstrate that the sample complexity in the presence of local dephasing with a constant error rate also increases exponentially with the scrambling quench time.

\section{Rydberg Atom Arrays: Extracting Fidelity, Energy transport, and Entanglement Entropy}\label{sec:rydberg}

\begin{figure*}[t]
\includegraphics[width=0.97\textwidth]{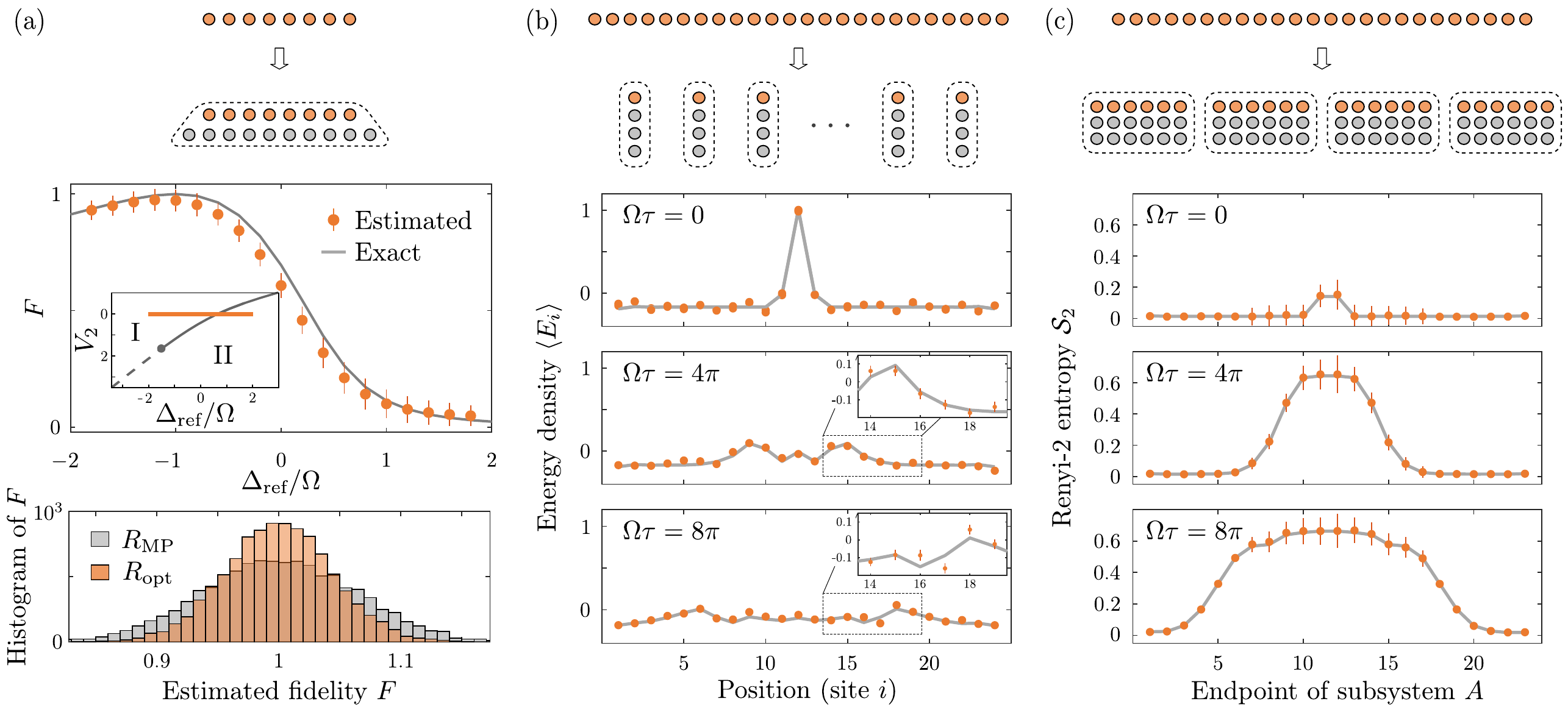}
\caption{
Extraction of various physical information in a system of interacting Rydberg atoms using our protocol.
 Rydberg atoms (orange circles) are coupled to ancillary atoms (gray circles) with different interactivities (denoted by the dashed boxes; distance not to scale): 
only atoms within the same dashed boxes may interact.
(a) We apply the global quench setup onto a state believed to be experimentally prepared in the ground state of $H_\ryd$ [\cref{eq:H_ryd}] at  parameters $(\Delta, V_2) = (-\Omega, 0)$, and extract its fidelity to several reference states: the ground states of $H_\ryd$ at $\Delta_\text{ref} \in [-2\Omega,2\Omega]$ and $V_2 = 0$. 
In the upper panel, we compare the fidelity extracted using a single set of $m = 2000$ samples (orange scatter points) and the exact values (gray solid line).
The error bar is the standard deviation of the estimated fidelity over $100$ independent sets of $m$ samples.
The inset indicates the phase diagram of $H_\ryd$ reproduced from Ref.~\cite{m:slagleQuantumSpinLiquids2022}.
Our considered range of $\Delta_\text{ref}$ (orange line) covers both the disordered (I) and $\mathbb Z_2$ ordered (II) phases. In the lower panels, we compare the distributions for extracted fidelity using the optimal recovery map $R_\text{opt}$ [\cref{eq:optimal-QL}, orange bars] and the state-agnostic Moore-Penrose version $R_{\text{MP}}$ [\cref{eq:naive-QL}, gray bars]. 
We obtain the distributions by generating $10^4$ different sets of $m=2000$ samples and estimate the fidelity from each set.
Note that the extracted fidelity need not be bounded above by one because the recovery map is generally not positive,  and may map each measurement snapshot to an unphysical state. 
The optimal recovery map $R_\text{opt}$  has a provably smaller variance, as corroborated by the numerics, and hence, a smaller sample complexity than $R_{\text{MP}}$.
(b) We extract the local energy density $\avg{E_i}$ [\cref{eq:energy_density}] after a local perturbation in the middle of the chain using a patched quench. This can be done by `fanning out' the atoms of the system into well-separated patches, and coupling each of them to their own ancillary atoms.
The extracted values (orange scatter points) from $m = 2000$ samples agree well with their exact values (gray solid line) and correctly capture the energy diffusion at all times $\tau$ after the perturbation.
(c) We extract the R\'enyi-2 entanglement entropy after the perturbation, using patched quenches with $m = 10^5$ samples. The error bar increases substantially when the subsystem $A$ overlaps with multiple patches.
}
\label{fig:rydberg-all}
\end{figure*}

The remaining part of the paper aims to numerically demonstrate our protocol for realistic, current experimental systems. We begin in this section by applying our protocol to a quantum simulator comprised of arrays of Rydberg atoms and demonstrate its basic capabilities, in particular showcasing how the different ways of coupling ancillae to the system affect its performance. 

Arrays of Rydberg atoms are a leading platform for analog quantum simulation, owing to their strong performance in many aspects including decoherence time, high fidelity quantum gates and readout, and programmability. 
In recent years, advances in analog quantum simulators based on Rydberg atoms have led to the discovery of quantum many-body scars~\cite{m:omranGenerationManipulationSchrodinger2019}, realization of various crystalline phases~\cite{m:bernienProbingManybodyDynamics2017,m:ebadiQuantumPhasesMatter2021}, and observations of signatures of topological order~\cite{m:semeghiniProbingTopologicalSpin2021}.

Here, we consider a linear array of $n_\sys$ interacting Rydberg atoms, each of which can be modeled as a two-level system.
In this example, for simplicity we assume that neighboring sites are blockaded, i.e., they are forbidden from being simultaneously excited. In the language of Rydberg atom quantum simulators, this amounts to working in the so-called ``blockade-radius'' $R_b \sim 1.1 - 1.3$. 
In that case the system can be well modeled by the Hamiltonian
\begin{align} 
	H_\ryd  = \frac{\Omega}{2} \sum_i \mathcal P X_i \mathcal P - \Delta \sum_{i} n_i
	+ V_2 \sum_{\abs{i-j}=2} n_i n_j,\label{eq:H_ryd}
\end{align}
where $\Omega$ represents the Rabi frequency of an external laser field that excites the atoms, $\Delta$ its detuning, $V_2$ the next nearest-neighbor interaction strength, $i,j$ are the positions of the atoms on the lattice, $n_i = (\mathbb I + Z_i)/2$ is the occupation number of the Rydberg state on site $i$, $X,Y,Z$ are the standard Pauli matrices, and $\mathcal P$ is the projector onto the subspace where no two atoms within their blockade radius are simultaneously excited, and we set $\Omega = 1$ for the rest of the section. 

Our objectives will be to characterize the ground states of \cref{eq:H_ryd}, and the dynamics of said states upon perturbation away from equilibrium. We focus on three properties of these states: (i) the quantum state fidelity between an experimentally prepared state against the true reference state, (ii) dynamics of local energy densities, and (iii) dynamics of entanglement entropy after a local perturbation.
These properties are key quantities for many-body physics, and respectively represent low-rank, local, and nonlinear observables. Depending on which observable we are interested in, we use different quench setups (\cref{sec:expansion}), which couple the ancillae to system atoms in different ways, to minimize the sample complexity.

\subsection{State preparation fidelity}
Let $\ket{\Delta,V_2}$ denote the ground state of the Rydberg Hamiltonian in \cref{eq:H_ryd} at the parameters $(\Delta,V_2)$. 
The phase diagram is shown in the inset of \cref{fig:rydberg-all}(a)~\cite{m:slagleQuantumSpinLiquids2022}, and for the parameter regime we are concerned with, hosts two phases: a disordered (I) phase and a $\mathbb{Z}_2$ ordered (II) phase.

In this example, we assume we have experimentally prepared (or alternatively, believe we have prepared) the state $\ket{\Delta,0}$ at $\Delta = -\Omega$ and would like to extract its fidelity to $\ket{\Delta_\text{ref},0}$ for some $\Delta_\text{ref}\in[-2\Omega,2\Omega]$.
In this case, the fidelity is the expectation value of the rank-one projector $\ketbra{\Delta_{\text{ref}},0}$.
Following our discussion in \cref{sec:quench-setups}, as this is a low-rank observable, we globally couple the state to ancillary degrees of freedom in order to minimize the sample complexity. This can be done, for example, by physically moving an ancillary array of atoms in their respective electronic ground states next to the original system, a capability that has been demonstrated using optical tweezers in~\cite{m:bluvsteinQuantumProcessorBased2022}, in order to initiate the `expansion' step of our protocol. 

As depicted in \cref{fig:rydberg-all}(a), we choose $n_\sys = 8$ (orange circles), and $n_\anc =  n_\sys + 2 = 10$ (gray circles), and place the ancillae such that the system and ancilla atoms are mutually blockaded [we choose $n_\anc$ slightly bigger than $n_\sys$ to account for the fact that doubling the system size does not exactly double the Hilbert space, owing to the Rydberg blockade forbidding some states of the extended system]. The choice of $n_\anc$ ensures that $d_\ext$ is well above $d_\sys^2$. 
We then quench evolve the extended system under the Hamiltonian $H_\ryd$ (for example, by driving the atoms with a global laser) at the same parameters as those of the prepared state: $\Delta = -\Omega, V_2 = 0$ for a time $t = 2\pi/\Omega$. Because the system and ancillae interact, the global state $|-\Omega,0\rangle \otimes |g^{\otimes n_\text{anc}} \rangle$ is not an eigenstate of the full Hamiltonian and thus undergoes time evolution.

We subsequently measure the resulting state and extract the expectation value of the projector $\ketbra{\Delta_{\text{ref}},0}$ for various $\Delta_\text{ref}$ by numerically applying the optimal recovery map $R$ [\cref{eq:optimal-QL}] to the measurement data. 
\Cref{fig:rydberg-all}(a) plots the extracted fidelity at $m = 2000$ samples, and showcases how our protocol successfully tracks the exact fidelity even at this relatively small number of measurement samples. 
This allows one to certify the preparation of the ground state and to calibrate experimental parameters to prepare ground states at particular points in the phase diagram. 
We also compare the histograms of the extracted fidelity using the Moore-Penrose recovery map [\cref{eq:naive-QL}]
and the optimal version derived in \cref{eq:optimal-QL}. 
As can be seen, the optimal data extraction using \cref{eq:optimal-QL} results in an almost two-fold reduction of the sample complexity over the use of the Moore-Penrose recovery map.

\subsection{Energy transport}
We next consider probing the hydrodynamics of energy transport in this system. 
Specifically, we imagine preparing an array of $n_\sys = 24$ atoms in the ground state $\ket{\Delta,V_2}$ of $H_\ryd$ at $\Delta = -\Omega$ and $V_2 = 0$. 
We then perturb the middle (twelfth) site by applying a $\pi/2$ rotation about the $y$-axis (in the blockaded Hilbert space):
\begin{align} 
	\ket{\psi(0)} = e^{-i \frac{\pi}{2}\mathcal{P} Y_{12} \mathcal{P} }\ket{\Delta,0},\label{eq:perturbation}
\end{align}
thus bringing the system out of equilibrium by introducing a slight excess of energy localized at the middle of the chain. 
We then let the perturbed system evolve under the same Hamiltonian $H_\ryd(-\Omega,0)$ for time $\tau$ and aim to extract the local energy density
\begin{align} 
	E_i \equiv  \frac{\Omega}{2} \mathcal P X_i \mathcal P - \Delta n_i~. \label{eq:energy_density}
\end{align}
at various $\tau$s. Note that the time $\tau$ here denotes the free evolution time in the hydrodynamics experiment (``physics quench") and is to be distinguished from $t$, the quench time of the extended system in our protocol (``measurement quench"), which is the time that the global system is evolved for after bringing in the ancillae. We have also chosen to set $V_2=0$ in this hypothetical experiment such that the energy density \cref{eq:energy_density} is a strictly single-site observable.
Since the total energy is conserved, the energy density necessarily obeys a continuity equation, and so the initial excess of energy at the middle of the chain is expected to spread to neighboring sites over time $\tau$. 

The discussion in \cref{sec:expansion} (and \cref{sec:quench-setups}) suggests that it is most efficient to extract the local energy density using a local patched quench, i.e., by coupling separate ancillae to local system degrees of freedom. 
In this incarnation of our protocol, we therefore imagine first physically moving the atoms apart to distances such that each atom can be considered to be isolated, using optical tweezer rearrangement capabilities. We then couple each system atom to three introduced ancillary atoms [\cref{fig:rydberg-all}(b), orange and gray circles, respectively] and let the extended system evolve under $H_\ryd$ at $\Delta = -\Omega, V_2 = 0$ for quench time $t = 4 \times 2\pi/\Omega$, before measuring. 
In \cref{fig:rydberg-all}(b), we plot the local energy density $\avg{E_i}$ extracted from post-processing $m = 2000$ measurements samples at different physical evolution times $\tau \in \{0,2,4\}\times 2\pi/\Omega$. As can be seen from the overlay of the exact values (solid line), the estimated values (dots) correctly capture the hydrodynamics of energy transport of the system for all $\tau$.

We briefly comment on the computational complexity of data post-processing. To extract the fidelity using the global setup, we have to numerically compute the \scrambler~$S$, which is a $d_\ext \times d_\sys^2$. Therefore, extracting information using the global setup is only feasible for small systems. In contrast, in the local setup used to measure the local energy density, the \scrambler~$S$ factorizes into a tensor product of $n_\sys$ \scrambler s, the size of which depends only on the dimension of the local extended system.
The computational cost of processing the snapshots from the local setup only increases linearly with total system size $n_\sys$, making such measurements feasible for large systems.

\subsection{Entanglement dynamics}\label{sec:entropy}
Lastly, we  study the dynamics of entanglement entropy in the same nonequilibrium experiment.
Concretely, after free evolution time $\tau$, we aim to extract the R\'enyi-2 bipartite entanglement entropy across various bipartitions that divide the system into subsystem $A$ comprised of the first $l$ sites and  subsystem $B$ comprised of the remaining sites:
\begin{align} 
	\mathcal S_2(A) \equiv - \log \Tr(\rho_A^2) = - \log \Tr(\rho_B^2), \label{eq:renyi2-def}
\end{align}
where $\rho_A \equiv \Tr_{B}(\rho)$ is the reduced density matrix of the subsystem $A$. While the R\'enyi-2 entropy is a 
quantity that depends non-linearly on the state $\rho$, it can be obtained from the expectation value of a linear operator $\mathbb{S}^A$ that is the swap operator that permutes two identical copies of the system $A$, evaluated within the replicated state $\rho \otimes \rho$: $\langle \mathbb{S}^A \rangle = \Tr(\rho_A^2)$.
As such, its estimation from measurement data requires only a simple modification of \cref{eq:ozdef}:
 \begin{align}
\langle  \mathbb{S}^A \rangle &\approx \frac{1}{m} \sum_{i,j = 1}^{m/2} s^A_{z_i,z_j}, \quad s^A_{z_i,z_j} = \obra{\mathbb{S}} R^{\otimes 2}|z_i, z_j\rangle,\label{eq:estimate_S}
\end{align}
where we divide the $m$ measurement snapshots into two independent sets $\{ z_i \}_{i=1}^{m/2}$ and $\{z_j\}_{j=1}^{m/2}$. 
In practice (and in what is demonstrated in our numerics), the following so-called $U$-statistics offers a more sample-efficient estimator of $\avg{\mathbb S^A}$~\cite{m:hoeffdingClassStatisticsAsymptotically1948,m:huangPredictingManyProperties2020}:
\begin{align} 
	 \avg{\mathbb{S}^A} \approx \frac{2}{m(m-1)} \sum_{1\leq i<j\leq m} s^A_{z_i,z_j}. 
\end{align}
Note that we can additionally define an estimator for $\mathbb S^B$, the operator that swaps two identical copies of the subsystem $B$.
Because the initial state is pure, the estimators of $\mathbb S^A$ and $\mathbb S^B$ converge to the same value as $m\rightarrow \infty$.
However, they may have different variances with a finite number of samples $m$ (i.e., their sample complexities may be different).
In our numerics, we compute both estimators from each set of samples and use the one with lower variance to estimate the R\'enyi-2 entropy.

To implement our protocol, we choose to arrange our ancillae atoms in a patched setup [\cref{fig:schemes}(b)].
Because 
the swap operator $\mathbb{S^{A}}$ (or $\mathbb{S^{B}}$) acts globally on the $A$ (or $B$) subsystems in the two-copy Hilbert space, we expect that the single-site patched setup used to estimate energy densities [\cref{fig:rydberg-all}(b)] would result in a high sample complexity. Instead, for an efficient extraction, we expect the optimal patch configuration to contain either the $A$ or $B$ subsystem.

Here, to balance computational and sample complexities, we separate the system of $n_\sys = 24$ atoms into four patches of six atoms each [\cref{fig:rydberg-all}(c), dashed boxes]. 
We couple each patch to $12$ ancillary atoms, which guarantees $d_\ext>d_\sys^2$ and also allows for easy classical simulability.
Following the free evolution $\tau$, we again quench the extended system under $H_\ryd$ at $\Delta = -\Omega, V_2 = 0$ for a measurement quench time of $t = 6 \times 2 \pi/\Omega$, before measuring.
We plot in \cref{fig:rydberg-all}(c) the extracted $\mathcal S_2(L)$ at different $\tau\in \{0,2,4\}\times 2\pi/\Omega$ and compare it with the exact values in \cref{fig:rydberg-all}(c). Our results demonstrate the successful extraction of the dynamics of entanglement entropy. 
Also, we observe that the sample complexity only depends on the number of patches that overlap with the subsystem $A$---it increases dramatically when the subsystem $A$ contains more than one patch, validating our expectations.

\section{Fermions on an Optical Lattice: Distinguishing s-wave from d-wave Superconductivity}\label{sec:bcs}

\begin{figure*}[t]
\includegraphics[width=0.8\textwidth]{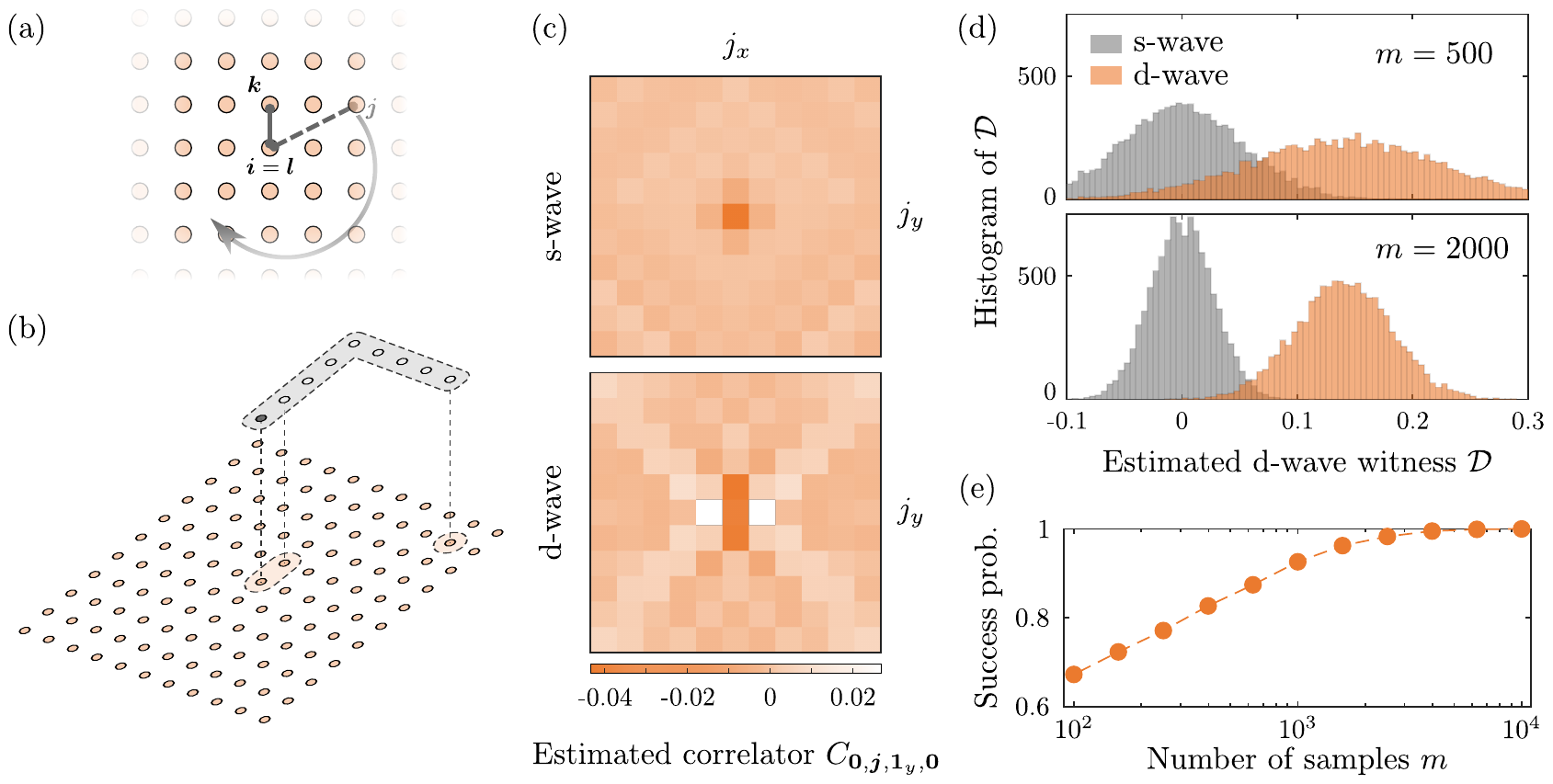}
\caption{Extraction of the superconducting pairing order parameter in a system of itinerant fermionic particles on a square optical lattice, using our protocol.
(a) We fix $\bm i, \bm k,\bm l$ and measure the correlator $C_{\bm i,\bm j,\bm k,\bm l}$ at different $\bm j = (j_x,j_y)$ on the two-dimensional lattice (orange dots).
In a $d$-wave superconductor, we expect $C_{\bm 0,\bm j,\bm 1_y,\bm 0}$ to change signs as $\bm j$ (dashed line) rotates around $\bm 1_y$ (solid line).
(b) To measure the correlators, we allow particles to tunnel vertically into the ``bridge''---an ancillary one-dimensional lattice (gray dots)---that connects the support of $C_{\bm i,\bm j,\bm k,\bm l}$.
The bridge initially has two fermions, one for each spin orientation, both located at one of its ends (denoted by the dark circle in the bridge).
(c) The correlators $C_{\bm 0,\bm j,\bm 1_y,\bm 0}$, each extracted using $m = 10^6$ samples, reveal patterns consistent with $s$-wave and $d$-wave pairing in the initial states. 
(d) The distribution for the estimates of the $d$-wave witness $\mathcal D$ (\cref{eq:d-witness-def}) obtained using $m=500$ and $m=2000$ samples.
We obtain the histograms by generating $10^4$ independent sets of $m$ samples and estimate $\mathcal D$ using each set. 
(e) The success probability in distinguishing $s$-wave and $d$-wave pairing using the witness $\mathcal D$ as a function of the number of samples $m$. 
We can distinguish between the $s$-wave (gray) and $d$-wave (orange) ans\"atze with near-100\% probability with only $m = 2000$ samples. 
}
\label{fig:bcs-all}
\end{figure*}

In this section, we demonstrate an example where the bridged setup discussed in \cref{sec:quench-setups} is required to overcome symmetry constraints.
Concretely, we consider a low-temperature  system of fermions in an optical lattice prepared in an unknown superconducting state, and discuss extracting signatures of their long-range pairing order, which is the expectation value of creating a Cooper pair in one region space and creating another pair in a far-away region in space. 
The dynamics of such fermions is well described by the Fermi-Hubbard model, which is also a paradigmatic model of high-temperature superconductivity in which it is believed that pairing is mediated by spin fluctuations~\cite{m:emeryTheoryHighMathrmT1987,lee2006doping,m:keimerQuantumMatterHightemperature2015}. However, due to its complexity and the presence of strong interactions, such conjectures have not been verified theoretically or numerically. Analog quantum simulators are able to simulate large systems, and show promise in shedding light on the nature of superconductivity in the Fermi-Hubbard model~\cite{m:hartObservationAntiferromagneticCorrelations2015,chiu2019string,m:hartkeDoublonHoleCorrelationsFluctuation2020,m:jiCouplingMobileHole2021}. In particular, high-temperature superconductors are known to exhibit unconventional, $d$-wave superconductivity~\cite{m:tsueiPairingSymmetryCuprate2000}. Verifying such superconducting order in an analog quantum simulator would constitute an important experimental milestone.

Specifically, we consider a system of spin-1/2 fermions on a two-dimensional square lattice of linear size $L = 11$.
We assume that the system is in a Bardeen-Cooper-Schrieffer (BCS) state $\ket{\psi}$ 
with either $s$-wave or $d$-wave pairing order~\cite{m:grosSuperconductivityCorrelatedWave1988} and our task is to distinguish this. Namely, we take
\begin{align} 
	\ket{\psi} \propto \exp\left(\sum_{\bm k} a(\bm k) c_{\bm k,\uparrow}^\dag c_{-\bm k,\downarrow}^\dag\right) \ket{0},
\end{align}
where $\ket{0}$ is the vacuum and $c_{\bm k,\sigma}^\dag$ is the fermionic operator that creates a fermion with momentum $\bm k = (k_x,k_y)$ and spin $\sigma \in \{\uparrow,\downarrow\}$. 
The function $a(\mathbf k)$ is given by
\begin{align} 
	a(\mathbf k) &= \Delta(\bm k)/\left[\xi_{\bm k} + \sqrt{\xi_{\bm k}^2 + \Delta^2(\bm k)}\right],
\end{align}
where the dispersion and gap functions are given respectively by
\begin{align}
	\xi_{\bm k} &= -2 (\cos k_x + \cos k_y) - \mu,\\
	\Delta(\bm k) &= 
	 \begin{cases}
		\Delta & \text{(s-wave)}\\
		\Delta(\cos k_x - \cos k_y) & \text{(d-wave)}
	\end{cases}.\label{eq:ak}
\end{align}
In particular, we choose $\mu = 1/2, \Delta = 5$, resulting in an average of $\bar \eta \approx 1.08$ and $\bar \eta\approx 0.9$ total fermions per site respectively in the s- and d-wave states, near half-filling. 
Note that such a state has indefinite particle number, and that it is a fermionic Gaussian state: this allows for easy numerical simulations as its reduced density matrices can be efficiently  constructed and sampled from~\cite{m:chungDensitymatrixSpectraSolvable2001}.

To reveal the pairing order of the state, we extract the correlators
\begin{align} 
	C_{\bm i,\bm j,\bm k, \bm l} \equiv  \bra{\psi}c_{\bm i,\uparrow}^\dag c_{\bm j,\downarrow}^\dag c_{\bm k,\uparrow}^{\mathstrut} c_{\bm l,\downarrow}^{\mathstrut} \ket{\psi}~,
	\label{eqn:long-range_correlator}
\end{align}
where $\bm i, \bm j, \bm k, \bm l$ are real space positions of sites on the lattice [\cref{fig:bcs-all}(a)].
This correlation function corresponds to annihilating a Cooper pair on sites $\bm k, \bm l$ and creating another on sites $\bm i,\bm j$. In a superconducting state with $d$-wave pairing, the correlator changes sign depending on the relative angles between the vectors $\bm i - \bm j$ and $\bm k - \bm l$. 
In contrast, an $s$-wave superconductor is isotropic and the correlator does not depend on this relative angle. Our correlator can be viewed as a real space analog of conventional momentum-space pairing correlation functions~\cite{m:jiangHighTemperatureSuperconductivity2021}. 
In our numerics, we fix $\bm i = \bm l$ to be the center of the lattice $\bm{0} = (0,0)$, fix $\bm k$ to be the site above $\bm{i}$: $\bm k = \bm{1}_y = (0,1)$, and extract $C_{\bm 0,\bm j, \bm{1}_y, \bm 0}$ at different positions $\bm j$.

The correlation function \cref{eqn:long-range_correlator} cannot be measured with our protocol, using local patched quenches (that is, when the regions $\{\bm{0}, \bm{1}_y\}$ and $\{\bm j\}$ are coupled to their own set of ancillae).
This is because such a quench preserves the particle numbers of each patch, while the correlation function involves annihilating and creating electrons on different sites and would not be accessible from the corresponding measurement snapshots.
Therefore, to implement our protocol, we have to design the interactivity of the system and ancillae carefully. We first envision  isolating the sites $\{\bm{0}, \bm{1}_y\}$ and $\{\bm j\}$ from the rest of the state, then coupling the system vertically to an ancillary second layer [\cref{fig:bcs-all}(b), gray circles], which forms a ``bridge" allowing particle exchange between $\{\bm{0}, \bm{1}_y\}$ and $\{\bm j\}$.
The bridge is initially empty, except for a site at one end of the bridge where it is filled with two electrons (one for each spin orientation). 
Such bilayer systems have been demonstrated in recent experiments, including Refs.~\cite{m:gallCompetingMagneticOrders2021,m:hartkeDoublonHoleCorrelationsFluctuation2020}. 
We then evolve the extended system under the Fermi-Hubbard model
\begin{align} 
	H_{\FH} = - J \sum_{(\bm i,\bm j)\in \mathcal A}\sum_{\sigma}c_{\bm i,\sigma}^\dag c_{\bm j,\sigma}^{\mathstrut} + U\sum_{\bm i} n_{\bm i,\uparrow}n_{\bm i,\downarrow}~,
\end{align}
at values $(J,U) = (1,1.5)$ for time $t = 2/J$, and measure the site-resolved occupation numbers of the extended lattice. 
Above, the set $\mathcal A$ contains vertically aligned sites and all nearest-neighbor pairs within each layer [sites which are boxed in \cref{fig:bcs-all}(b) and also joined vertically with dashed lines].

Next, we compute the recovery map according to \cref{sec:framework} and use it to obtain an estimate of 
$C_{\bm 0,\bm j, \bm{1}_y, \bm 0}$ from the measurement snapshots.
Now, since the scrambling quench only involves the sites in the support of $C_{\bm 0,\bm j, \bm{1}_y, \bm 0}$, 
 the recovery map $R$ is only sensitive to the bit-string data supported on sites in $\mathcal{A}$ only.
That is to say, even though we may be measuring the global system to produce global bit-strings, we post-process the data ignoring the bit-string information of sites outside of $\mathcal{A}$---essentially, tracing those degrees of freedom out.
Hence, our current numerical simulations benefit from 
a convenient tracing out of those sites~\footnote{We can obtain any reduced density matrix of the BCS states by following the procedure in Ref.~\cite{m:chungDensitymatrixSpectraSolvable2001}. Note that Eq.~(16) in Ref.~\cite{m:chungDensitymatrixSpectraSolvable2001} should read $2\alpha = a^{11}-ca^{22}c^T$.} before applying the scrambling quench.
In \cref{fig:bcs-all}(c), we plot the correlators $C_{\bm 0,\bm j, \bm{1}_y, \bm 0}$
at different sites $\bm j$ for $s$- and $d$-wave BCS states; each site $\bm j$ requires $m = 10^6$ samples. The extracted correlators indeed display patterns characteristic of their respective spatial symmetries. 

While visually appealing, constructing such a detailed spatial map requires a relatively large number of samples. If we are instead interested in the simpler task of just {\it distinguishing} between $s$- and $d$- wave superconductivity,  we may achieve this with far fewer samples.
This can be accomplished by defining a $d$-wave witness
\begin{align} 
	\mathcal{D} \equiv \sum_{\bm j\in \mathcal{N}(\bm i)} \chi_{\bm j}  
	C_{\bm 0,\bm j, \bm{1}_y, \bm 0},\label{eq:d-witness-def}
\end{align}
where $\mathcal{N}(\bm i)$ is the set of four neighbor sites that neighbor $\bm i$ and $\chi_{\bm j}$ is a filter function taking value $-1$ when $\langle\bm{ij}\rangle$ is a vertical bond and $+1$ when $\langle\bm{ij}\rangle$ is horizontal.
Clearly, $\mathcal D = 0$ if 
$C_{\bm 0,\bm j, \bm{1}_y, \bm 0}$ is isotropic (corresponding to the $s$-wave pattern) and $\mathcal D \neq 0$ otherwise.
In \cref{fig:bcs-all}(d), we plot the histograms of the values of $\mathcal D$ one would get at different sample sizes.
As in \cref{sec:rydberg}, these histograms are of the estimated $\mathcal D$ obtained by simulating a large number ($10^4$ independent sets) of $m$ samples. Even at moderate $m$, one sees that there is a discernible difference between the peaks of the histogram ($\mathcal{D} \approx 0$ for $s$-wave and $\mathcal{D} \approx 0.136$ for $d$-wave). 
These peaks become more sharply defined and well-isolated for larger $m$, indicating that the success probability of correctly identifying $d$-wave pairing from any single fixed set of $m$ measurement snapshots tends to unity as $m \to \infty$.

Indeed, to estimate the success probability quantitatively, we simply define, for each $m$, the fraction of $m$-sample sets that have $\mathcal{D} \geq \mathcal D^*\approx 0.068$, where $\mathcal{D}^*$ is the mid-point between the two peaks in \cref{fig:bcs-all}(d).
\cref{fig:bcs-all}(e) shows the success probability as a function of $m$. We see that a relatively modest sample size of $m \sim 2000$ is already sufficient to reliably (i.e., with high probability) distinguish $s$- from $d$-wave ordering.

Finally, we note that, in this numerical example, we extract each correlator $C_{\bm i,\bm j,\bm k, \bm l}$ using a different set of $m$ samples. Therefore, the total numbers of samples used in computing the spatial pattern in \cref{fig:bcs-all}(c) and the witness in \cref{fig:bcs-all}(d) are $L^2m$ and $4m$, respectively. 
However, in practice, different correlators $C_{\bm 0,\bm j,\bm 1_y, \bm 0}$ can be extracted in parallel from the same experimental runs if the quench bridges are defined on non-overlapping regions of a large two-dimensional lattice.
If the initial BCS state is translationally invariant, for example, the four correlators that contribute to the witness $\mathcal D$ can be simultaneously measured on different sublattices and $\mathcal D$ can be extracted using $m$ instead of $4m$ samples. Indeed, these $4m$ correlator samples can even be obtained from a \textit{single} experimental shot on a state with $\sim 11m$ lattice sites, partitioned into small patched quenches of various geometries.

\section{Bosons on an Optical Lattice: Extracting Signatures of Topological Order}\label{sec:bose-hubbard}

\begin{figure*}
    \centering
    \includegraphics[width=0.95\textwidth]{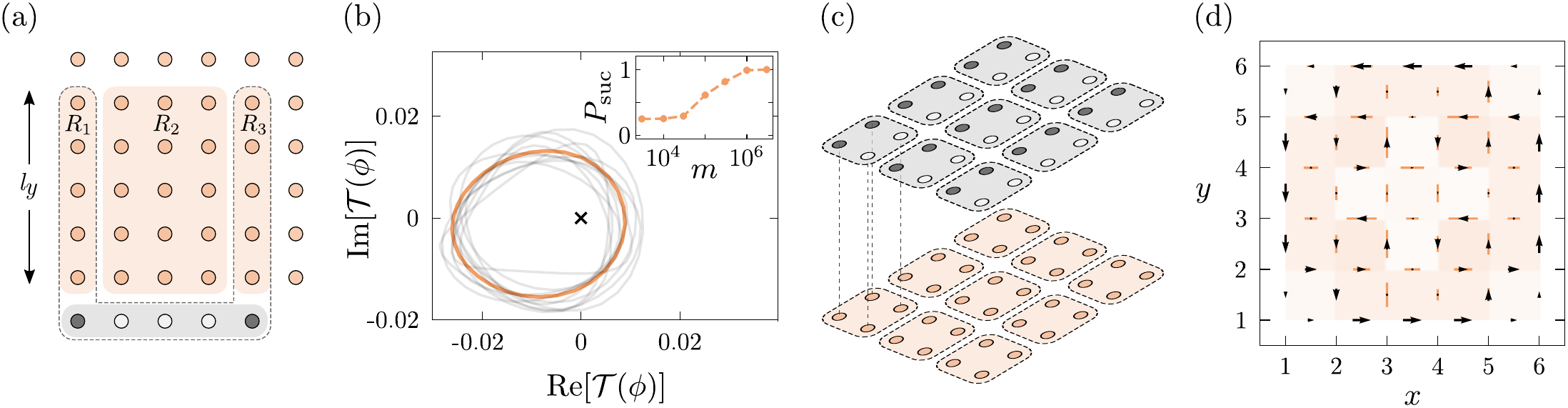}
    \caption{Measurement of many-body Chern number (MBCN) and bond currents in the Hofstadter-Bose-Hubbard (HBH) model. 
    (a) Quench scheme for measurement of MBCN. This quantity is estimated by a quench on a small patch: coupling rectangular regions $R_1$ and $R_3$ (orange vertical strips of length $l_y$) with five additional sites containing two bosons (gray horizontal strip), followed by projective measurement on $R_2$ (large orange block).
    (b) We illustrate the 
    parametric plot of the quantity $\langle \mathcal{T}(\phi)\rangle$ as a function of $\phi \in [0,2\pi)$ for the ground state of a HBH model, with three particles on 36 sites, and a flux of $2\pi \alpha =\pi/2$ per plaquette in open boundary conditions. The true curve of $\mathcal{T}(\phi)$ is plotted in orange, which we see winds around the origin once, corresponding to a MBCN of $\mathcal{C} = 1$. We plot in gray several representative curves of $\hat{\mathcal{T}}(\phi)$ estimated with $m = 3\times 10^6$ random samples. 
    Such curves deviate from the ideal $\mathcal{T}(\phi)$, but nevertheless all have winding number 1, hence their estimated MBCN $\hat{C}$ are all 1. Inset: The probability $P_\text{suc}$ of obtaining the correct MBCN is plotted as a function of number of samples.
    Beyond a critical number of measurements $\sim 3\times 10^5$, the MBCN can be reliably estimated.  
    (c) Quench scheme to measure currents across every bond. We isolate plaquettes of four lattice sites and couple it to an additional plaquette of four sites and two bosons in a second layer. In order to determine the currents on every bond, we perform two sets of patched quenches on non-overlapping plaquettes (not shown). 
    (d) The ground state shows a pattern of persistent edge currents. The statistical uncertainty arising from 5000 sample measurements is indicated by the orange bars: for this state, 5000 measurements are sufficient to achieve a good signal-to-noise ratio. As a visual aid, we color each plaquette with the curl of the current in that plaquette.}
    \label{fig:BH}
\end{figure*}

Having demonstrated several variations of our protocol in the previous two sections, we consider a final example that highlights our protocol in extracting exotic physical properties in current experiments with limited control and readout capabilities. Specifically, we consider a system of itinerant bosonic particles hopping on an optical lattice with a gauge field, modeled by the Hofstadter-Bose-Hubbard (HBH) Hamiltonian.
This model has been investigated both theoretically and experimentally as a likely host of bosonic fractional quantum Hall states of matter~\cite{m:cooperFractionalQuantumHall2020}. However, it is extremely challenging to directly measure signatures of topological order in current experiments, as this either requires the extraction of complicated quantities like global winding numbers or nontrivial Wilson loops.

The HBH Hamiltonian is given by~\cite{m:aidelsburgerMeasuringChernNumber2015}
\begin{align}
    H_\text{HBH} = &-J \sum_{x,y} \L(b^\dagger_{x+1,y} b_{x,y}^{\mathstrut} + e^{2\pi i \alpha x} b^\dagger_{x,y+1} b_{x,y}^{\mathstrut} + \text{h.c.}\R) \nonumber\\
    &+ \frac{U}{2} \sum_{x,y} n_{x,y} \L(n_{x,y}-1\R),
    \label{eq:HBH_ham}
\end{align}
where $b^\dagger_{x,y}$ is the bosonic creation operator on site $(x,y)$ and $n_{x,y} = b^\dagger_{x,y} b_{x,y}$ measures the number of bosons on that site. Above, we have assumed that the model is of $N$ bosons hopping on a two-dimensional $L_x \times L_y$ rectangular lattice, with a flux of $2\pi \alpha$ per plaquette and on-site interaction $U$.
Experimentally, such complex-valued tunneling phases have been realized through techniques such as laser assisted tunneling~\cite{m:aidelsburgerMeasuringChernNumber2015,m:taiMicroscopyInteractingHarper2017,m:leonardRealizationFractionalQuantum2022}.

In the thermodynamic limit, the HBH model hosts a rich ground state phase diagram (\cref{app:MBCN}). 
In particular, the Laughlin state at a filling fraction $\nu = N/N_\text{flux} = 1/2$ (where $N_\text{flux} = 2\pi \alpha L_x L_y$ in periodic boundary conditions) is a fractional quantum Hall state and has been thoroughly investigated, numerically and analytically. It is also observed to be robust to finite size effects as well as boundary conditions~\cite{m:petrescuPrecursorLaughlinState2017,m:dehghaniExtractionManybodyChern2021}. Here, we apply our protocol to characterize this state---namely, the ground state of the HBH model with a flux of $\alpha = 0.25$ on a  $L_x=L_y=6$ lattice with three bosons, and with an on-site interaction strength of $U/J=5$, 
and attempt to extract complex quantities that are typically difficult to measure: the so-called many-body Chern number, as well as the currents along every bond (i.e., bond currents). 

\subsection{Many-body Chern number}
We first briefly explain the many-body Chern number (MBCN),  introduced by Niu, Thouless and Wu in~\cite{m:niuQuantizedHallConductance1985} as a generalization of the single-particle Chern number, a topological invariant that characterizes the topology of a single-particle band. The MBCN is defined by placing the system on a torus and threading fluxes through its non-contractible loops. This amounts to enforcing twisted boundary conditions, with twist angles $\theta$ and $\phi$ serving as the amount of threaded flux. As with the single-particle case, the MBCN is defined through the derivatives of the ground state $\ket{E_0}$ with twist angle:
\begin{equation}
\mathcal{C} = \frac{1}{2\pi i} \int_0^{2\pi} d\phi \int_0^{2\pi} d\theta \left(\braket{\partial_\theta E_0}{\partial_\phi E_0} - \braket{\partial_\phi E_0}{\partial_\theta E_0}\right)~.\label{eq:NTW_MBCN}
\end{equation}
While conceptually appealing, Eq.~\eqref{eq:NTW_MBCN}
as stated is impractical to experimentally realize. Instead, conventional experimental approaches to estimate such a topological quantity rely on indirect methods, such as measuring response functions like the Hall conductivity, which are proportional to the MBCN in the thermodynamic limit. For example, Ref.~\cite{m:aidelsburgerMeasuringChernNumber2015} measured the Hall conductivity of a non-interacting model upon applying a linear potential, and Refs.~\cite{m:motrukDetectingFractionalChern2020,m:repellinFractionalChernInsulators2020} proposed schemes to do so in interacting models. As other examples, Ref.~\cite{m:taiMicroscopyInteractingHarper2017} observed signatures of chiral currents in a ladder geometry via time-of-flight measurements, and Ref.~\cite{m:leonardRealizationFractionalQuantum2022} observed signatures of a fractional quantum Hall state by measuring particle densities and their correlation functions.

It would be ideal to extract the MBCN in a direct fashion, given an experimentally prepared topological state. 
To this end, we consider a recently-introduced, novel method to measure the MBCN from a \textit{single} wavefunction,  proposed by Ref.~\cite{m:dehghaniExtractionManybodyChern2021}, and derived using techniques from topological quantum field theory. 
It involves measuring the winding number of a complex-valued observable $ \mathcal{T}(\phi)$, with $\phi$ winding from $0$ to $2\pi$:
concretely, the MBCN $\mathcal{C}$ is given by 
\begin{align}
    \mathcal{C} &= \frac{1}{2\pi i} \oint \frac{d\langle \mathcal{T}(\phi) \rangle}{\langle \mathcal{T}(\phi) \rangle}.
    \label{eqn:winding}
\end{align}
While formally theoretically justified   in the limit of large lattices and periodic or cylindrical boundary conditions, it was numerically verified that this formula accurately estimates the MBCN even in open boundary conditions and moderate system sizes, in which case the operator  $\mathcal{T}(\phi)$ is defined as follows.
On a rectangular lattice, choose three disjoint, rectangular subsystems $R_{1,2,3}$, of width $l_y$ [\cref{fig:BH}(a)] 
Then   $\mathcal{T}(\phi)$ is composed of a product of operations on $R_{1,2,3}$:
\begin{equation}
    \mathcal{T}(\phi) = W_{1}^\dagger(\phi) \mathbb{S}_{1,3} W_{1}^{\mathstrut}(\phi) V_{1}^s V_{2}^s~,
    \label{eqn:T_operator}
\end{equation}
where $V_j$ is the so-called \textit{polarization operator} acting on region $R_j$,
$W(\phi)$ is the so-called \textit{twist angle} operator with variable $\phi \in [0,2\pi)$, and the swap operator $\mathbb{S}_{1,3}$ exchanges the regions $R_1$ and $R_3$:~\cite{m:cianManyBodyChernNumber2021}
\begin{gather}
 V_j = \prod_{(x,y) \in R_j} \exp(i \frac{2\pi y} {l_y} n_{x,y})~,\\
 W_j(\phi) = \prod_{(x,y)\in R_j} \exp(i n_{x,y}\phi)~,\\ 
\mathbb{S}_{1,3} = \bigotimes_{\substack{(x,y) \in R_1\\ (x',y) \in R_3}} \sum_{a,b=1}^{d_\text{loc.}} \ketbra*{a^{\mathstrut}_{(x,y)}b_{(x',y)}}{b_{(x,y)}a^{\mathstrut}_{(x',y)}}~,
\end{gather}
In \cref{eqn:T_operator}, $s$ is an integer which equals the expected ground state degeneracy, with $s=2$ for the $\nu = 1/2$ Laughlin state~\cite{m:dehghaniExtractionManybodyChern2021}, and $a,b$ are basis elements of the $d_\text{loc.}$-dimensional on-site Hilbert space.

We employ our protocol to measure $\mathcal{T}(\phi)$, and hence $\mathcal{C}$, using \cref{eqn:winding}.
Concretely, on the system of $6\times 6$ square lattice with $3$ bosons we are considering, we choose $R_1$ and $R_3$ to be disjoint, well-separated $5\times 1$ rectangular strips and $R_2$ to be the intervening $5 \times 3$ rectangular block separating the two. 
This choice of subsystems is the largest one for which the map $S$ is numerically invertible.
In order to minimize experimental requirements, we exploit the fact that $V_j$ (and $W_j$) is diagonal in the standard measurement basis and can be measured from standard measurements on $R_j$.
The only observable that requires our protocol is the swap operator $\mathbb{S}_{1,3}$. To estimate $\mathcal{T}(\phi)$, it suffices to quench evolve a single patch $R_1 \cup R_3$ to measure $W_1^\dagger (\phi) \mathbb{S}_{1,3} W_1(\phi) V_1^s$; the scrambling map is only defined on $R_1 \cup R_3$, and $V_2^s$ can be simultaneously measured by projective readout on $R_2$; the remaining sites can either be left unmeasured or their measurement outcomes can be traced out. 
We therefore imagine coupling $R_1$ and $R_3$ to an ancillary system of five additional sites with two bosons in a bridged quench  [\cref{fig:BH}(b)]~\footnote{For numerical tractability, we disregarded all instances of having three bosons on $R_1 \cup R_3$ (since these are extremely rare ($<0.1 \%$) and computationally expensive to invert)}, before quench evolving the joint system under the same HBH Hamiltonian for 10 hopping times: $T = 10/J$, and then measuring.

 The results of our numerical experiments are plotted in Fig.~\ref{fig:BH}(b). With a finite number of measurement samples $m$, there are statistical fluctuations: each gray curve   is a parametric plot of $\langle \mathcal{T}(\phi) \rangle$, extracted from a single instance of an experimental run with $m = 3\times 10^6$ measurement samples---we see that the curve $\langle \mathcal{T}(\phi) \rangle$ is deformed from its true curve (orange loop). Despite this, we notice the winding number around the origin is always faithfully estimated. If the statistical fluctuations are large enough (when $m$ is too small) that the curve $\langle \mathcal{T}(\phi)\rangle$  is deformed, it may not enclose the origin and the winding number will be incorrectly estimated.
Therefore, the number of measurements must be sufficiently large in order to reliably estimate MBCN.
We plot the success probability $P_\text{suc}$ of obtaining the correct MBCN with $m$, defined as the fraction of experimental runs with fixed measurement samples $m$ that reproduce the correct winding number [Fig.~\ref{fig:BH}(b) inset]. We see that above  $m \sim 3\times 10^5$, there is a high  probability ($P_\text{suc} >  0.8$) of measuring the correct MBCN. This number is approximately determined by the signal-to-noise of $\mathcal{T}(\phi^*)$ at the angle $\phi^*$ where $\abs{\langle\mathcal{T}(\phi^*)\rangle}$ is smallest: $m^* \approx \text{Var}_z(\abs{\langle\mathcal{T}(\phi^*)\rangle})/\abs{\langle\mathcal{T}(\phi^*)\rangle}^2$.
While this number is relatively high, we note that this can potentially be reduced by judicious choice of subsystem and total system sizes: for example, the MBCN can be reliably measured on a smaller system system of $4 \times 4$ lattice sites and 3 bosons with $\sim 2000$ measurements (data not shown)

Before moving on, we make two remarks. First, our extraction of MBCN in this example is performed for systems which are smaller than those investigated in Ref.~\cite{m:dehghaniExtractionManybodyChern2021,m:cianManyBodyChernNumber2021}. 
However, while the MBCN is in general fairly sensitive to choices of parameter such as subsystem sizes and positions,  we verify that it remains robust in our regime of interest, faithfully capturing the transition into the Laughlin state, in agreement with conventional observables such as the density of doublons~\cite{m:palmBosonicPfaffianState2021}~(\cref{app:MBCN}). 
Second, Ref.~\cite{m:cianManyBodyChernNumber2021} has proposed an alternative protocol to measure the MBCN with \textit{randomized measurements}: the application of random unitary gates on the state of interest, followed by measurements in the computational basis. However, despite experimental proposals and limited realizations with disordered potentials, applying random unitary gates in a system of bosonic itinerant particles remains challenging: approaches in the literature are based off quench evolutions with random potentials, which were numerically argued to converge to a two-design~\cite{m:ohligerEfficientFeasibleState2013,mm:vermerschUnitaryDesignsRandom2018}. 
The virtue of our protocol is twofold: it does not require any fine-grained control of dynamics and it does not rely on any assumptions on the formation of a two-design, therefore can be trusted to be quantitatively accurate.

\subsection{Bond currents}
Besides the MBCN, the presence of chiral edge currents  can also serve as a hallmark for topological order. 
In particular, different phases of matter in the HBH model 
are distinguishable by the spatial patterns of their bond currents, like  the Meissner~\cite{m:petrescuBosonicMottInsulator2013} and vortex-lattice phases~\cite{m:piraudVortexMeissnerPhases2015,greschner2015spontaneous} in ladder geometries.
However, these current distributions have thus far only been coarsely measured through ad-hoc schemes such as the aforementioned time-of-flight measurement~\cite{m:taiMicroscopyInteractingHarper2017}. 

Here, we apply our method to measure the local currents in the $\nu = 1/2$ Laughlin state. At each bond $((x,y),(x',y'))$ on the rectangular lattice, we measure the current across it~\cite{m:piraudVortexMeissnerPhases2015}: 
\begin{equation}
j_{((x,y),(x',y'))} = i J_{((x,y),(x',y'))} \langle b^\dagger_{x,y} b^{\mathstrut}_{x',y'} \rangle + \text{h.c.}~,    
\end{equation}
where $J_{((x,y),(x',y'))}$ is the hopping amplitude from $(x',y')$ to $(x,y)$ appearing in the Hamiltonian~\cref{eq:HBH_ham}. Since this is a local operator, it is more efficiently measured via local patched quenches.

In \cref{fig:BH}(c) we illustrate our hypothetical setup: we divide the lattice into patches of non-overlapping $2\times 2$ plaquettes, and overlay a second layer of ancillary sites also divided into $2\times 2$ plaquettes, each initialized in a pre-determined state of two filled sites with one boson each and two empty sites.
Due to the symmetry constraints discussed in \cref{sec:recoverability}, we can only measure the currents across bonds contained in a single patch. As such, in order to measure the current across every bond, we require two different experiments  using two  different sets of patched quenches, one in which patches are all shifted  by one site in both the $x$ and $y$ directions from the other. Doing so thus ensures that we extract information from every bond.
 Quench evolution is initiated by  allowing vertical tunneling between plaquettes with the same HBH Hamiltonian as before, for 10 hopping times $t = 10/J$, before measurements are taken. 
 
\cref{fig:BH}(d) shows the results derived from $m = 5\times 10^3$ measurement samples. 
Qualitatively, we can observe the presence of a chiral edge current localized on the boundary of the system, in line with expectations of the effective edge theory of a topological system placed on a system with a boundary~\cite{m:cornfeldChiralCurrentsOnedimensional2015,m:senthilIntegerQuantumHall2013}.
It would be extremely exciting to be able to use such measurements to quantitatively extract parameters of this edge theory, such as the chiral central charge; however, the system size considered is simply too small to be able make a definitive statement in the present numerics. 
Nevertheless, what our quench protocol demonstrates is the ability to reliably measure the currents on every bond, a key step towards such an experiment.

\section{Discussion \& Outlook}\label{sec:conclusion}
 
In this paper, we have  proposed a universal, scalable, and noise-resilient protocol to measure arbitrary physical properties in analog quantum simulators.
It circumvents limited controllability of present-day experimental platforms, and exploits naturally-realizable quantum many-body dynamics to scramble initially inaccessible information into ancillary degrees of freedom, following which such information can be recovered via appropriate classical data processing of the global measurement outcomes.

We discussed in detail the performance of our protocol: the required number of samples, the classical computational overhead, the evolution time required, and its robustness in the presence of noise.
Additionally, we provided detailed numerical examples of its successful employment in systems of arrays of Rydberg atoms and itinerant particles on optical lattices. These demonstrate practical and feasible near-term applications of our protocol for the extraction of novel and important observables that are otherwise difficult to measure in experiments. Our protocol hence promises to greatly increase the utility and versatility of quantum technologies today. 

From a practical standpoint, our protocol is extremely flexible and can be modified or improved along various fronts.
For example, the introduction of ancillae in our protocol serves to effect the exponentially many linear maps required of a randomized measurement protocol.
This step may be replaced (or in fact augmented) by 
evolving under different quench unitaries, e.g., quench under different Hamiltonians, or quench under the same Hamiltonian for different times~\cite{m:huClassicalShadowTomography2021}, or both, thus also achieving an ensemble of linear maps (though note a drawback of this approach is the need to consider exponentially many quench dynamics with different evolution times to achieve an ensemble of ``rotations'' with similar size as $N$ ancillae). 
One can analyze this setting in much the same framework introduced in this work. Indeed, by introducing a fictitious ancillary register that keeps track of the particular unitary applied to the system, we can extract observables by following the same data processing procedure outlined in the paper.
This generalization in fact enables possible trade-offs between, for instance, adding ancillae to the system and evolving the extended system for different times or under different Hamiltonians.

Another aspect where our protocol's performance may be improved is if we impose additional structure/knowledge on the state to be characterized, or the scrambling unitaries utilized. In this work, we have assumed that the  state on the system of interest is an arbitrary state in a $d_\sys$-dimensional Hilbert space. In experiments, it is often the case that we have prior information that limits the initial state to a smaller part of the Hilbert space (for example, that the state is pure~\cite{m:grierSampleoptimalClassicalShadows2022}, or if the state has a low enough temperature near the ground state). 
Intuitively, such knowledge should lower the required number of ancillae and reduce the computational overhead for data processing.
How to account for such information in a modification of our protocol is an interesting and important practical question. 
Further, the data processing procedure in our protocol relies on numerical computation of the recovery map, which is its main computational cost. 
Besides using patched quenches to control this computational overhead, one may consider subjecting the extended system to quench evolutions that admit efficient classical simulations, such as free-fermionic dynamics~\cite{m:wanMatchgateShadowsFermionic2022} or interacting integrable dynamics. Despite their non quantum-chaotic nature, our numerical investigations found that their dynamics still lead to the tomographic completeness of the protocol.
With such classical simulability comes the possibility of also an efficient classical implementation of the recovery map; this approach would be similar to the use of random Clifford circuits, which can be efficiently classically simulated, to realize classical shadow tomography~\cite{m:huangPredictingManyProperties2020}.

There are numerous advanced randomized measurement schemes with various capabilities, including process tomography~\cite{m:kunjummenShadowProcessTomography2022, m:levyClassicalShadowsQuantum2021,m:huangLearningPredictArbitrary2022}, feedforward techniques to determine the best next measurement~\cite{m:huangEfficientEstimationPauli2021,m:rathImportanceSamplingRandomized2021}, and tomography with shallow circuits~\cite{m:akhtarScalableFlexibleClassical2022,m:bertoniShallowShadowsExpectation2022,m:arienzoClosedformAnalyticExpressions2022a}. Adapting these innovations into the setting of natural quench dynamics could yield further improvements to our protocol. In particular, state and process learning has applications to the task of quantum sensing. To this end, it is promising to tailor our protocol for metrological purposes.

From a conceptual viewpoint, our protocol can be understood as a special type of quantum-classical hybrid algorithm,  enabled by the ergodicity of generic interacting quantum dynamics.
This is captured by a relation we identified  between the out-of-time-ordered correlator (OTOC), commonly used in information scrambling, and the tomographic completeness of our protocol. It would be extremely interesting to explore future connections, such as quantifying the relation between the degree of information scrambling and the sample and computational complexity. 
Deeper understanding of the structure of entanglement and information scrambling in chaotic quantum dynamics may enable the development of more advanced quantum algorithms that can be implemented in near-term devices with limited controllability.

{\it Note:} A  manuscript appearing in the same arXiv posting \cite{accompanying} also proposes a quantum state learning protocol whose operating principle is similar to ours:   originally inaccessible information of a system of interest, once scrambled into
ancillary degrees of freedom, can be recovered by classical processing of the global measurement data. However, a key difference is in the assumptions that enter in the scrambling of information and hence the recovering of information. In Ref.~\cite{accompanying}, it is assumed that the scrambling map exhibits certain universal statistics (namely, that the tomographic ensemble forms a quantum state-design),  such that a particular analytical recovery map can be used. In our work, we do not require the appearance of such universal distributions. 

\begin{acknowledgments}
We would like to thank 
Immanuel~Bloch,
Hong-Ye~Hu,
Hsin-Yuan~Huang,
Guang~Hao~Low,
and Yi-Zhuang~You
for useful discussions.
M.~C.~T. acknowledges support from
the DARPA (134371-5113608)
and the Quantum Algorithms and Machine Learning grant from NTT
(AGMT DTD 9/24/20).
D.~K.~M.~is supported in part by NSF CIQC (2016245).
W.~W.~H.~is supported in part by the Stanford
Institute of Theoretical Physics, and in part by the National University of Singapore (NUS)  start-up grants A-8000599-00-00 and A-8000599-01-00.
\end{acknowledgments}
\bibliography{zotero-generated-do-not-edit,my-bib}

\begin{thebibliography}{111}%
\makeatletter
\providecommand \@ifxundefined [1]{%
 \@ifx{#1\undefined}
}%
\providecommand \@ifnum [1]{%
 \ifnum #1\expandafter \@firstoftwo
 \else \expandafter \@secondoftwo
 \fi
}%
\providecommand \@ifx [1]{%
 \ifx #1\expandafter \@firstoftwo
 \else \expandafter \@secondoftwo
 \fi
}%
\providecommand \natexlab [1]{#1}%
\providecommand \enquote  [1]{``#1''}%
\providecommand \bibnamefont  [1]{#1}%
\providecommand \bibfnamefont [1]{#1}%
\providecommand \citenamefont [1]{#1}%
\providecommand \href@noop [0]{\@secondoftwo}%
\providecommand \href [0]{\begingroup \@sanitize@url \@href}%
\providecommand \@href[1]{\@@startlink{#1}\@@href}%
\providecommand \@@href[1]{\endgroup#1\@@endlink}%
\providecommand \@sanitize@url [0]{\catcode `\\12\catcode `\$12\catcode
  `\&12\catcode `\#12\catcode `\^12\catcode `\_12\catcode `\%12\relax}%
\providecommand \@@startlink[1]{}%
\providecommand \@@endlink[0]{}%
\providecommand \url  [0]{\begingroup\@sanitize@url \@url }%
\providecommand \@url [1]{\endgroup\@href {#1}{\urlprefix }}%
\providecommand \urlprefix  [0]{URL }%
\providecommand \Eprint [0]{\href }%
\providecommand \doibase [0]{https://doi.org/}%
\providecommand \selectlanguage [0]{\@gobble}%
\providecommand \bibinfo  [0]{\@secondoftwo}%
\providecommand \bibfield  [0]{\@secondoftwo}%
\providecommand \translation [1]{[#1]}%
\providecommand \BibitemOpen [0]{}%
\providecommand \bibitemStop [0]{}%
\providecommand \bibitemNoStop [0]{.\EOS\space}%
\providecommand \EOS [0]{\spacefactor3000\relax}%
\providecommand \BibitemShut  [1]{\csname bibitem#1\endcsname}%
\let\auto@bib@innerbib\@empty
%</preamble>
\bibitem [{\citenamefont {Hart}\ \emph {et~al.}(2015)\citenamefont {Hart},
  \citenamefont {Duarte}, \citenamefont {Yang}, \citenamefont {Liu},
  \citenamefont {Paiva}, \citenamefont {Khatami}, \citenamefont {Scalettar},
  \citenamefont {Trivedi}, \citenamefont {Huse},\ and\ \citenamefont
  {Hulet}}]{m:hartObservationAntiferromagneticCorrelations2015}%
  \BibitemOpen
  \bibfield  {author} {\bibinfo {author} {\bibfnamefont {R.~A.}\ \bibnamefont
  {Hart}}, \bibinfo {author} {\bibfnamefont {P.~M.}\ \bibnamefont {Duarte}},
  \bibinfo {author} {\bibfnamefont {T.-L.}\ \bibnamefont {Yang}}, \bibinfo
  {author} {\bibfnamefont {X.}~\bibnamefont {Liu}}, \bibinfo {author}
  {\bibfnamefont {T.}~\bibnamefont {Paiva}}, \bibinfo {author} {\bibfnamefont
  {E.}~\bibnamefont {Khatami}}, \bibinfo {author} {\bibfnamefont {R.~T.}\
  \bibnamefont {Scalettar}}, \bibinfo {author} {\bibfnamefont {N.}~\bibnamefont
  {Trivedi}}, \bibinfo {author} {\bibfnamefont {D.~A.}\ \bibnamefont {Huse}},\
  and\ \bibinfo {author} {\bibfnamefont {R.~G.}\ \bibnamefont {Hulet}},\
  }\bibfield  {title} {\bibinfo {title} {Observation of antiferromagnetic
  correlations in the {{Hubbard}} model with ultracold atoms},\ }\href
  {https://doi.org/10.1038/nature14223} {\bibfield  {journal} {\bibinfo
  {journal} {Nature}\ }\textbf {\bibinfo {volume} {519}},\ \bibinfo {pages}
  {211} (\bibinfo {year} {2015})}\BibitemShut {NoStop}%
\bibitem [{\citenamefont {Chiu}\ \emph
  {et~al.}(2019{\natexlab{a}})\citenamefont {Chiu}, \citenamefont {Ji},
  \citenamefont {Bohrdt}, \citenamefont {Xu}, \citenamefont {Knap},
  \citenamefont {Demler}, \citenamefont {Grusdt}, \citenamefont {Greiner},\
  and\ \citenamefont {Greif}}]{m:chiuStringPatternsDoped2019}%
  \BibitemOpen
  \bibfield  {author} {\bibinfo {author} {\bibfnamefont {C.~S.}\ \bibnamefont
  {Chiu}}, \bibinfo {author} {\bibfnamefont {G.}~\bibnamefont {Ji}}, \bibinfo
  {author} {\bibfnamefont {A.}~\bibnamefont {Bohrdt}}, \bibinfo {author}
  {\bibfnamefont {M.}~\bibnamefont {Xu}}, \bibinfo {author} {\bibfnamefont
  {M.}~\bibnamefont {Knap}}, \bibinfo {author} {\bibfnamefont {E.}~\bibnamefont
  {Demler}}, \bibinfo {author} {\bibfnamefont {F.}~\bibnamefont {Grusdt}},
  \bibinfo {author} {\bibfnamefont {M.}~\bibnamefont {Greiner}},\ and\ \bibinfo
  {author} {\bibfnamefont {D.}~\bibnamefont {Greif}},\ }\bibfield  {title}
  {\bibinfo {title} {String patterns in the doped {{Hubbard}} model},\ }\href
  {https://doi.org/10.1126/science.aav3587} {\bibfield  {journal} {\bibinfo
  {journal} {Science}\ }\textbf {\bibinfo {volume} {365}},\ \bibinfo {pages}
  {251} (\bibinfo {year} {2019}{\natexlab{a}})}\BibitemShut {NoStop}%
\bibitem [{\citenamefont {Hartke}\ \emph {et~al.}(2020)\citenamefont {Hartke},
  \citenamefont {Oreg}, \citenamefont {Jia},\ and\ \citenamefont
  {Zwierlein}}]{m:hartkeDoublonHoleCorrelationsFluctuation2020}%
  \BibitemOpen
  \bibfield  {author} {\bibinfo {author} {\bibfnamefont {T.}~\bibnamefont
  {Hartke}}, \bibinfo {author} {\bibfnamefont {B.}~\bibnamefont {Oreg}},
  \bibinfo {author} {\bibfnamefont {N.}~\bibnamefont {Jia}},\ and\ \bibinfo
  {author} {\bibfnamefont {M.}~\bibnamefont {Zwierlein}},\ }\bibfield  {title}
  {\bibinfo {title} {Doublon-{{Hole Correlations}} and {{Fluctuation
  Thermometry}} in a {{Fermi-Hubbard Gas}}},\ }\href
  {https://doi.org/10.1103/PhysRevLett.125.113601} {\bibfield  {journal}
  {\bibinfo  {journal} {Phys. Rev. Lett.}\ }\textbf {\bibinfo {volume} {125}},\
  \bibinfo {pages} {113601} (\bibinfo {year} {2020})}\BibitemShut {NoStop}%
\bibitem [{\citenamefont {Ji}\ \emph {et~al.}(2021)\citenamefont {Ji},
  \citenamefont {Xu}, \citenamefont {Kendrick}, \citenamefont {Chiu},
  \citenamefont {Br{\"u}ggenj{\"u}rgen}, \citenamefont {Greif}, \citenamefont
  {Bohrdt}, \citenamefont {Grusdt}, \citenamefont {Demler}, \citenamefont
  {Lebrat},\ and\ \citenamefont {Greiner}}]{m:jiCouplingMobileHole2021}%
  \BibitemOpen
  \bibfield  {author} {\bibinfo {author} {\bibfnamefont {G.}~\bibnamefont
  {Ji}}, \bibinfo {author} {\bibfnamefont {M.}~\bibnamefont {Xu}}, \bibinfo
  {author} {\bibfnamefont {L.~H.}\ \bibnamefont {Kendrick}}, \bibinfo {author}
  {\bibfnamefont {C.~S.}\ \bibnamefont {Chiu}}, \bibinfo {author}
  {\bibfnamefont {J.~C.}\ \bibnamefont {Br{\"u}ggenj{\"u}rgen}}, \bibinfo
  {author} {\bibfnamefont {D.}~\bibnamefont {Greif}}, \bibinfo {author}
  {\bibfnamefont {A.}~\bibnamefont {Bohrdt}}, \bibinfo {author} {\bibfnamefont
  {F.}~\bibnamefont {Grusdt}}, \bibinfo {author} {\bibfnamefont
  {E.}~\bibnamefont {Demler}}, \bibinfo {author} {\bibfnamefont
  {M.}~\bibnamefont {Lebrat}},\ and\ \bibinfo {author} {\bibfnamefont
  {M.}~\bibnamefont {Greiner}},\ }\bibfield  {title} {\bibinfo {title}
  {Coupling a {{Mobile Hole}} to an {{Antiferromagnetic Spin Background}}:
  {{Transient Dynamics}} of a {{Magnetic Polaron}}},\ }\href
  {https://doi.org/10.1103/PhysRevX.11.021022} {\bibfield  {journal} {\bibinfo
  {journal} {Phys. Rev. X}\ }\textbf {\bibinfo {volume} {11}},\ \bibinfo
  {pages} {021022} (\bibinfo {year} {2021})}\BibitemShut {NoStop}%
\bibitem [{\citenamefont {Aidelsburger}\ \emph {et~al.}(2013)\citenamefont
  {Aidelsburger}, \citenamefont {Atala}, \citenamefont {Lohse}, \citenamefont
  {Barreiro}, \citenamefont {Paredes},\ and\ \citenamefont
  {Bloch}}]{m:aidelsburgerRealizationHofstadterHamiltonian2013}%
  \BibitemOpen
  \bibfield  {author} {\bibinfo {author} {\bibfnamefont {M.}~\bibnamefont
  {Aidelsburger}}, \bibinfo {author} {\bibfnamefont {M.}~\bibnamefont {Atala}},
  \bibinfo {author} {\bibfnamefont {M.}~\bibnamefont {Lohse}}, \bibinfo
  {author} {\bibfnamefont {J.~T.}\ \bibnamefont {Barreiro}}, \bibinfo {author}
  {\bibfnamefont {B.}~\bibnamefont {Paredes}},\ and\ \bibinfo {author}
  {\bibfnamefont {I.}~\bibnamefont {Bloch}},\ }\bibfield  {title} {\bibinfo
  {title} {Realization of the {{Hofstadter Hamiltonian}} with {{Ultracold
  Atoms}} in {{Optical Lattices}}},\ }\href
  {https://doi.org/10.1103/PhysRevLett.111.185301} {\bibfield  {journal}
  {\bibinfo  {journal} {Phys. Rev. Lett.}\ }\textbf {\bibinfo {volume} {111}},\
  \bibinfo {pages} {185301} (\bibinfo {year} {2013})}\BibitemShut {NoStop}%
\bibitem [{\citenamefont {Aidelsburger}\ \emph {et~al.}(2015)\citenamefont
  {Aidelsburger}, \citenamefont {Lohse}, \citenamefont {Schweizer},
  \citenamefont {Atala}, \citenamefont {Barreiro}, \citenamefont
  {Nascimb{\`e}ne}, \citenamefont {Cooper}, \citenamefont {Bloch},\ and\
  \citenamefont {Goldman}}]{m:aidelsburgerMeasuringChernNumber2015}%
  \BibitemOpen
  \bibfield  {author} {\bibinfo {author} {\bibfnamefont {M.}~\bibnamefont
  {Aidelsburger}}, \bibinfo {author} {\bibfnamefont {M.}~\bibnamefont {Lohse}},
  \bibinfo {author} {\bibfnamefont {C.}~\bibnamefont {Schweizer}}, \bibinfo
  {author} {\bibfnamefont {M.}~\bibnamefont {Atala}}, \bibinfo {author}
  {\bibfnamefont {J.~T.}\ \bibnamefont {Barreiro}}, \bibinfo {author}
  {\bibfnamefont {S.}~\bibnamefont {Nascimb{\`e}ne}}, \bibinfo {author}
  {\bibfnamefont {N.~R.}\ \bibnamefont {Cooper}}, \bibinfo {author}
  {\bibfnamefont {I.}~\bibnamefont {Bloch}},\ and\ \bibinfo {author}
  {\bibfnamefont {N.}~\bibnamefont {Goldman}},\ }\bibfield  {title} {\bibinfo
  {title} {Measuring the {{Chern}} number of {{Hofstadter}} bands with
  ultracold bosonic atoms},\ }\href {https://doi.org/10.1038/nphys3171}
  {\bibfield  {journal} {\bibinfo  {journal} {Nat. Phys.}\ }\textbf {\bibinfo
  {volume} {11}},\ \bibinfo {pages} {162} (\bibinfo {year} {2015})}\BibitemShut
  {NoStop}%
\bibitem [{\citenamefont {Yang}\ \emph {et~al.}(2020)\citenamefont {Yang},
  \citenamefont {Sun}, \citenamefont {Ott}, \citenamefont {Wang}, \citenamefont
  {Zache}, \citenamefont {Halimeh}, \citenamefont {Yuan}, \citenamefont
  {Hauke},\ and\ \citenamefont {Pan}}]{m:yangObservationGaugeInvariance2020}%
  \BibitemOpen
  \bibfield  {author} {\bibinfo {author} {\bibfnamefont {B.}~\bibnamefont
  {Yang}}, \bibinfo {author} {\bibfnamefont {H.}~\bibnamefont {Sun}}, \bibinfo
  {author} {\bibfnamefont {R.}~\bibnamefont {Ott}}, \bibinfo {author}
  {\bibfnamefont {H.-Y.}\ \bibnamefont {Wang}}, \bibinfo {author}
  {\bibfnamefont {T.~V.}\ \bibnamefont {Zache}}, \bibinfo {author}
  {\bibfnamefont {J.~C.}\ \bibnamefont {Halimeh}}, \bibinfo {author}
  {\bibfnamefont {Z.-S.}\ \bibnamefont {Yuan}}, \bibinfo {author}
  {\bibfnamefont {P.}~\bibnamefont {Hauke}},\ and\ \bibinfo {author}
  {\bibfnamefont {J.-W.}\ \bibnamefont {Pan}},\ }\bibfield  {title} {\bibinfo
  {title} {Observation of gauge invariance in a 71-site
  {{Bose}}\textendash{{Hubbard}} quantum simulator},\ }\href
  {https://doi.org/10.1038/s41586-020-2910-8} {\bibfield  {journal} {\bibinfo
  {journal} {Nature}\ }\textbf {\bibinfo {volume} {587}},\ \bibinfo {pages}
  {392} (\bibinfo {year} {2020})}\BibitemShut {NoStop}%
\bibitem [{\citenamefont {L{\'e}onard}\ \emph {et~al.}()\citenamefont
  {L{\'e}onard}, \citenamefont {Kim}, \citenamefont {Kwan}, \citenamefont
  {Segura}, \citenamefont {Grusdt}, \citenamefont {Repellin}, \citenamefont
  {Goldman},\ and\ \citenamefont
  {Greiner}}]{m:leonardRealizationFractionalQuantum2022}%
  \BibitemOpen
  \bibfield  {author} {\bibinfo {author} {\bibfnamefont {J.}~\bibnamefont
  {L{\'e}onard}}, \bibinfo {author} {\bibfnamefont {S.}~\bibnamefont {Kim}},
  \bibinfo {author} {\bibfnamefont {J.}~\bibnamefont {Kwan}}, \bibinfo {author}
  {\bibfnamefont {P.}~\bibnamefont {Segura}}, \bibinfo {author} {\bibfnamefont
  {F.}~\bibnamefont {Grusdt}}, \bibinfo {author} {\bibfnamefont
  {C.}~\bibnamefont {Repellin}}, \bibinfo {author} {\bibfnamefont
  {N.}~\bibnamefont {Goldman}},\ and\ \bibinfo {author} {\bibfnamefont
  {M.}~\bibnamefont {Greiner}},\ }\bibfield  {title} {\bibinfo {title}
  {Realization of a fractional quantum {{Hall}} state with ultracold atoms},\
  }\href {http://arxiv.org/abs/2210.10919} {\ }\Eprint
  {https://arxiv.org/abs/2210.10919} {arXiv:2210.10919} \BibitemShut {NoStop}%
\bibitem [{\citenamefont {Scholl}\ \emph {et~al.}(2021)\citenamefont {Scholl},
  \citenamefont {Schuler}, \citenamefont {Williams}, \citenamefont
  {Eberharter}, \citenamefont {Barredo}, \citenamefont {Schymik}, \citenamefont
  {Lienhard}, \citenamefont {Henry}, \citenamefont {Lang}, \citenamefont
  {Lahaye}, \citenamefont {L{\"a}uchli},\ and\ \citenamefont
  {Browaeys}}]{m:schollQuantumSimulation2D2021}%
  \BibitemOpen
  \bibfield  {author} {\bibinfo {author} {\bibfnamefont {P.}~\bibnamefont
  {Scholl}}, \bibinfo {author} {\bibfnamefont {M.}~\bibnamefont {Schuler}},
  \bibinfo {author} {\bibfnamefont {H.~J.}\ \bibnamefont {Williams}}, \bibinfo
  {author} {\bibfnamefont {A.~A.}\ \bibnamefont {Eberharter}}, \bibinfo
  {author} {\bibfnamefont {D.}~\bibnamefont {Barredo}}, \bibinfo {author}
  {\bibfnamefont {K.-N.}\ \bibnamefont {Schymik}}, \bibinfo {author}
  {\bibfnamefont {V.}~\bibnamefont {Lienhard}}, \bibinfo {author}
  {\bibfnamefont {L.-P.}\ \bibnamefont {Henry}}, \bibinfo {author}
  {\bibfnamefont {T.~C.}\ \bibnamefont {Lang}}, \bibinfo {author}
  {\bibfnamefont {T.}~\bibnamefont {Lahaye}}, \bibinfo {author} {\bibfnamefont
  {A.~M.}\ \bibnamefont {L{\"a}uchli}},\ and\ \bibinfo {author} {\bibfnamefont
  {A.}~\bibnamefont {Browaeys}},\ }\bibfield  {title} {\bibinfo {title}
  {Quantum simulation of {{2D}} antiferromagnets with hundreds of {{Rydberg}}
  atoms},\ }\href {https://doi.org/10.1038/s41586-021-03585-1} {\bibfield
  {journal} {\bibinfo  {journal} {Nature}\ }\textbf {\bibinfo {volume} {595}},\
  \bibinfo {pages} {233} (\bibinfo {year} {2021})}\BibitemShut {NoStop}%
\bibitem [{\citenamefont {Ebadi}\ \emph {et~al.}(2021)\citenamefont {Ebadi},
  \citenamefont {Wang}, \citenamefont {Levine}, \citenamefont {Keesling},
  \citenamefont {Semeghini}, \citenamefont {Omran}, \citenamefont {Bluvstein},
  \citenamefont {Samajdar}, \citenamefont {Pichler}, \citenamefont {Ho},
  \citenamefont {Choi}, \citenamefont {Sachdev}, \citenamefont {Greiner},
  \citenamefont {Vuleti{\'c}},\ and\ \citenamefont
  {Lukin}}]{m:ebadiQuantumPhasesMatter2021}%
  \BibitemOpen
  \bibfield  {author} {\bibinfo {author} {\bibfnamefont {S.}~\bibnamefont
  {Ebadi}}, \bibinfo {author} {\bibfnamefont {T.~T.}\ \bibnamefont {Wang}},
  \bibinfo {author} {\bibfnamefont {H.}~\bibnamefont {Levine}}, \bibinfo
  {author} {\bibfnamefont {A.}~\bibnamefont {Keesling}}, \bibinfo {author}
  {\bibfnamefont {G.}~\bibnamefont {Semeghini}}, \bibinfo {author}
  {\bibfnamefont {A.}~\bibnamefont {Omran}}, \bibinfo {author} {\bibfnamefont
  {D.}~\bibnamefont {Bluvstein}}, \bibinfo {author} {\bibfnamefont
  {R.}~\bibnamefont {Samajdar}}, \bibinfo {author} {\bibfnamefont
  {H.}~\bibnamefont {Pichler}}, \bibinfo {author} {\bibfnamefont {W.~W.}\
  \bibnamefont {Ho}}, \bibinfo {author} {\bibfnamefont {S.}~\bibnamefont
  {Choi}}, \bibinfo {author} {\bibfnamefont {S.}~\bibnamefont {Sachdev}},
  \bibinfo {author} {\bibfnamefont {M.}~\bibnamefont {Greiner}}, \bibinfo
  {author} {\bibfnamefont {V.}~\bibnamefont {Vuleti{\'c}}},\ and\ \bibinfo
  {author} {\bibfnamefont {M.~D.}\ \bibnamefont {Lukin}},\ }\bibfield  {title}
  {\bibinfo {title} {Quantum phases of matter on a 256-atom programmable
  quantum simulator},\ }\href {https://doi.org/10.1038/s41586-021-03582-4}
  {\bibfield  {journal} {\bibinfo  {journal} {Nature}\ }\textbf {\bibinfo
  {volume} {595}},\ \bibinfo {pages} {227} (\bibinfo {year}
  {2021})}\BibitemShut {NoStop}%
\bibitem [{\citenamefont {Neyenhuis}\ \emph {et~al.}(2017)\citenamefont
  {Neyenhuis}, \citenamefont {Zhang}, \citenamefont {Hess}, \citenamefont
  {Smith}, \citenamefont {Lee}, \citenamefont {Richerme}, \citenamefont {Gong},
  \citenamefont {Gorshkov},\ and\ \citenamefont
  {Monroe}}]{m:neyenhuisObservationPrethermalizationLongrange2017}%
  \BibitemOpen
  \bibfield  {author} {\bibinfo {author} {\bibfnamefont {B.}~\bibnamefont
  {Neyenhuis}}, \bibinfo {author} {\bibfnamefont {J.}~\bibnamefont {Zhang}},
  \bibinfo {author} {\bibfnamefont {P.~W.}\ \bibnamefont {Hess}}, \bibinfo
  {author} {\bibfnamefont {J.}~\bibnamefont {Smith}}, \bibinfo {author}
  {\bibfnamefont {A.~C.}\ \bibnamefont {Lee}}, \bibinfo {author} {\bibfnamefont
  {P.}~\bibnamefont {Richerme}}, \bibinfo {author} {\bibfnamefont {Z.-X.}\
  \bibnamefont {Gong}}, \bibinfo {author} {\bibfnamefont {A.~V.}\ \bibnamefont
  {Gorshkov}},\ and\ \bibinfo {author} {\bibfnamefont {C.}~\bibnamefont
  {Monroe}},\ }\bibfield  {title} {\bibinfo {title} {Observation of
  prethermalization in long-range interacting spin chains},\ }\href
  {https://doi.org/10.1126/sciadv.1700672} {\bibfield  {journal} {\bibinfo
  {journal} {Sci. Adv.}\ }\textbf {\bibinfo {volume} {3}},\ \bibinfo {pages}
  {e1700672} (\bibinfo {year} {2017})}\BibitemShut {NoStop}%
\bibitem [{\citenamefont {Choi}\ \emph {et~al.}(2019)\citenamefont {Choi},
  \citenamefont {Zhou}, \citenamefont {Choi}, \citenamefont {Landig},
  \citenamefont {Ho}, \citenamefont {Isoya}, \citenamefont {Jelezko},
  \citenamefont {Onoda}, \citenamefont {Sumiya}, \citenamefont {Abanin},\ and\
  \citenamefont {Lukin}}]{m:choiProbingQuantumThermalization2019}%
  \BibitemOpen
  \bibfield  {author} {\bibinfo {author} {\bibfnamefont {J.}~\bibnamefont
  {Choi}}, \bibinfo {author} {\bibfnamefont {H.}~\bibnamefont {Zhou}}, \bibinfo
  {author} {\bibfnamefont {S.}~\bibnamefont {Choi}}, \bibinfo {author}
  {\bibfnamefont {R.}~\bibnamefont {Landig}}, \bibinfo {author} {\bibfnamefont
  {W.~W.}\ \bibnamefont {Ho}}, \bibinfo {author} {\bibfnamefont
  {J.}~\bibnamefont {Isoya}}, \bibinfo {author} {\bibfnamefont
  {F.}~\bibnamefont {Jelezko}}, \bibinfo {author} {\bibfnamefont
  {S.}~\bibnamefont {Onoda}}, \bibinfo {author} {\bibfnamefont
  {H.}~\bibnamefont {Sumiya}}, \bibinfo {author} {\bibfnamefont {D.~A.}\
  \bibnamefont {Abanin}},\ and\ \bibinfo {author} {\bibfnamefont {M.~D.}\
  \bibnamefont {Lukin}},\ }\bibfield  {title} {\bibinfo {title} {Probing
  {{Quantum Thermalization}} of a {{Disordered Dipolar Spin Ensemble}} with
  {{Discrete Time-Crystalline Order}}},\ }\href
  {https://doi.org/10.1103/PhysRevLett.122.043603} {\bibfield  {journal}
  {\bibinfo  {journal} {Phys. Rev. Lett.}\ }\textbf {\bibinfo {volume} {122}},\
  \bibinfo {pages} {043603} (\bibinfo {year} {2019})}\BibitemShut {NoStop}%
\bibitem [{\citenamefont {Peng}\ \emph {et~al.}(2021)\citenamefont {Peng},
  \citenamefont {Yin}, \citenamefont {Huang}, \citenamefont {Ramanathan},\ and\
  \citenamefont {Cappellaro}}]{m:pengFloquetPrethermalizationDipolar2021}%
  \BibitemOpen
  \bibfield  {author} {\bibinfo {author} {\bibfnamefont {P.}~\bibnamefont
  {Peng}}, \bibinfo {author} {\bibfnamefont {C.}~\bibnamefont {Yin}}, \bibinfo
  {author} {\bibfnamefont {X.}~\bibnamefont {Huang}}, \bibinfo {author}
  {\bibfnamefont {C.}~\bibnamefont {Ramanathan}},\ and\ \bibinfo {author}
  {\bibfnamefont {P.}~\bibnamefont {Cappellaro}},\ }\bibfield  {title}
  {\bibinfo {title} {Floquet prethermalization in dipolar spin chains},\ }\href
  {https://doi.org/10.1038/s41567-020-01120-z} {\bibfield  {journal} {\bibinfo
  {journal} {Nat. Phys.}\ }\textbf {\bibinfo {volume} {17}},\ \bibinfo {pages}
  {444} (\bibinfo {year} {2021})}\BibitemShut {NoStop}%
\bibitem [{\citenamefont {Choi}\ \emph {et~al.}(2017)\citenamefont {Choi},
  \citenamefont {Choi}, \citenamefont {Landig}, \citenamefont {Kucsko},
  \citenamefont {Zhou}, \citenamefont {Isoya}, \citenamefont {Jelezko},
  \citenamefont {Onoda}, \citenamefont {Sumiya}, \citenamefont {Khemani},
  \citenamefont {{von Keyserlingk}}, \citenamefont {Yao}, \citenamefont
  {Demler},\ and\ \citenamefont
  {Lukin}}]{m:choiObservationDiscreteTimecrystalline2017}%
  \BibitemOpen
  \bibfield  {author} {\bibinfo {author} {\bibfnamefont {S.}~\bibnamefont
  {Choi}}, \bibinfo {author} {\bibfnamefont {J.}~\bibnamefont {Choi}}, \bibinfo
  {author} {\bibfnamefont {R.}~\bibnamefont {Landig}}, \bibinfo {author}
  {\bibfnamefont {G.}~\bibnamefont {Kucsko}}, \bibinfo {author} {\bibfnamefont
  {H.}~\bibnamefont {Zhou}}, \bibinfo {author} {\bibfnamefont {J.}~\bibnamefont
  {Isoya}}, \bibinfo {author} {\bibfnamefont {F.}~\bibnamefont {Jelezko}},
  \bibinfo {author} {\bibfnamefont {S.}~\bibnamefont {Onoda}}, \bibinfo
  {author} {\bibfnamefont {H.}~\bibnamefont {Sumiya}}, \bibinfo {author}
  {\bibfnamefont {V.}~\bibnamefont {Khemani}}, \bibinfo {author} {\bibfnamefont
  {C.}~\bibnamefont {{von Keyserlingk}}}, \bibinfo {author} {\bibfnamefont
  {N.~Y.}\ \bibnamefont {Yao}}, \bibinfo {author} {\bibfnamefont
  {E.}~\bibnamefont {Demler}},\ and\ \bibinfo {author} {\bibfnamefont {M.~D.}\
  \bibnamefont {Lukin}},\ }\bibfield  {title} {\bibinfo {title} {Observation of
  discrete time-crystalline order in a disordered dipolar many-body system},\
  }\href {https://doi.org/10.1038/nature21426} {\bibfield  {journal} {\bibinfo
  {journal} {Nature}\ }\textbf {\bibinfo {volume} {543}},\ \bibinfo {pages}
  {221} (\bibinfo {year} {2017})}\BibitemShut {NoStop}%
\bibitem [{\citenamefont {Zhang}\ \emph {et~al.}(2017)\citenamefont {Zhang},
  \citenamefont {Hess}, \citenamefont {Kyprianidis}, \citenamefont {Becker},
  \citenamefont {Lee}, \citenamefont {Smith}, \citenamefont {Pagano},
  \citenamefont {Potirniche}, \citenamefont {Potter}, \citenamefont
  {Vishwanath}, \citenamefont {Yao},\ and\ \citenamefont
  {Monroe}}]{m:zhangObservationDiscreteTime2017}%
  \BibitemOpen
  \bibfield  {author} {\bibinfo {author} {\bibfnamefont {J.}~\bibnamefont
  {Zhang}}, \bibinfo {author} {\bibfnamefont {P.~W.}\ \bibnamefont {Hess}},
  \bibinfo {author} {\bibfnamefont {A.}~\bibnamefont {Kyprianidis}}, \bibinfo
  {author} {\bibfnamefont {P.}~\bibnamefont {Becker}}, \bibinfo {author}
  {\bibfnamefont {A.}~\bibnamefont {Lee}}, \bibinfo {author} {\bibfnamefont
  {J.}~\bibnamefont {Smith}}, \bibinfo {author} {\bibfnamefont
  {G.}~\bibnamefont {Pagano}}, \bibinfo {author} {\bibfnamefont {I.-D.}\
  \bibnamefont {Potirniche}}, \bibinfo {author} {\bibfnamefont {A.~C.}\
  \bibnamefont {Potter}}, \bibinfo {author} {\bibfnamefont {A.}~\bibnamefont
  {Vishwanath}}, \bibinfo {author} {\bibfnamefont {N.~Y.}\ \bibnamefont
  {Yao}},\ and\ \bibinfo {author} {\bibfnamefont {C.}~\bibnamefont {Monroe}},\
  }\bibfield  {title} {\bibinfo {title} {Observation of a discrete time
  crystal},\ }\href {https://doi.org/10.1038/nature21413} {\bibfield  {journal}
  {\bibinfo  {journal} {Nature}\ }\textbf {\bibinfo {volume} {543}},\ \bibinfo
  {pages} {217} (\bibinfo {year} {2017})}\BibitemShut {NoStop}%
\bibitem [{\citenamefont {Autti}\ \emph {et~al.}(2018)\citenamefont {Autti},
  \citenamefont {Eltsov},\ and\ \citenamefont
  {Volovik}}]{m:auttiObservationTimeQuasicrystal2018}%
  \BibitemOpen
  \bibfield  {author} {\bibinfo {author} {\bibfnamefont {S.}~\bibnamefont
  {Autti}}, \bibinfo {author} {\bibfnamefont {V.~B.}\ \bibnamefont {Eltsov}},\
  and\ \bibinfo {author} {\bibfnamefont {G.~E.}\ \bibnamefont {Volovik}},\
  }\bibfield  {title} {\bibinfo {title} {Observation of a {{Time Quasicrystal}}
  and {{Its Transition}} to a {{Superfluid Time Crystal}}},\ }\href
  {https://doi.org/10.1103/PhysRevLett.120.215301} {\bibfield  {journal}
  {\bibinfo  {journal} {Phys. Rev. Lett.}\ }\textbf {\bibinfo {volume} {120}},\
  \bibinfo {pages} {215301} (\bibinfo {year} {2018})}\BibitemShut {NoStop}%
\bibitem [{\citenamefont {Rovny}\ \emph {et~al.}(2018)\citenamefont {Rovny},
  \citenamefont {Blum},\ and\ \citenamefont
  {Barrett}}]{m:rovnyObservationDiscreteTimeCrystalSignatures2018}%
  \BibitemOpen
  \bibfield  {author} {\bibinfo {author} {\bibfnamefont {J.}~\bibnamefont
  {Rovny}}, \bibinfo {author} {\bibfnamefont {R.~L.}\ \bibnamefont {Blum}},\
  and\ \bibinfo {author} {\bibfnamefont {S.~E.}\ \bibnamefont {Barrett}},\
  }\bibfield  {title} {\bibinfo {title} {Observation of {{Discrete-Time-Crystal
  Signatures}} in an {{Ordered Dipolar Many-Body System}}},\ }\href
  {https://doi.org/10.1103/PhysRevLett.120.180603} {\bibfield  {journal}
  {\bibinfo  {journal} {Phys. Rev. Lett.}\ }\textbf {\bibinfo {volume} {120}},\
  \bibinfo {pages} {180603} (\bibinfo {year} {2018})}\BibitemShut {NoStop}%
\bibitem [{\citenamefont {Smits}\ \emph {et~al.}(2018)\citenamefont {Smits},
  \citenamefont {Liao}, \citenamefont {Stoof},\ and\ \citenamefont {{van der
  Straten}}}]{m:smitsObservationSpaceTimeCrystal2018}%
  \BibitemOpen
  \bibfield  {author} {\bibinfo {author} {\bibfnamefont {J.}~\bibnamefont
  {Smits}}, \bibinfo {author} {\bibfnamefont {L.}~\bibnamefont {Liao}},
  \bibinfo {author} {\bibfnamefont {H.~T.~C.}\ \bibnamefont {Stoof}},\ and\
  \bibinfo {author} {\bibfnamefont {P.}~\bibnamefont {{van der Straten}}},\
  }\bibfield  {title} {\bibinfo {title} {Observation of a {{Space-Time
  Crystal}} in a {{Superfluid Quantum Gas}}},\ }\href
  {https://doi.org/10.1103/PhysRevLett.121.185301} {\bibfield  {journal}
  {\bibinfo  {journal} {Phys. Rev. Lett.}\ }\textbf {\bibinfo {volume} {121}},\
  \bibinfo {pages} {185301} (\bibinfo {year} {2018})}\BibitemShut {NoStop}%
\bibitem [{\citenamefont {O'Sullivan}\ \emph {et~al.}()\citenamefont
  {O'Sullivan}, \citenamefont {Lunt}, \citenamefont {Zollitsch}, \citenamefont
  {Thewalt}, \citenamefont {Morton},\ and\ \citenamefont
  {Pal}}]{m:osullivanDissipativeDiscreteTime2019}%
  \BibitemOpen
  \bibfield  {author} {\bibinfo {author} {\bibfnamefont {J.}~\bibnamefont
  {O'Sullivan}}, \bibinfo {author} {\bibfnamefont {O.}~\bibnamefont {Lunt}},
  \bibinfo {author} {\bibfnamefont {C.~W.}\ \bibnamefont {Zollitsch}}, \bibinfo
  {author} {\bibfnamefont {M.~L.~W.}\ \bibnamefont {Thewalt}}, \bibinfo
  {author} {\bibfnamefont {J.~J.~L.}\ \bibnamefont {Morton}},\ and\ \bibinfo
  {author} {\bibfnamefont {A.}~\bibnamefont {Pal}},\ }\bibfield  {title}
  {\bibinfo {title} {Dissipative discrete time crystals},\ }\href
  {http://arxiv.org/abs/1807.09884} {\ }\Eprint
  {https://arxiv.org/abs/1807.09884} {arXiv:1807.09884} \BibitemShut {NoStop}%
\bibitem [{\citenamefont {Kyprianidis}\ \emph {et~al.}(2021)\citenamefont
  {Kyprianidis}, \citenamefont {Machado}, \citenamefont {Morong}, \citenamefont
  {Becker}, \citenamefont {Collins}, \citenamefont {Else}, \citenamefont
  {Feng}, \citenamefont {Hess}, \citenamefont {Nayak}, \citenamefont {Pagano},
  \citenamefont {Yao},\ and\ \citenamefont
  {Monroe}}]{m:kyprianidisObservationPrethermalDiscrete2021}%
  \BibitemOpen
  \bibfield  {author} {\bibinfo {author} {\bibfnamefont {A.}~\bibnamefont
  {Kyprianidis}}, \bibinfo {author} {\bibfnamefont {F.}~\bibnamefont
  {Machado}}, \bibinfo {author} {\bibfnamefont {W.}~\bibnamefont {Morong}},
  \bibinfo {author} {\bibfnamefont {P.}~\bibnamefont {Becker}}, \bibinfo
  {author} {\bibfnamefont {K.~S.}\ \bibnamefont {Collins}}, \bibinfo {author}
  {\bibfnamefont {D.~V.}\ \bibnamefont {Else}}, \bibinfo {author}
  {\bibfnamefont {L.}~\bibnamefont {Feng}}, \bibinfo {author} {\bibfnamefont
  {P.~W.}\ \bibnamefont {Hess}}, \bibinfo {author} {\bibfnamefont
  {C.}~\bibnamefont {Nayak}}, \bibinfo {author} {\bibfnamefont
  {G.}~\bibnamefont {Pagano}}, \bibinfo {author} {\bibfnamefont {N.~Y.}\
  \bibnamefont {Yao}},\ and\ \bibinfo {author} {\bibfnamefont {C.}~\bibnamefont
  {Monroe}},\ }\bibfield  {title} {\bibinfo {title} {Observation of a
  prethermal discrete time crystal},\ }\href
  {https://doi.org/10.1126/science.abg8102} {\bibfield  {journal} {\bibinfo
  {journal} {Science}\ }\textbf {\bibinfo {volume} {372}},\ \bibinfo {pages}
  {1192} (\bibinfo {year} {2021})}\BibitemShut {NoStop}%
\bibitem [{\citenamefont {Randall}\ \emph {et~al.}(2021)\citenamefont
  {Randall}, \citenamefont {Bradley}, \citenamefont {{van der Gronden}},
  \citenamefont {Galicia}, \citenamefont {Abobeih}, \citenamefont {Markham},
  \citenamefont {Twitchen}, \citenamefont {Machado}, \citenamefont {Yao},\ and\
  \citenamefont
  {Taminiau}}]{m:randallObservationManybodylocalizedDiscrete2021}%
  \BibitemOpen
  \bibfield  {author} {\bibinfo {author} {\bibfnamefont {J.}~\bibnamefont
  {Randall}}, \bibinfo {author} {\bibfnamefont {C.~E.}\ \bibnamefont
  {Bradley}}, \bibinfo {author} {\bibfnamefont {F.~V.}\ \bibnamefont {{van der
  Gronden}}}, \bibinfo {author} {\bibfnamefont {A.}~\bibnamefont {Galicia}},
  \bibinfo {author} {\bibfnamefont {M.~H.}\ \bibnamefont {Abobeih}}, \bibinfo
  {author} {\bibfnamefont {M.}~\bibnamefont {Markham}}, \bibinfo {author}
  {\bibfnamefont {D.~J.}\ \bibnamefont {Twitchen}}, \bibinfo {author}
  {\bibfnamefont {F.}~\bibnamefont {Machado}}, \bibinfo {author} {\bibfnamefont
  {N.~Y.}\ \bibnamefont {Yao}},\ and\ \bibinfo {author} {\bibfnamefont {T.~H.}\
  \bibnamefont {Taminiau}},\ }\bibfield  {title} {\bibinfo {title} {Observation
  of a many-body-localized discrete time crystal with a programmable spin-based
  quantum simulator},\ }\href {https://doi.org/10.1126/science.abk0603}
  {\bibfield  {journal} {\bibinfo  {journal} {Science}\ }\textbf {\bibinfo
  {volume} {374}},\ \bibinfo {pages} {1474} (\bibinfo {year}
  {2021})}\BibitemShut {NoStop}%
\bibitem [{\citenamefont {Bernien}\ \emph {et~al.}(2017)\citenamefont
  {Bernien}, \citenamefont {Schwartz}, \citenamefont {Keesling}, \citenamefont
  {Levine}, \citenamefont {Omran}, \citenamefont {Pichler}, \citenamefont
  {Choi}, \citenamefont {Zibrov}, \citenamefont {Endres}, \citenamefont
  {Greiner}, \citenamefont {Vuleti{\'c}},\ and\ \citenamefont
  {Lukin}}]{m:bernienProbingManybodyDynamics2017}%
  \BibitemOpen
  \bibfield  {author} {\bibinfo {author} {\bibfnamefont {H.}~\bibnamefont
  {Bernien}}, \bibinfo {author} {\bibfnamefont {S.}~\bibnamefont {Schwartz}},
  \bibinfo {author} {\bibfnamefont {A.}~\bibnamefont {Keesling}}, \bibinfo
  {author} {\bibfnamefont {H.}~\bibnamefont {Levine}}, \bibinfo {author}
  {\bibfnamefont {A.}~\bibnamefont {Omran}}, \bibinfo {author} {\bibfnamefont
  {H.}~\bibnamefont {Pichler}}, \bibinfo {author} {\bibfnamefont
  {S.}~\bibnamefont {Choi}}, \bibinfo {author} {\bibfnamefont {A.~S.}\
  \bibnamefont {Zibrov}}, \bibinfo {author} {\bibfnamefont {M.}~\bibnamefont
  {Endres}}, \bibinfo {author} {\bibfnamefont {M.}~\bibnamefont {Greiner}},
  \bibinfo {author} {\bibfnamefont {V.}~\bibnamefont {Vuleti{\'c}}},\ and\
  \bibinfo {author} {\bibfnamefont {M.~D.}\ \bibnamefont {Lukin}},\ }\bibfield
  {title} {\bibinfo {title} {Probing many-body dynamics on a 51-atom quantum
  simulator},\ }\href {https://doi.org/10.1038/nature24622} {\bibfield
  {journal} {\bibinfo  {journal} {Nature}\ }\textbf {\bibinfo {volume} {551}},\
  \bibinfo {pages} {579} (\bibinfo {year} {2017})}\BibitemShut {NoStop}%
\bibitem [{\citenamefont {Ohliger}\ \emph {et~al.}(2013)\citenamefont
  {Ohliger}, \citenamefont {Nesme},\ and\ \citenamefont
  {Eisert}}]{m:ohligerEfficientFeasibleState2013}%
  \BibitemOpen
  \bibfield  {author} {\bibinfo {author} {\bibfnamefont {M.}~\bibnamefont
  {Ohliger}}, \bibinfo {author} {\bibfnamefont {V.}~\bibnamefont {Nesme}},\
  and\ \bibinfo {author} {\bibfnamefont {J.}~\bibnamefont {Eisert}},\
  }\bibfield  {title} {\bibinfo {title} {Efficient and feasible state
  tomography of quantum many-body systems},\ }\href
  {https://doi.org/10.1088/1367-2630/15/1/015024} {\bibfield  {journal}
  {\bibinfo  {journal} {New J. Phys.}\ }\textbf {\bibinfo {volume} {15}},\
  \bibinfo {pages} {015024} (\bibinfo {year} {2013})}\BibitemShut {NoStop}%
\bibitem [{\citenamefont {Islam}\ \emph {et~al.}(2015)\citenamefont {Islam},
  \citenamefont {Ma}, \citenamefont {Preiss}, \citenamefont {Eric~Tai},
  \citenamefont {Lukin}, \citenamefont {Rispoli},\ and\ \citenamefont
  {Greiner}}]{m:islamMeasuringEntanglementEntropy2015}%
  \BibitemOpen
  \bibfield  {author} {\bibinfo {author} {\bibfnamefont {R.}~\bibnamefont
  {Islam}}, \bibinfo {author} {\bibfnamefont {R.}~\bibnamefont {Ma}}, \bibinfo
  {author} {\bibfnamefont {P.~M.}\ \bibnamefont {Preiss}}, \bibinfo {author}
  {\bibfnamefont {M.}~\bibnamefont {Eric~Tai}}, \bibinfo {author}
  {\bibfnamefont {A.}~\bibnamefont {Lukin}}, \bibinfo {author} {\bibfnamefont
  {M.}~\bibnamefont {Rispoli}},\ and\ \bibinfo {author} {\bibfnamefont
  {M.}~\bibnamefont {Greiner}},\ }\bibfield  {title} {\bibinfo {title}
  {Measuring entanglement entropy in a quantum many-body system},\ }\href
  {https://doi.org/10.1038/nature15750} {\bibfield  {journal} {\bibinfo
  {journal} {Nature}\ }\textbf {\bibinfo {volume} {528}},\ \bibinfo {pages}
  {77} (\bibinfo {year} {2015})}\BibitemShut {NoStop}%
\bibitem [{\citenamefont {Pichler}\ \emph {et~al.}(2016)\citenamefont
  {Pichler}, \citenamefont {Zhu}, \citenamefont {Seif}, \citenamefont
  {Zoller},\ and\ \citenamefont
  {Hafezi}}]{m:pichlerMeasurementProtocolEntanglement2016}%
  \BibitemOpen
  \bibfield  {author} {\bibinfo {author} {\bibfnamefont {H.}~\bibnamefont
  {Pichler}}, \bibinfo {author} {\bibfnamefont {G.}~\bibnamefont {Zhu}},
  \bibinfo {author} {\bibfnamefont {A.}~\bibnamefont {Seif}}, \bibinfo {author}
  {\bibfnamefont {P.}~\bibnamefont {Zoller}},\ and\ \bibinfo {author}
  {\bibfnamefont {M.}~\bibnamefont {Hafezi}},\ }\bibfield  {title} {\bibinfo
  {title} {Measurement {{Protocol}} for the {{Entanglement Spectrum}} of {{Cold
  Atoms}}},\ }\href {https://doi.org/10.1103/PhysRevX.6.041033} {\bibfield
  {journal} {\bibinfo  {journal} {Phys. Rev. X}\ }\textbf {\bibinfo {volume}
  {6}},\ \bibinfo {pages} {041033} (\bibinfo {year} {2016})}\BibitemShut
  {NoStop}%
\bibitem [{\citenamefont {Brydges}\ \emph {et~al.}(2019)\citenamefont
  {Brydges}, \citenamefont {Elben}, \citenamefont {Jurcevic}, \citenamefont
  {Vermersch}, \citenamefont {Maier}, \citenamefont {Lanyon}, \citenamefont
  {Zoller}, \citenamefont {Blatt},\ and\ \citenamefont
  {Roos}}]{m:brydgesProbingRenyiEntanglement2019}%
  \BibitemOpen
  \bibfield  {author} {\bibinfo {author} {\bibfnamefont {T.}~\bibnamefont
  {Brydges}}, \bibinfo {author} {\bibfnamefont {A.}~\bibnamefont {Elben}},
  \bibinfo {author} {\bibfnamefont {P.}~\bibnamefont {Jurcevic}}, \bibinfo
  {author} {\bibfnamefont {B.}~\bibnamefont {Vermersch}}, \bibinfo {author}
  {\bibfnamefont {C.}~\bibnamefont {Maier}}, \bibinfo {author} {\bibfnamefont
  {B.~P.}\ \bibnamefont {Lanyon}}, \bibinfo {author} {\bibfnamefont
  {P.}~\bibnamefont {Zoller}}, \bibinfo {author} {\bibfnamefont
  {R.}~\bibnamefont {Blatt}},\ and\ \bibinfo {author} {\bibfnamefont {C.~F.}\
  \bibnamefont {Roos}},\ }\bibfield  {title} {\bibinfo {title} {Probing
  {{R\'enyi}} entanglement entropy via randomized measurements},\ }\href
  {https://doi.org/10.1126/science.aau4963} {\bibfield  {journal} {\bibinfo
  {journal} {Science}\ }\textbf {\bibinfo {volume} {364}},\ \bibinfo {pages}
  {260} (\bibinfo {year} {2019})}\BibitemShut {NoStop}%
\bibitem [{\citenamefont {Huang}\ \emph {et~al.}(2020)\citenamefont {Huang},
  \citenamefont {Kueng},\ and\ \citenamefont
  {Preskill}}]{m:huangPredictingManyProperties2020}%
  \BibitemOpen
  \bibfield  {author} {\bibinfo {author} {\bibfnamefont {H.-Y.}\ \bibnamefont
  {Huang}}, \bibinfo {author} {\bibfnamefont {R.}~\bibnamefont {Kueng}},\ and\
  \bibinfo {author} {\bibfnamefont {J.}~\bibnamefont {Preskill}},\ }\bibfield
  {title} {\bibinfo {title} {Predicting many properties of a quantum system
  from very few measurements},\ }\href
  {https://doi.org/10.1038/s41567-020-0932-7} {\bibfield  {journal} {\bibinfo
  {journal} {Nat. Phys.}\ }\textbf {\bibinfo {volume} {16}},\ \bibinfo {pages}
  {1050} (\bibinfo {year} {2020})}\BibitemShut {NoStop}%
\bibitem [{\citenamefont {Elben}\ \emph {et~al.}()\citenamefont {Elben},
  \citenamefont {Flammia}, \citenamefont {Huang}, \citenamefont {Kueng},
  \citenamefont {Preskill}, \citenamefont {Vermersch},\ and\ \citenamefont
  {Zoller}}]{m:elbenRandomizedMeasurementToolbox2022}%
  \BibitemOpen
  \bibfield  {author} {\bibinfo {author} {\bibfnamefont {A.}~\bibnamefont
  {Elben}}, \bibinfo {author} {\bibfnamefont {S.~T.}\ \bibnamefont {Flammia}},
  \bibinfo {author} {\bibfnamefont {H.-Y.}\ \bibnamefont {Huang}}, \bibinfo
  {author} {\bibfnamefont {R.}~\bibnamefont {Kueng}}, \bibinfo {author}
  {\bibfnamefont {J.}~\bibnamefont {Preskill}}, \bibinfo {author}
  {\bibfnamefont {B.}~\bibnamefont {Vermersch}},\ and\ \bibinfo {author}
  {\bibfnamefont {P.}~\bibnamefont {Zoller}},\ }\bibfield  {title} {\bibinfo
  {title} {The randomized measurement toolbox},\ }\href
  {http://arxiv.org/abs/2203.11374} {\ }\Eprint
  {https://arxiv.org/abs/2203.11374} {arXiv:2203.11374} \BibitemShut {NoStop}%
\bibitem [{\citenamefont {Boixo}\ \emph {et~al.}(2018)\citenamefont {Boixo},
  \citenamefont {Isakov}, \citenamefont {Smelyanskiy}, \citenamefont {Babbush},
  \citenamefont {Ding}, \citenamefont {Jiang}, \citenamefont {Bremner},
  \citenamefont {Martinis},\ and\ \citenamefont
  {Neven}}]{m:boixoCharacterizingQuantumSupremacy2018}%
  \BibitemOpen
  \bibfield  {author} {\bibinfo {author} {\bibfnamefont {S.}~\bibnamefont
  {Boixo}}, \bibinfo {author} {\bibfnamefont {S.~V.}\ \bibnamefont {Isakov}},
  \bibinfo {author} {\bibfnamefont {V.~N.}\ \bibnamefont {Smelyanskiy}},
  \bibinfo {author} {\bibfnamefont {R.}~\bibnamefont {Babbush}}, \bibinfo
  {author} {\bibfnamefont {N.}~\bibnamefont {Ding}}, \bibinfo {author}
  {\bibfnamefont {Z.}~\bibnamefont {Jiang}}, \bibinfo {author} {\bibfnamefont
  {M.~J.}\ \bibnamefont {Bremner}}, \bibinfo {author} {\bibfnamefont {J.~M.}\
  \bibnamefont {Martinis}},\ and\ \bibinfo {author} {\bibfnamefont
  {H.}~\bibnamefont {Neven}},\ }\bibfield  {title} {\bibinfo {title}
  {Characterizing quantum supremacy in near-term devices},\ }\href
  {https://doi.org/10.1038/s41567-018-0124-x} {\bibfield  {journal} {\bibinfo
  {journal} {Nat. Phys.}\ }\textbf {\bibinfo {volume} {14}},\ \bibinfo {pages}
  {595} (\bibinfo {year} {2018})}\BibitemShut {NoStop}%
\bibitem [{\citenamefont {Choi}\ \emph {et~al.}()\citenamefont {Choi},
  \citenamefont {Shaw}, \citenamefont {Madjarov}, \citenamefont {Xie},
  \citenamefont {Finkelstein}, \citenamefont {Covey}, \citenamefont {Cotler},
  \citenamefont {Mark}, \citenamefont {Huang}, \citenamefont {Kale},
  \citenamefont {Pichler}, \citenamefont {Brand{\~a}o}, \citenamefont {Choi},\
  and\ \citenamefont {Endres}}]{m:choiEmergentQuantumRandomness2022}%
  \BibitemOpen
  \bibfield  {author} {\bibinfo {author} {\bibfnamefont {J.}~\bibnamefont
  {Choi}}, \bibinfo {author} {\bibfnamefont {A.~L.}\ \bibnamefont {Shaw}},
  \bibinfo {author} {\bibfnamefont {I.~S.}\ \bibnamefont {Madjarov}}, \bibinfo
  {author} {\bibfnamefont {X.}~\bibnamefont {Xie}}, \bibinfo {author}
  {\bibfnamefont {R.}~\bibnamefont {Finkelstein}}, \bibinfo {author}
  {\bibfnamefont {J.~P.}\ \bibnamefont {Covey}}, \bibinfo {author}
  {\bibfnamefont {J.~S.}\ \bibnamefont {Cotler}}, \bibinfo {author}
  {\bibfnamefont {D.~K.}\ \bibnamefont {Mark}}, \bibinfo {author}
  {\bibfnamefont {H.-Y.}\ \bibnamefont {Huang}}, \bibinfo {author}
  {\bibfnamefont {A.}~\bibnamefont {Kale}}, \bibinfo {author} {\bibfnamefont
  {H.}~\bibnamefont {Pichler}}, \bibinfo {author} {\bibfnamefont {F.~G. S.~L.}\
  \bibnamefont {Brand{\~a}o}}, \bibinfo {author} {\bibfnamefont
  {S.}~\bibnamefont {Choi}},\ and\ \bibinfo {author} {\bibfnamefont
  {M.}~\bibnamefont {Endres}},\ }\bibfield  {title} {\bibinfo {title} {Emergent
  {{Quantum Randomness}} and {{Benchmarking}} from {{Hamiltonian Many-body
  Dynamics}}},\ }\href {http://arxiv.org/abs/2103.03535} {\ }\Eprint
  {https://arxiv.org/abs/2103.03535} {arXiv:2103.03535} \BibitemShut {NoStop}%
\bibitem [{\citenamefont {Mark}\ \emph {et~al.}()\citenamefont {Mark},
  \citenamefont {Choi}, \citenamefont {Shaw}, \citenamefont {Endres},\ and\
  \citenamefont {Choi}}]{m:markBenchmarkingQuantumSimulators2022}%
  \BibitemOpen
  \bibfield  {author} {\bibinfo {author} {\bibfnamefont {D.~K.}\ \bibnamefont
  {Mark}}, \bibinfo {author} {\bibfnamefont {J.}~\bibnamefont {Choi}}, \bibinfo
  {author} {\bibfnamefont {A.~L.}\ \bibnamefont {Shaw}}, \bibinfo {author}
  {\bibfnamefont {M.}~\bibnamefont {Endres}},\ and\ \bibinfo {author}
  {\bibfnamefont {S.}~\bibnamefont {Choi}},\ }\bibfield  {title} {\bibinfo
  {title} {Benchmarking {{Quantum Simulators}} using {{Quantum Chaos}}},\
  }\href {http://arxiv.org/abs/2205.12211} {\ }\Eprint
  {https://arxiv.org/abs/2205.12211} {arXiv:2205.12211} \BibitemShut {NoStop}%
\bibitem [{\citenamefont {Cooper}(2020)}]{m:cooperFractionalQuantumHall2020}%
  \BibitemOpen
  \bibfield  {author} {\bibinfo {author} {\bibfnamefont {N.~R.}\ \bibnamefont
  {Cooper}},\ }\bibfield  {title} {\bibinfo {title} {Fractional {{Quantum Hall
  States}} of {{Bosons}}: {{Properties}} and {{Prospects}} for {{Experimental
  Realization}}},\ }in\ \href {https://doi.org/10.1142/9789811217494_0010}
  {\emph {\bibinfo {booktitle} {Fractional {{Quantum Hall Effects}}}}}\
  (\bibinfo  {publisher} {{World Scientific}},\ \bibinfo {year} {2020})\ pp.\
  \bibinfo {pages} {487--521}\BibitemShut {NoStop}%
\bibitem [{\citenamefont {Tai}\ \emph {et~al.}(2017)\citenamefont {Tai},
  \citenamefont {Lukin}, \citenamefont {Rispoli}, \citenamefont {Schittko},
  \citenamefont {Menke}, \citenamefont {{Dan Borgnia}}, \citenamefont {Preiss},
  \citenamefont {Grusdt}, \citenamefont {Kaufman},\ and\ \citenamefont
  {Greiner}}]{m:taiMicroscopyInteractingHarper2017}%
  \BibitemOpen
  \bibfield  {author} {\bibinfo {author} {\bibfnamefont {M.~E.}\ \bibnamefont
  {Tai}}, \bibinfo {author} {\bibfnamefont {A.}~\bibnamefont {Lukin}}, \bibinfo
  {author} {\bibfnamefont {M.}~\bibnamefont {Rispoli}}, \bibinfo {author}
  {\bibfnamefont {R.}~\bibnamefont {Schittko}}, \bibinfo {author}
  {\bibfnamefont {T.}~\bibnamefont {Menke}}, \bibinfo {author} {\bibnamefont
  {{Dan Borgnia}}}, \bibinfo {author} {\bibfnamefont {P.~M.}\ \bibnamefont
  {Preiss}}, \bibinfo {author} {\bibfnamefont {F.}~\bibnamefont {Grusdt}},
  \bibinfo {author} {\bibfnamefont {A.~M.}\ \bibnamefont {Kaufman}},\ and\
  \bibinfo {author} {\bibfnamefont {M.}~\bibnamefont {Greiner}},\ }\bibfield
  {title} {\bibinfo {title} {Microscopy of the interacting
  {{Harper}}\textendash{{Hofstadter}} model in the two-body limit},\ }\href
  {https://doi.org/10.1038/nature22811} {\bibfield  {journal} {\bibinfo
  {journal} {Nature}\ }\textbf {\bibinfo {volume} {546}},\ \bibinfo {pages}
  {519} (\bibinfo {year} {2017})}\BibitemShut {NoStop}%
\bibitem [{\citenamefont {Hall}\ \emph {et~al.}(1998)\citenamefont {Hall},
  \citenamefont {Matthews}, \citenamefont {Ensher}, \citenamefont {Wieman},\
  and\ \citenamefont {Cornell}}]{m:hallDynamicsComponentSeparation1998}%
  \BibitemOpen
  \bibfield  {author} {\bibinfo {author} {\bibfnamefont {D.~S.}\ \bibnamefont
  {Hall}}, \bibinfo {author} {\bibfnamefont {M.~R.}\ \bibnamefont {Matthews}},
  \bibinfo {author} {\bibfnamefont {J.~R.}\ \bibnamefont {Ensher}}, \bibinfo
  {author} {\bibfnamefont {C.~E.}\ \bibnamefont {Wieman}},\ and\ \bibinfo
  {author} {\bibfnamefont {E.~A.}\ \bibnamefont {Cornell}},\ }\bibfield
  {title} {\bibinfo {title} {Dynamics of {{Component Separation}} in a {{Binary
  Mixture}} of {{Bose-Einstein Condensates}}},\ }\href
  {https://doi.org/10.1103/PhysRevLett.81.1539} {\bibfield  {journal} {\bibinfo
   {journal} {Phys. Rev. Lett.}\ }\textbf {\bibinfo {volume} {81}},\ \bibinfo
  {pages} {1539} (\bibinfo {year} {1998})}\BibitemShut {NoStop}%
\bibitem [{\citenamefont {Ippoliti}\ and\ \citenamefont
  {Ho}()}]{mm:ippolitiSolvableModelDeep2022}%
  \BibitemOpen
  \bibfield  {author} {\bibinfo {author} {\bibfnamefont {M.}~\bibnamefont
  {Ippoliti}}\ and\ \bibinfo {author} {\bibfnamefont {W.~W.}\ \bibnamefont
  {Ho}},\ }\bibfield  {title} {\bibinfo {title} {Solvable model of deep
  thermalization with distinct design times},\ }\href
  {http://arxiv.org/abs/2208.10542} {\ }\Eprint
  {https://arxiv.org/abs/2208.10542} {arXiv:2208.10542} \BibitemShut {NoStop}%
\bibitem [{\citenamefont {Cotler}\ \emph {et~al.}()\citenamefont {Cotler},
  \citenamefont {Mark}, \citenamefont {Huang}, \citenamefont {Hernandez},
  \citenamefont {Choi}, \citenamefont {Shaw}, \citenamefont {Endres},\ and\
  \citenamefont {Choi}}]{m:cotlerEmergentQuantumState2021}%
  \BibitemOpen
  \bibfield  {author} {\bibinfo {author} {\bibfnamefont {J.~S.}\ \bibnamefont
  {Cotler}}, \bibinfo {author} {\bibfnamefont {D.~K.}\ \bibnamefont {Mark}},
  \bibinfo {author} {\bibfnamefont {H.-Y.}\ \bibnamefont {Huang}}, \bibinfo
  {author} {\bibfnamefont {F.}~\bibnamefont {Hernandez}}, \bibinfo {author}
  {\bibfnamefont {J.}~\bibnamefont {Choi}}, \bibinfo {author} {\bibfnamefont
  {A.~L.}\ \bibnamefont {Shaw}}, \bibinfo {author} {\bibfnamefont
  {M.}~\bibnamefont {Endres}},\ and\ \bibinfo {author} {\bibfnamefont
  {S.}~\bibnamefont {Choi}},\ }\bibfield  {title} {\bibinfo {title} {Emergent
  quantum state designs from individual many-body wavefunctions},\ }\href
  {http://arxiv.org/abs/2103.03536} {\ }\Eprint
  {https://arxiv.org/abs/2103.03536} {arXiv:2103.03536} \BibitemShut {NoStop}%
\bibitem [{\citenamefont {Ho}\ and\ \citenamefont
  {Choi}(2022)}]{m:hoExactEmergentQuantum2022}%
  \BibitemOpen
  \bibfield  {author} {\bibinfo {author} {\bibfnamefont {W.~W.}\ \bibnamefont
  {Ho}}\ and\ \bibinfo {author} {\bibfnamefont {S.}~\bibnamefont {Choi}},\
  }\bibfield  {title} {\bibinfo {title} {Exact emergent quantum state designs
  from quantum chaotic dynamics},\ }\href
  {https://doi.org/10.1103/PhysRevLett.128.060601} {\bibfield  {journal}
  {\bibinfo  {journal} {Phys. Rev. Lett.}\ }\textbf {\bibinfo {volume} {128}},\
  \bibinfo {pages} {060601} (\bibinfo {year} {2022})}\BibitemShut {NoStop}%
\bibitem [{\citenamefont {Wilming}\ and\ \citenamefont
  {Roth}()}]{m:wilmingHightemperatureThermalizationImplies2022}%
  \BibitemOpen
  \bibfield  {author} {\bibinfo {author} {\bibfnamefont {H.}~\bibnamefont
  {Wilming}}\ and\ \bibinfo {author} {\bibfnamefont {I.}~\bibnamefont {Roth}},\
  }\bibfield  {title} {\bibinfo {title} {High-temperature thermalization
  implies the emergence of quantum state designs},\ }\href
  {http://arxiv.org/abs/2202.01669} {\ }\Eprint
  {https://arxiv.org/abs/2202.01669} {arXiv:2202.01669} \BibitemShut {NoStop}%
\bibitem [{\citenamefont {Claeys}\ and\ \citenamefont
  {Lamacraft}(2022)}]{m:claeysEmergentQuantumState2022}%
  \BibitemOpen
  \bibfield  {author} {\bibinfo {author} {\bibfnamefont {P.~W.}\ \bibnamefont
  {Claeys}}\ and\ \bibinfo {author} {\bibfnamefont {A.}~\bibnamefont
  {Lamacraft}},\ }\bibfield  {title} {\bibinfo {title} {Emergent quantum state
  designs and biunitarity in dual-unitary circuit dynamics},\ }\href
  {https://doi.org/10.22331/q-2022-06-15-738} {\bibfield  {journal} {\bibinfo
  {journal} {Quantum}\ }\textbf {\bibinfo {volume} {6}},\ \bibinfo {pages}
  {738} (\bibinfo {year} {2022})}\BibitemShut {NoStop}%
\bibitem [{\citenamefont {Ippoliti\natexlab{b}}\ and\ \citenamefont
  {Ho}()}]{mm:ippolitiDynamicalPurificationEmergence2022}%
  \BibitemOpen
  \bibfield  {author} {\bibinfo {author} {\bibfnamefont {M.}~\bibnamefont
  {Ippoliti\natexlab{b}}}\ and\ \bibinfo {author} {\bibfnamefont {W.~W.}\
  \bibnamefont {Ho}},\ }\bibfield  {title} {\bibinfo {title} {Dynamical
  purification and the emergence of quantum state designs from the projected
  ensemble},\ }\href {http://arxiv.org/abs/2204.13657} {\ }\Eprint
  {https://arxiv.org/abs/2204.13657} {arXiv:2204.13657} \BibitemShut {NoStop}%
\bibitem [{\citenamefont
  {Busch}(1991)}]{m:buschInformationallyCompleteSets1991}%
  \BibitemOpen
  \bibfield  {author} {\bibinfo {author} {\bibfnamefont {P.}~\bibnamefont
  {Busch}},\ }\bibfield  {title} {\bibinfo {title} {Informationally complete
  sets of physical quantities},\ }\href {https://doi.org/10.1007/BF00671008}
  {\bibfield  {journal} {\bibinfo  {journal} {Int. J. Theor. Phys.}\ }\textbf
  {\bibinfo {volume} {30}},\ \bibinfo {pages} {1217} (\bibinfo {year}
  {1991})}\BibitemShut {NoStop}%
\bibitem [{\citenamefont
  {DeBrota}(2020)}]{m:debrotaInformationallyCompleteMeasurements2020}%
  \BibitemOpen
  \bibfield  {author} {\bibinfo {author} {\bibfnamefont {J.}~\bibnamefont
  {DeBrota}},\ }\bibfield  {title} {\bibinfo {title} {Informationally
  {{Complete Measurements}} and {{Optimal Representations}} of {{Quantum
  Theory}}},\ }\href {https://scholarworks.umb.edu/doctoral_dissertations/617}
  {\bibfield  {journal} {\bibinfo  {journal} {Graduate Doctoral Dissertations}\
  } (\bibinfo {year} {2020})}\BibitemShut {NoStop}%
\bibitem [{Note1()}]{Note1}%
  \BibitemOpen
  \bibinfo {note} {A set of generalized measurements is specified by a set of a
  positive, semi-definite operators $\{ E_i \}_{i=1}^N$ which sum to the
  identity: $\DOTSB \sum@ \slimits@ _{i=1}^N E_i = \protect \mathbb {I}$, such
  that outcome $i$ occurs with probability $p_i = \Tr (E_i \rho )$. This set is
  also known as a positive operator-valued measure (POVM). It is a fact in
  quantum state tomography that a POVM requires at least $N = d_\protect \text
  {sys}^2$ elements for $\rho $ to be reconstructible from the statistics
  $p_i$. When $\rho $ is reconstructible, the POVM is called {\protect \it
  informationally complete} ({\protect \it minimally informationally-complete}
  if the number of elements $N$ is exactly $d_\protect \text {sys}^2$). Our
  protocol can be equivalently cast in this language upon identifying $E_{s,a}
  = \protect \tilde U(a)^\protect \dag \ketbra {s} \protect \tilde {U}(a)$,
  immediately yielding the claimed requirement $d_\protect \text {anc} \geq
  d_\protect \text {sys}$.}\BibitemShut {Stop}%
\bibitem [{\citenamefont {Bluvstein}\ \emph {et~al.}(2022)\citenamefont
  {Bluvstein}, \citenamefont {Levine}, \citenamefont {Semeghini}, \citenamefont
  {Wang}, \citenamefont {Ebadi}, \citenamefont {Kalinowski}, \citenamefont
  {Keesling}, \citenamefont {Maskara}, \citenamefont {Pichler}, \citenamefont
  {Greiner}, \citenamefont {Vuleti{\'c}},\ and\ \citenamefont
  {Lukin}}]{m:bluvsteinQuantumProcessorBased2022}%
  \BibitemOpen
  \bibfield  {author} {\bibinfo {author} {\bibfnamefont {D.}~\bibnamefont
  {Bluvstein}}, \bibinfo {author} {\bibfnamefont {H.}~\bibnamefont {Levine}},
  \bibinfo {author} {\bibfnamefont {G.}~\bibnamefont {Semeghini}}, \bibinfo
  {author} {\bibfnamefont {T.~T.}\ \bibnamefont {Wang}}, \bibinfo {author}
  {\bibfnamefont {S.}~\bibnamefont {Ebadi}}, \bibinfo {author} {\bibfnamefont
  {M.}~\bibnamefont {Kalinowski}}, \bibinfo {author} {\bibfnamefont
  {A.}~\bibnamefont {Keesling}}, \bibinfo {author} {\bibfnamefont
  {N.}~\bibnamefont {Maskara}}, \bibinfo {author} {\bibfnamefont
  {H.}~\bibnamefont {Pichler}}, \bibinfo {author} {\bibfnamefont
  {M.}~\bibnamefont {Greiner}}, \bibinfo {author} {\bibfnamefont
  {V.}~\bibnamefont {Vuleti{\'c}}},\ and\ \bibinfo {author} {\bibfnamefont
  {M.~D.}\ \bibnamefont {Lukin}},\ }\bibfield  {title} {\bibinfo {title} {A
  quantum processor based on coherent transport of entangled atom arrays},\
  }\href {https://doi.org/10.1038/s41586-022-04592-6} {\bibfield  {journal}
  {\bibinfo  {journal} {Nature}\ }\textbf {\bibinfo {volume} {604}},\ \bibinfo
  {pages} {451} (\bibinfo {year} {2022})}\BibitemShut {NoStop}%
\bibitem [{\citenamefont {Allcock}\ \emph {et~al.}(2021)\citenamefont
  {Allcock}, \citenamefont {Campbell}, \citenamefont {Chiaverini},
  \citenamefont {Chuang}, \citenamefont {Hudson}, \citenamefont {Moore},
  \citenamefont {Ransford}, \citenamefont {Roman}, \citenamefont {Sage},\ and\
  \citenamefont {Wineland}}]{m:allcockOmgBlueprintTrapped2021}%
  \BibitemOpen
  \bibfield  {author} {\bibinfo {author} {\bibfnamefont {D.~T.~C.}\
  \bibnamefont {Allcock}}, \bibinfo {author} {\bibfnamefont {W.~C.}\
  \bibnamefont {Campbell}}, \bibinfo {author} {\bibfnamefont {J.}~\bibnamefont
  {Chiaverini}}, \bibinfo {author} {\bibfnamefont {I.~L.}\ \bibnamefont
  {Chuang}}, \bibinfo {author} {\bibfnamefont {E.~R.}\ \bibnamefont {Hudson}},
  \bibinfo {author} {\bibfnamefont {I.~D.}\ \bibnamefont {Moore}}, \bibinfo
  {author} {\bibfnamefont {A.}~\bibnamefont {Ransford}}, \bibinfo {author}
  {\bibfnamefont {C.}~\bibnamefont {Roman}}, \bibinfo {author} {\bibfnamefont
  {J.~M.}\ \bibnamefont {Sage}},\ and\ \bibinfo {author} {\bibfnamefont
  {D.~J.}\ \bibnamefont {Wineland}},\ }\bibfield  {title} {\bibinfo {title}
  {Omg blueprint for trapped ion quantum computing with metastable states},\
  }\href {https://doi.org/10.1063/5.0069544} {\bibfield  {journal} {\bibinfo
  {journal} {Appl. Phys. Lett.}\ }\textbf {\bibinfo {volume} {119}},\ \bibinfo
  {pages} {214002} (\bibinfo {year} {2021})}\BibitemShut {NoStop}%
\bibitem [{\citenamefont {Chen}\ \emph {et~al.}(2022)\citenamefont {Chen},
  \citenamefont {Li}, \citenamefont {Huie}, \citenamefont {Zhao}, \citenamefont
  {Vetter}, \citenamefont {Greene},\ and\ \citenamefont
  {Covey}}]{m:chenAnalyzingRydbergbasedOpticalmetastableground2022}%
  \BibitemOpen
  \bibfield  {author} {\bibinfo {author} {\bibfnamefont {N.}~\bibnamefont
  {Chen}}, \bibinfo {author} {\bibfnamefont {L.}~\bibnamefont {Li}}, \bibinfo
  {author} {\bibfnamefont {W.}~\bibnamefont {Huie}}, \bibinfo {author}
  {\bibfnamefont {M.}~\bibnamefont {Zhao}}, \bibinfo {author} {\bibfnamefont
  {I.}~\bibnamefont {Vetter}}, \bibinfo {author} {\bibfnamefont {C.~H.}\
  \bibnamefont {Greene}},\ and\ \bibinfo {author} {\bibfnamefont {J.~P.}\
  \bibnamefont {Covey}},\ }\bibfield  {title} {\bibinfo {title} {Analyzing the
  rydberg-based optical-metastable-ground architecture for $^{171}\mathrm{Yb}$
  nuclear spins},\ }\href {https://doi.org/10.1103/PhysRevA.105.052438}
  {\bibfield  {journal} {\bibinfo  {journal} {Phys. Rev. A}\ }\textbf {\bibinfo
  {volume} {105}},\ \bibinfo {pages} {052438} (\bibinfo {year}
  {2022})}\BibitemShut {NoStop}%
\bibitem [{\citenamefont {Wu}\ \emph {et~al.}(2022)\citenamefont {Wu},
  \citenamefont {Kolkowitz}, \citenamefont {Puri},\ and\ \citenamefont
  {Thompson}}]{m:wuErasureConversionFaulttolerant2022}%
  \BibitemOpen
  \bibfield  {author} {\bibinfo {author} {\bibfnamefont {Y.}~\bibnamefont
  {Wu}}, \bibinfo {author} {\bibfnamefont {S.}~\bibnamefont {Kolkowitz}},
  \bibinfo {author} {\bibfnamefont {S.}~\bibnamefont {Puri}},\ and\ \bibinfo
  {author} {\bibfnamefont {J.~D.}\ \bibnamefont {Thompson}},\ }\bibfield
  {title} {\bibinfo {title} {Erasure conversion for fault-tolerant quantum
  computing in alkaline earth {{Rydberg}} atom arrays},\ }\href
  {https://doi.org/10.1038/s41467-022-32094-6} {\bibfield  {journal} {\bibinfo
  {journal} {Nat. Commun.}\ }\textbf {\bibinfo {volume} {13}},\ \bibinfo
  {pages} {4657} (\bibinfo {year} {2022})}\BibitemShut {NoStop}%
\bibitem [{\citenamefont {Stricker}\ \emph {et~al.}(2022)\citenamefont
  {Stricker}, \citenamefont {Meth}, \citenamefont {Postler}, \citenamefont
  {Edmunds}, \citenamefont {Ferrie}, \citenamefont {Blatt}, \citenamefont
  {Schindler}, \citenamefont {Monz}, \citenamefont {Kueng},\ and\ \citenamefont
  {Ringbauer}}]{m:strickerExperimentalSinglesettingQuantum2022}%
  \BibitemOpen
  \bibfield  {author} {\bibinfo {author} {\bibfnamefont {R.}~\bibnamefont
  {Stricker}}, \bibinfo {author} {\bibfnamefont {M.}~\bibnamefont {Meth}},
  \bibinfo {author} {\bibfnamefont {L.}~\bibnamefont {Postler}}, \bibinfo
  {author} {\bibfnamefont {C.}~\bibnamefont {Edmunds}}, \bibinfo {author}
  {\bibfnamefont {C.}~\bibnamefont {Ferrie}}, \bibinfo {author} {\bibfnamefont
  {R.}~\bibnamefont {Blatt}}, \bibinfo {author} {\bibfnamefont
  {P.}~\bibnamefont {Schindler}}, \bibinfo {author} {\bibfnamefont
  {T.}~\bibnamefont {Monz}}, \bibinfo {author} {\bibfnamefont {R.}~\bibnamefont
  {Kueng}},\ and\ \bibinfo {author} {\bibfnamefont {M.}~\bibnamefont
  {Ringbauer}},\ }\bibfield  {title} {\bibinfo {title} {Experimental
  {{Single-Setting Quantum State Tomography}}},\ }\href
  {https://doi.org/10.1103/PRXQuantum.3.040310} {\bibfield  {journal} {\bibinfo
   {journal} {PRX Quantum}\ }\textbf {\bibinfo {volume} {3}},\ \bibinfo {pages}
  {040310} (\bibinfo {year} {2022})}\BibitemShut {NoStop}%
\bibitem [{\citenamefont {Hu}\ \emph {et~al.}()\citenamefont {Hu},
  \citenamefont {Choi},\ and\ \citenamefont
  {You}}]{m:huClassicalShadowTomography2021}%
  \BibitemOpen
  \bibfield  {author} {\bibinfo {author} {\bibfnamefont {H.-Y.}\ \bibnamefont
  {Hu}}, \bibinfo {author} {\bibfnamefont {S.}~\bibnamefont {Choi}},\ and\
  \bibinfo {author} {\bibfnamefont {Y.-Z.}\ \bibnamefont {You}},\ }\bibfield
  {title} {\bibinfo {title} {Classical {{Shadow Tomography}} with {{Locally
  Scrambled Quantum Dynamics}}},\ }\href {http://arxiv.org/abs/2107.04817} {\
  }\Eprint {https://arxiv.org/abs/2107.04817} {arXiv:2107.04817} \BibitemShut
  {NoStop}%
\bibitem [{\citenamefont {Haah}\ \emph {et~al.}(2015)\citenamefont {Haah},
  \citenamefont {Harrow}, \citenamefont {Ji}, \citenamefont {Wu},\ and\
  \citenamefont {Yu}}]{m:haahSampleOptimalTomographyQuantum2015}%
  \BibitemOpen
  \bibfield  {author} {\bibinfo {author} {\bibfnamefont {J.}~\bibnamefont
  {Haah}}, \bibinfo {author} {\bibfnamefont {A.~W.}\ \bibnamefont {Harrow}},
  \bibinfo {author} {\bibfnamefont {Z.}~\bibnamefont {Ji}}, \bibinfo {author}
  {\bibfnamefont {X.}~\bibnamefont {Wu}},\ and\ \bibinfo {author}
  {\bibfnamefont {N.}~\bibnamefont {Yu}},\ }\bibfield  {title} {\bibinfo
  {title} {Sample-{{Optimal Tomography}} of {{Quantum States}}},\ }\href
  {https://doi.org/10.1109/TIT.2017.2719044} {\bibfield  {journal} {\bibinfo
  {journal} {IEEE Trans. Inf. Theory}\ }\textbf {\bibinfo {volume} {63}},\
  \bibinfo {pages} {5628} (\bibinfo {year} {2015})}\BibitemShut {NoStop}%
\bibitem [{\citenamefont
  {Hoeffding}(1948)}]{m:hoeffdingClassStatisticsAsymptotically1948}%
  \BibitemOpen
  \bibfield  {author} {\bibinfo {author} {\bibfnamefont {W.}~\bibnamefont
  {Hoeffding}},\ }\bibfield  {title} {\bibinfo {title} {A {{Class}} of
  {{Statistics}} with {{Asymptotically Normal Distribution}}},\ }\href
  {https://doi.org/10.1214/aoms/1177730196} {\bibfield  {journal} {\bibinfo
  {journal} {Ann. Math. Stat.}\ }\textbf {\bibinfo {volume} {19}},\ \bibinfo
  {pages} {293} (\bibinfo {year} {1948})}\BibitemShut {NoStop}%
\bibitem [{\citenamefont {Nahum}\ \emph {et~al.}(2018)\citenamefont {Nahum},
  \citenamefont {Vijay},\ and\ \citenamefont
  {Haah}}]{m:nahumOperatorSpreadingRandom2018}%
  \BibitemOpen
  \bibfield  {author} {\bibinfo {author} {\bibfnamefont {A.}~\bibnamefont
  {Nahum}}, \bibinfo {author} {\bibfnamefont {S.}~\bibnamefont {Vijay}},\ and\
  \bibinfo {author} {\bibfnamefont {J.}~\bibnamefont {Haah}},\ }\bibfield
  {title} {\bibinfo {title} {Operator {{Spreading}} in {{Random Unitary
  Circuits}}},\ }\href {https://doi.org/10.1103/PhysRevX.8.021014} {\bibfield
  {journal} {\bibinfo  {journal} {Phys. Rev. X}\ }\textbf {\bibinfo {volume}
  {8}},\ \bibinfo {pages} {021014} (\bibinfo {year} {2018})}\BibitemShut
  {NoStop}%
\bibitem [{\citenamefont {{von Keyserlingk}}\ \emph {et~al.}(2018)\citenamefont
  {{von Keyserlingk}}, \citenamefont {Rakovszky}, \citenamefont {Pollmann},\
  and\ \citenamefont
  {Sondhi}}]{m:vonkeyserlingkOperatorHydrodynamicsOTOCs2018}%
  \BibitemOpen
  \bibfield  {author} {\bibinfo {author} {\bibfnamefont {C.~W.}\ \bibnamefont
  {{von Keyserlingk}}}, \bibinfo {author} {\bibfnamefont {T.}~\bibnamefont
  {Rakovszky}}, \bibinfo {author} {\bibfnamefont {F.}~\bibnamefont
  {Pollmann}},\ and\ \bibinfo {author} {\bibfnamefont {S.~L.}\ \bibnamefont
  {Sondhi}},\ }\bibfield  {title} {\bibinfo {title} {Operator
  {{Hydrodynamics}}, {{OTOCs}}, and {{Entanglement Growth}} in {{Systems}}
  without {{Conservation Laws}}},\ }\href
  {https://doi.org/10.1103/PhysRevX.8.021013} {\bibfield  {journal} {\bibinfo
  {journal} {Phys. Rev. X}\ }\textbf {\bibinfo {volume} {8}},\ \bibinfo {pages}
  {021013} (\bibinfo {year} {2018})}\BibitemShut {NoStop}%
\bibitem [{\citenamefont {Khemani}\ \emph {et~al.}(2018)\citenamefont
  {Khemani}, \citenamefont {Vishwanath},\ and\ \citenamefont
  {Huse}}]{m:khemaniOperatorSpreadingEmergence2018}%
  \BibitemOpen
  \bibfield  {author} {\bibinfo {author} {\bibfnamefont {V.}~\bibnamefont
  {Khemani}}, \bibinfo {author} {\bibfnamefont {A.}~\bibnamefont
  {Vishwanath}},\ and\ \bibinfo {author} {\bibfnamefont {D.~A.}\ \bibnamefont
  {Huse}},\ }\bibfield  {title} {\bibinfo {title} {Operator {{Spreading}} and
  the {{Emergence}} of {{Dissipative Hydrodynamics}} under {{Unitary
  Evolution}} with {{Conservation Laws}}},\ }\href
  {https://doi.org/10.1103/PhysRevX.8.031057} {\bibfield  {journal} {\bibinfo
  {journal} {Phys. Rev. X}\ }\textbf {\bibinfo {volume} {8}},\ \bibinfo {pages}
  {031057} (\bibinfo {year} {2018})}\BibitemShut {NoStop}%
\bibitem [{\citenamefont {Goldstein}\ \emph {et~al.}(2006)\citenamefont
  {Goldstein}, \citenamefont {Lebowitz}, \citenamefont {Tumulka},\ and\
  \citenamefont {Zangh{\`i}}}]{m:goldsteinDistributionWaveFunction2006}%
  \BibitemOpen
  \bibfield  {author} {\bibinfo {author} {\bibfnamefont {S.}~\bibnamefont
  {Goldstein}}, \bibinfo {author} {\bibfnamefont {J.~L.}\ \bibnamefont
  {Lebowitz}}, \bibinfo {author} {\bibfnamefont {R.}~\bibnamefont {Tumulka}},\
  and\ \bibinfo {author} {\bibfnamefont {N.}~\bibnamefont {Zangh{\`i}}},\
  }\bibfield  {title} {\bibinfo {title} {On the {{Distribution}} of the {{Wave
  Function}} for {{Systems}} in {{Thermal Equilibrium}}},\ }\href
  {https://doi.org/10.1007/s10955-006-9210-z} {\bibfield  {journal} {\bibinfo
  {journal} {J. Stat. Phys.}\ }\textbf {\bibinfo {volume} {125}},\ \bibinfo
  {pages} {1193} (\bibinfo {year} {2006})}\BibitemShut {NoStop}%
\bibitem [{\citenamefont
  {Reimann}(2008)}]{m:reimannFoundationStatisticalMechanics2008}%
  \BibitemOpen
  \bibfield  {author} {\bibinfo {author} {\bibfnamefont {P.}~\bibnamefont
  {Reimann}},\ }\bibfield  {title} {\bibinfo {title} {Foundation of
  {{Statistical Mechanics}} under {{Experimentally Realistic Conditions}}},\
  }\href {https://doi.org/10.1103/PhysRevLett.101.190403} {\bibfield  {journal}
  {\bibinfo  {journal} {Phys. Rev. Lett.}\ }\textbf {\bibinfo {volume} {101}},\
  \bibinfo {pages} {190403} (\bibinfo {year} {2008})}\BibitemShut {NoStop}%
\bibitem [{\citenamefont {Linden}\ \emph {et~al.}(2009)\citenamefont {Linden},
  \citenamefont {Popescu}, \citenamefont {Short},\ and\ \citenamefont
  {Winter}}]{m:lindenQuantumMechanicalEvolution2009}%
  \BibitemOpen
  \bibfield  {author} {\bibinfo {author} {\bibfnamefont {N.}~\bibnamefont
  {Linden}}, \bibinfo {author} {\bibfnamefont {S.}~\bibnamefont {Popescu}},
  \bibinfo {author} {\bibfnamefont {A.~J.}\ \bibnamefont {Short}},\ and\
  \bibinfo {author} {\bibfnamefont {A.}~\bibnamefont {Winter}},\ }\bibfield
  {title} {\bibinfo {title} {Quantum mechanical evolution towards thermal
  equilibrium},\ }\href {https://doi.org/10.1103/PhysRevE.79.061103} {\bibfield
   {journal} {\bibinfo  {journal} {Phys. Rev. E}\ }\textbf {\bibinfo {volume}
  {79}},\ \bibinfo {pages} {061103} (\bibinfo {year} {2009})}\BibitemShut
  {NoStop}%
\bibitem [{\citenamefont {Kaneko}\ \emph {et~al.}(2020)\citenamefont {Kaneko},
  \citenamefont {Iyoda},\ and\ \citenamefont
  {Sagawa}}]{m:kanekoCharacterizingComplexityManybody2020}%
  \BibitemOpen
  \bibfield  {author} {\bibinfo {author} {\bibfnamefont {K.}~\bibnamefont
  {Kaneko}}, \bibinfo {author} {\bibfnamefont {E.}~\bibnamefont {Iyoda}},\ and\
  \bibinfo {author} {\bibfnamefont {T.}~\bibnamefont {Sagawa}},\ }\bibfield
  {title} {\bibinfo {title} {Characterizing complexity of many-body quantum
  dynamics by higher-order eigenstate thermalization},\ }\href
  {https://doi.org/10.1103/PhysRevA.101.042126} {\bibfield  {journal} {\bibinfo
   {journal} {Phys. Rev. A}\ }\textbf {\bibinfo {volume} {101}},\ \bibinfo
  {pages} {042126} (\bibinfo {year} {2020})}\BibitemShut {NoStop}%
\bibitem [{\citenamefont {Huang}()}]{m:huangExtensiveEntropyUnitary2021}%
  \BibitemOpen
  \bibfield  {author} {\bibinfo {author} {\bibfnamefont {Y.}~\bibnamefont
  {Huang}},\ }\bibfield  {title} {\bibinfo {title} {Extensive entropy from
  unitary evolution},\ }\href {http://arxiv.org/abs/2104.02053} {\ }\Eprint
  {https://arxiv.org/abs/2104.02053} {arXiv:2104.02053} \BibitemShut {NoStop}%
\bibitem [{\citenamefont {Lieb}\ and\ \citenamefont
  {Robinson}(2004)}]{m:liebFiniteGroupVelocity2004}%
  \BibitemOpen
  \bibfield  {author} {\bibinfo {author} {\bibfnamefont {E.~H.}\ \bibnamefont
  {Lieb}}\ and\ \bibinfo {author} {\bibfnamefont {D.~W.}\ \bibnamefont
  {Robinson}},\ }\bibfield  {title} {\bibinfo {title} {The {{Finite Group
  Velocity}} of {{Quantum Spin Systems}}},\ }in\ \href
  {https://doi.org/10.1007/978-3-662-10018-9_25} {\emph {\bibinfo {booktitle}
  {Statistical {{Mechanics}}: {{Selecta}} of {{Elliott H}}. {{Lieb}}}}},\
  \bibinfo {editor} {edited by\ \bibinfo {editor} {\bibfnamefont
  {B.}~\bibnamefont {Nachtergaele}}, \bibinfo {editor} {\bibfnamefont {J.~P.}\
  \bibnamefont {Solovej}},\ and\ \bibinfo {editor} {\bibfnamefont
  {J.}~\bibnamefont {Yngvason}}}\ (\bibinfo  {publisher} {{Springer}},\
  \bibinfo {address} {{Berlin, Heidelberg}},\ \bibinfo {year} {2004})\ pp.\
  \bibinfo {pages} {425--431}\BibitemShut {NoStop}%
\bibitem [{\citenamefont {Tran}\ \emph {et~al.}(2021)\citenamefont {Tran},
  \citenamefont {Guo}, \citenamefont {Baldwin}, \citenamefont {Ehrenberg},
  \citenamefont {Gorshkov},\ and\ \citenamefont
  {Lucas}}]{m:tranLiebRobinsonLightCone2021}%
  \BibitemOpen
  \bibfield  {author} {\bibinfo {author} {\bibfnamefont {M.~C.}\ \bibnamefont
  {Tran}}, \bibinfo {author} {\bibfnamefont {A.~Y.}\ \bibnamefont {Guo}},
  \bibinfo {author} {\bibfnamefont {C.~L.}\ \bibnamefont {Baldwin}}, \bibinfo
  {author} {\bibfnamefont {A.}~\bibnamefont {Ehrenberg}}, \bibinfo {author}
  {\bibfnamefont {A.~V.}\ \bibnamefont {Gorshkov}},\ and\ \bibinfo {author}
  {\bibfnamefont {A.}~\bibnamefont {Lucas}},\ }\bibfield  {title} {\bibinfo
  {title} {Lieb-{{Robinson Light Cone}} for {{Power-Law Interactions}}},\
  }\href {https://doi.org/10.1103/PhysRevLett.127.160401} {\bibfield  {journal}
  {\bibinfo  {journal} {Phys. Rev. Lett.}\ }\textbf {\bibinfo {volume} {127}},\
  \bibinfo {pages} {160401} (\bibinfo {year} {2021})}\BibitemShut {NoStop}%
\bibitem [{\citenamefont {Yin}\ and\ \citenamefont
  {Lucas}(2022)}]{m:yinFiniteSpeedQuantum2022}%
  \BibitemOpen
  \bibfield  {author} {\bibinfo {author} {\bibfnamefont {C.}~\bibnamefont
  {Yin}}\ and\ \bibinfo {author} {\bibfnamefont {A.}~\bibnamefont {Lucas}},\
  }\bibfield  {title} {\bibinfo {title} {Finite {{Speed}} of {{Quantum
  Information}} in {{Models}} of {{Interacting Bosons}} at {{Finite
  Density}}},\ }\href {https://doi.org/10.1103/PhysRevX.12.021039} {\bibfield
  {journal} {\bibinfo  {journal} {Phys. Rev. X}\ }\textbf {\bibinfo {volume}
  {12}},\ \bibinfo {pages} {021039} (\bibinfo {year} {2022})}\BibitemShut
  {NoStop}%
\bibitem [{\citenamefont {Kuwahara}\ \emph {et~al.}()\citenamefont {Kuwahara},
  \citenamefont {Van~Vu},\ and\ \citenamefont
  {Saito}}]{m:kuwaharaOptimalLightCone2022}%
  \BibitemOpen
  \bibfield  {author} {\bibinfo {author} {\bibfnamefont {T.}~\bibnamefont
  {Kuwahara}}, \bibinfo {author} {\bibfnamefont {T.}~\bibnamefont {Van~Vu}},\
  and\ \bibinfo {author} {\bibfnamefont {K.}~\bibnamefont {Saito}},\ }\bibfield
   {title} {\bibinfo {title} {Optimal light cone and digital quantum simulation
  of interacting bosons},\ }\href {http://arxiv.org/abs/2206.14736} {\ }\Eprint
  {https://arxiv.org/abs/2206.14736} {arXiv:2206.14736} \BibitemShut {NoStop}%
\bibitem [{\citenamefont {Haah}\ \emph {et~al.}(2021)\citenamefont {Haah},
  \citenamefont {Hastings}, \citenamefont {Kothari},\ and\ \citenamefont
  {Low}}]{m:haahQuantumAlgorithmSimulating2021}%
  \BibitemOpen
  \bibfield  {author} {\bibinfo {author} {\bibfnamefont {J.}~\bibnamefont
  {Haah}}, \bibinfo {author} {\bibfnamefont {M.~B.}\ \bibnamefont {Hastings}},
  \bibinfo {author} {\bibfnamefont {R.}~\bibnamefont {Kothari}},\ and\ \bibinfo
  {author} {\bibfnamefont {G.~H.}\ \bibnamefont {Low}},\ }\bibfield  {title}
  {\bibinfo {title} {Quantum {{Algorithm}} for {{Simulating Real Time
  Evolution}} of {{Lattice Hamiltonians}}},\ }\href
  {https://doi.org/10.1137/18M1231511} {\bibfield  {journal} {\bibinfo
  {journal} {SIAM J. Comput.}\ ,\ \bibinfo {pages} {FOCS18}} (\bibinfo {year}
  {2021})}\BibitemShut {NoStop}%
\bibitem [{\citenamefont {Tran}\ \emph {et~al.}(2019)\citenamefont {Tran},
  \citenamefont {Guo}, \citenamefont {Su}, \citenamefont {Garrison},
  \citenamefont {Eldredge}, \citenamefont {{Foss-Feig}}, \citenamefont
  {Childs},\ and\ \citenamefont {Gorshkov}}]{m:tranLocalityDigitalQuantum2019}%
  \BibitemOpen
  \bibfield  {author} {\bibinfo {author} {\bibfnamefont {M.~C.}\ \bibnamefont
  {Tran}}, \bibinfo {author} {\bibfnamefont {A.~Y.}\ \bibnamefont {Guo}},
  \bibinfo {author} {\bibfnamefont {Y.}~\bibnamefont {Su}}, \bibinfo {author}
  {\bibfnamefont {J.~R.}\ \bibnamefont {Garrison}}, \bibinfo {author}
  {\bibfnamefont {Z.}~\bibnamefont {Eldredge}}, \bibinfo {author}
  {\bibfnamefont {M.}~\bibnamefont {{Foss-Feig}}}, \bibinfo {author}
  {\bibfnamefont {A.~M.}\ \bibnamefont {Childs}},\ and\ \bibinfo {author}
  {\bibfnamefont {A.~V.}\ \bibnamefont {Gorshkov}},\ }\bibfield  {title}
  {\bibinfo {title} {Locality and {{Digital Quantum Simulation}} of {{Power-Law
  Interactions}}},\ }\href {https://doi.org/10.1103/PhysRevX.9.031006}
  {\bibfield  {journal} {\bibinfo  {journal} {Phys. Rev. X}\ }\textbf {\bibinfo
  {volume} {9}},\ \bibinfo {pages} {031006} (\bibinfo {year}
  {2019})}\BibitemShut {NoStop}%
\bibitem [{\citenamefont
  {Scott}(2006)}]{m:scottTightInformationallyComplete2006}%
  \BibitemOpen
  \bibfield  {author} {\bibinfo {author} {\bibfnamefont {A.~J.}\ \bibnamefont
  {Scott}},\ }\bibfield  {title} {\bibinfo {title} {Tight informationally
  complete quantum measurements},\ }\href
  {https://doi.org/10.1088/0305-4470/39/43/009} {\bibfield  {journal} {\bibinfo
   {journal} {J. Phys. A: Math. Gen.}\ }\textbf {\bibinfo {volume} {39}},\
  \bibinfo {pages} {13507} (\bibinfo {year} {2006})}\BibitemShut {NoStop}%
\bibitem [{\citenamefont {Daubechies}(1992)}]{Daubechiestenlectures}%
  \BibitemOpen
  \bibfield  {author} {\bibinfo {author} {\bibfnamefont {I.}~\bibnamefont
  {Daubechies}},\ }\href {https://doi.org/10.1137/1.9781611970104} {\emph
  {\bibinfo {title} {Ten Lectures on Wavelets}}}\ (\bibinfo  {publisher}
  {Society for Industrial and Applied Mathematics},\ \bibinfo {year}
  {1992})\BibitemShut {NoStop}%
\bibitem [{\citenamefont {Renes}\ \emph {et~al.}(2004)\citenamefont {Renes},
  \citenamefont {{Blume-Kohout}}, \citenamefont {Scott},\ and\ \citenamefont
  {Caves}}]{m:renesSymmetricInformationallyComplete2004}%
  \BibitemOpen
  \bibfield  {author} {\bibinfo {author} {\bibfnamefont {J.~M.}\ \bibnamefont
  {Renes}}, \bibinfo {author} {\bibfnamefont {R.}~\bibnamefont
  {{Blume-Kohout}}}, \bibinfo {author} {\bibfnamefont {A.~J.}\ \bibnamefont
  {Scott}},\ and\ \bibinfo {author} {\bibfnamefont {C.~M.}\ \bibnamefont
  {Caves}},\ }\bibfield  {title} {\bibinfo {title} {Symmetric informationally
  complete quantum measurements},\ }\href {https://doi.org/10.1063/1.1737053}
  {\bibfield  {journal} {\bibinfo  {journal} {J. Math. Phys.}\ }\textbf
  {\bibinfo {volume} {45}},\ \bibinfo {pages} {2171} (\bibinfo {year}
  {2004})}\BibitemShut {NoStop}%
\bibitem [{\citenamefont {Klappenecker}\ and\ \citenamefont
  {Rotteler}(2005)}]{m:klappeneckerMutuallyUnbiasedBases2005}%
  \BibitemOpen
  \bibfield  {author} {\bibinfo {author} {\bibfnamefont {A.}~\bibnamefont
  {Klappenecker}}\ and\ \bibinfo {author} {\bibfnamefont {M.}~\bibnamefont
  {Rotteler}},\ }\bibfield  {title} {\bibinfo {title} {Mutually unbiased bases
  are complex projective 2-designs},\ }in\ \href
  {https://doi.org/10.1109/ISIT.2005.1523643} {\emph {\bibinfo {booktitle}
  {Proceedings. {{International Symposium}} on {{Information Theory}}, 2005.
  {{ISIT}} 2005.}}}\ (\bibinfo {year} {2005})\ pp.\ \bibinfo {pages}
  {1740--1744}\BibitemShut {NoStop}%
\bibitem [{\citenamefont {Slagle}\ \emph {et~al.}(2022)\citenamefont {Slagle},
  \citenamefont {Liu}, \citenamefont {Aasen}, \citenamefont {Pichler},
  \citenamefont {Mong}, \citenamefont {Chen}, \citenamefont {Endres},\ and\
  \citenamefont {Alicea}}]{m:slagleQuantumSpinLiquids2022}%
  \BibitemOpen
  \bibfield  {author} {\bibinfo {author} {\bibfnamefont {K.}~\bibnamefont
  {Slagle}}, \bibinfo {author} {\bibfnamefont {Y.}~\bibnamefont {Liu}},
  \bibinfo {author} {\bibfnamefont {D.}~\bibnamefont {Aasen}}, \bibinfo
  {author} {\bibfnamefont {H.}~\bibnamefont {Pichler}}, \bibinfo {author}
  {\bibfnamefont {R.~S.~K.}\ \bibnamefont {Mong}}, \bibinfo {author}
  {\bibfnamefont {X.}~\bibnamefont {Chen}}, \bibinfo {author} {\bibfnamefont
  {M.}~\bibnamefont {Endres}},\ and\ \bibinfo {author} {\bibfnamefont
  {J.}~\bibnamefont {Alicea}},\ }\bibfield  {title} {\bibinfo {title} {Quantum
  spin liquids bootstrapped from {{Ising}} criticality in {{Rydberg}} arrays},\
  }\href {https://doi.org/10.1103/PhysRevB.106.115122} {\bibfield  {journal}
  {\bibinfo  {journal} {Phys. Rev. B}\ }\textbf {\bibinfo {volume} {106}},\
  \bibinfo {pages} {115122} (\bibinfo {year} {2022})}\BibitemShut {NoStop}%
\bibitem [{\citenamefont {Omran}\ \emph {et~al.}(2019)\citenamefont {Omran},
  \citenamefont {Levine}, \citenamefont {Keesling}, \citenamefont {Semeghini},
  \citenamefont {Wang}, \citenamefont {Ebadi}, \citenamefont {Bernien},
  \citenamefont {Zibrov}, \citenamefont {Pichler}, \citenamefont {Choi},
  \citenamefont {Cui}, \citenamefont {Rossignolo}, \citenamefont {Rembold},
  \citenamefont {Montangero}, \citenamefont {Calarco}, \citenamefont {Endres},
  \citenamefont {Greiner}, \citenamefont {Vuleti{\'c}},\ and\ \citenamefont
  {Lukin}}]{m:omranGenerationManipulationSchrodinger2019}%
  \BibitemOpen
  \bibfield  {author} {\bibinfo {author} {\bibfnamefont {A.}~\bibnamefont
  {Omran}}, \bibinfo {author} {\bibfnamefont {H.}~\bibnamefont {Levine}},
  \bibinfo {author} {\bibfnamefont {A.}~\bibnamefont {Keesling}}, \bibinfo
  {author} {\bibfnamefont {G.}~\bibnamefont {Semeghini}}, \bibinfo {author}
  {\bibfnamefont {T.~T.}\ \bibnamefont {Wang}}, \bibinfo {author}
  {\bibfnamefont {S.}~\bibnamefont {Ebadi}}, \bibinfo {author} {\bibfnamefont
  {H.}~\bibnamefont {Bernien}}, \bibinfo {author} {\bibfnamefont {A.~S.}\
  \bibnamefont {Zibrov}}, \bibinfo {author} {\bibfnamefont {H.}~\bibnamefont
  {Pichler}}, \bibinfo {author} {\bibfnamefont {S.}~\bibnamefont {Choi}},
  \bibinfo {author} {\bibfnamefont {J.}~\bibnamefont {Cui}}, \bibinfo {author}
  {\bibfnamefont {M.}~\bibnamefont {Rossignolo}}, \bibinfo {author}
  {\bibfnamefont {P.}~\bibnamefont {Rembold}}, \bibinfo {author} {\bibfnamefont
  {S.}~\bibnamefont {Montangero}}, \bibinfo {author} {\bibfnamefont
  {T.}~\bibnamefont {Calarco}}, \bibinfo {author} {\bibfnamefont
  {M.}~\bibnamefont {Endres}}, \bibinfo {author} {\bibfnamefont
  {M.}~\bibnamefont {Greiner}}, \bibinfo {author} {\bibfnamefont
  {V.}~\bibnamefont {Vuleti{\'c}}},\ and\ \bibinfo {author} {\bibfnamefont
  {M.~D.}\ \bibnamefont {Lukin}},\ }\bibfield  {title} {\bibinfo {title}
  {Generation and manipulation of {{Schr\"odinger}} cat states in {{Rydberg}}
  atom arrays},\ }\href {https://doi.org/10.1126/science.aax9743} {\bibfield
  {journal} {\bibinfo  {journal} {Science}\ }\textbf {\bibinfo {volume}
  {365}},\ \bibinfo {pages} {570} (\bibinfo {year} {2019})}\BibitemShut
  {NoStop}%
\bibitem [{\citenamefont {Semeghini}\ \emph {et~al.}(2021)\citenamefont
  {Semeghini}, \citenamefont {Levine}, \citenamefont {Keesling}, \citenamefont
  {Ebadi}, \citenamefont {Wang}, \citenamefont {Bluvstein}, \citenamefont
  {Verresen}, \citenamefont {Pichler}, \citenamefont {Kalinowski},
  \citenamefont {Samajdar}, \citenamefont {Omran}, \citenamefont {Sachdev},
  \citenamefont {Vishwanath}, \citenamefont {Greiner}, \citenamefont
  {Vuleti{\'c}},\ and\ \citenamefont
  {Lukin}}]{m:semeghiniProbingTopologicalSpin2021}%
  \BibitemOpen
  \bibfield  {author} {\bibinfo {author} {\bibfnamefont {G.}~\bibnamefont
  {Semeghini}}, \bibinfo {author} {\bibfnamefont {H.}~\bibnamefont {Levine}},
  \bibinfo {author} {\bibfnamefont {A.}~\bibnamefont {Keesling}}, \bibinfo
  {author} {\bibfnamefont {S.}~\bibnamefont {Ebadi}}, \bibinfo {author}
  {\bibfnamefont {T.~T.}\ \bibnamefont {Wang}}, \bibinfo {author}
  {\bibfnamefont {D.}~\bibnamefont {Bluvstein}}, \bibinfo {author}
  {\bibfnamefont {R.}~\bibnamefont {Verresen}}, \bibinfo {author}
  {\bibfnamefont {H.}~\bibnamefont {Pichler}}, \bibinfo {author} {\bibfnamefont
  {M.}~\bibnamefont {Kalinowski}}, \bibinfo {author} {\bibfnamefont
  {R.}~\bibnamefont {Samajdar}}, \bibinfo {author} {\bibfnamefont
  {A.}~\bibnamefont {Omran}}, \bibinfo {author} {\bibfnamefont
  {S.}~\bibnamefont {Sachdev}}, \bibinfo {author} {\bibfnamefont
  {A.}~\bibnamefont {Vishwanath}}, \bibinfo {author} {\bibfnamefont
  {M.}~\bibnamefont {Greiner}}, \bibinfo {author} {\bibfnamefont
  {V.}~\bibnamefont {Vuleti{\'c}}},\ and\ \bibinfo {author} {\bibfnamefont
  {M.~D.}\ \bibnamefont {Lukin}},\ }\bibfield  {title} {\bibinfo {title}
  {Probing topological spin liquids on a programmable quantum simulator},\
  }\href {https://doi.org/10.1126/science.abi8794} {\bibfield  {journal}
  {\bibinfo  {journal} {Science}\ }\textbf {\bibinfo {volume} {374}},\ \bibinfo
  {pages} {1242} (\bibinfo {year} {2021})}\BibitemShut {NoStop}%
\bibitem [{\citenamefont {Emery}(1987)}]{m:emeryTheoryHighMathrmT1987}%
  \BibitemOpen
  \bibfield  {author} {\bibinfo {author} {\bibfnamefont {V.~J.}\ \bibnamefont
  {Emery}},\ }\bibfield  {title} {\bibinfo {title} {Theory of
  high-${\mathrm{t}}_{\mathrm{c}}$ superconductivity in oxides},\ }\href
  {https://doi.org/10.1103/PhysRevLett.58.2794} {\bibfield  {journal} {\bibinfo
   {journal} {Phys. Rev. Lett.}\ }\textbf {\bibinfo {volume} {58}},\ \bibinfo
  {pages} {2794} (\bibinfo {year} {1987})}\BibitemShut {NoStop}%
\bibitem [{\citenamefont {Lee}\ \emph {et~al.}(2006)\citenamefont {Lee},
  \citenamefont {Nagaosa},\ and\ \citenamefont {Wen}}]{lee2006doping}%
  \BibitemOpen
  \bibfield  {author} {\bibinfo {author} {\bibfnamefont {P.~A.}\ \bibnamefont
  {Lee}}, \bibinfo {author} {\bibfnamefont {N.}~\bibnamefont {Nagaosa}},\ and\
  \bibinfo {author} {\bibfnamefont {X.-G.}\ \bibnamefont {Wen}},\ }\bibfield
  {title} {\bibinfo {title} {Doping a mott insulator: Physics of
  high-temperature superconductivity},\ }\href
  {https://doi.org/10.1103/RevModPhys.78.17} {\bibfield  {journal} {\bibinfo
  {journal} {Rev. Mod. Phys.}\ }\textbf {\bibinfo {volume} {78}},\ \bibinfo
  {pages} {17} (\bibinfo {year} {2006})}\BibitemShut {NoStop}%
\bibitem [{\citenamefont {Keimer}\ \emph {et~al.}(2015)\citenamefont {Keimer},
  \citenamefont {Kivelson}, \citenamefont {Norman}, \citenamefont {Uchida},\
  and\ \citenamefont {Zaanen}}]{m:keimerQuantumMatterHightemperature2015}%
  \BibitemOpen
  \bibfield  {author} {\bibinfo {author} {\bibfnamefont {B.}~\bibnamefont
  {Keimer}}, \bibinfo {author} {\bibfnamefont {S.~A.}\ \bibnamefont
  {Kivelson}}, \bibinfo {author} {\bibfnamefont {M.~R.}\ \bibnamefont
  {Norman}}, \bibinfo {author} {\bibfnamefont {S.}~\bibnamefont {Uchida}},\
  and\ \bibinfo {author} {\bibfnamefont {J.}~\bibnamefont {Zaanen}},\
  }\bibfield  {title} {\bibinfo {title} {From quantum matter to
  high-temperature superconductivity in copper oxides},\ }\href
  {https://doi.org/10.1038/nature14165} {\bibfield  {journal} {\bibinfo
  {journal} {Nature}\ }\textbf {\bibinfo {volume} {518}},\ \bibinfo {pages}
  {179} (\bibinfo {year} {2015})}\BibitemShut {NoStop}%
\bibitem [{\citenamefont {Chiu}\ \emph
  {et~al.}(2019{\natexlab{b}})\citenamefont {Chiu}, \citenamefont {Ji},
  \citenamefont {Bohrdt}, \citenamefont {Xu}, \citenamefont {Knap},
  \citenamefont {Demler}, \citenamefont {Grusdt}, \citenamefont {Greiner},\
  and\ \citenamefont {Greif}}]{chiu2019string}%
  \BibitemOpen
  \bibfield  {author} {\bibinfo {author} {\bibfnamefont {C.~S.}\ \bibnamefont
  {Chiu}}, \bibinfo {author} {\bibfnamefont {G.}~\bibnamefont {Ji}}, \bibinfo
  {author} {\bibfnamefont {A.}~\bibnamefont {Bohrdt}}, \bibinfo {author}
  {\bibfnamefont {M.}~\bibnamefont {Xu}}, \bibinfo {author} {\bibfnamefont
  {M.}~\bibnamefont {Knap}}, \bibinfo {author} {\bibfnamefont {E.}~\bibnamefont
  {Demler}}, \bibinfo {author} {\bibfnamefont {F.}~\bibnamefont {Grusdt}},
  \bibinfo {author} {\bibfnamefont {M.}~\bibnamefont {Greiner}},\ and\ \bibinfo
  {author} {\bibfnamefont {D.}~\bibnamefont {Greif}},\ }\bibfield  {title}
  {\bibinfo {title} {String patterns in the doped {H}ubbard model},\ }\href
  {https://www.science.org/doi/10.1126/science.aav3587} {\bibfield  {journal}
  {\bibinfo  {journal} {Science}\ }\textbf {\bibinfo {volume} {365}},\ \bibinfo
  {pages} {251} (\bibinfo {year} {2019}{\natexlab{b}})}\BibitemShut {NoStop}%
\bibitem [{\citenamefont {Tsuei}\ and\ \citenamefont
  {Kirtley}(2000)}]{m:tsueiPairingSymmetryCuprate2000}%
  \BibitemOpen
  \bibfield  {author} {\bibinfo {author} {\bibfnamefont {C.~C.}\ \bibnamefont
  {Tsuei}}\ and\ \bibinfo {author} {\bibfnamefont {J.~R.}\ \bibnamefont
  {Kirtley}},\ }\bibfield  {title} {\bibinfo {title} {Pairing symmetry in
  cuprate superconductors},\ }\href {https://doi.org/10.1103/RevModPhys.72.969}
  {\bibfield  {journal} {\bibinfo  {journal} {Rev. Mod. Phys.}\ }\textbf
  {\bibinfo {volume} {72}},\ \bibinfo {pages} {969} (\bibinfo {year}
  {2000})}\BibitemShut {NoStop}%
\bibitem [{\citenamefont
  {Gros}(1988)}]{m:grosSuperconductivityCorrelatedWave1988}%
  \BibitemOpen
  \bibfield  {author} {\bibinfo {author} {\bibfnamefont {C.}~\bibnamefont
  {Gros}},\ }\bibfield  {title} {\bibinfo {title} {Superconductivity in
  correlated wave functions},\ }\href {https://doi.org/10.1103/PhysRevB.38.931}
  {\bibfield  {journal} {\bibinfo  {journal} {Phys. Rev. B}\ }\textbf {\bibinfo
  {volume} {38}},\ \bibinfo {pages} {931} (\bibinfo {year} {1988})}\BibitemShut
  {NoStop}%
\bibitem [{\citenamefont {Chung}\ and\ \citenamefont
  {Peschel}(2001)}]{m:chungDensitymatrixSpectraSolvable2001}%
  \BibitemOpen
  \bibfield  {author} {\bibinfo {author} {\bibfnamefont {M.-C.}\ \bibnamefont
  {Chung}}\ and\ \bibinfo {author} {\bibfnamefont {I.}~\bibnamefont
  {Peschel}},\ }\bibfield  {title} {\bibinfo {title} {Density-matrix spectra of
  solvable fermionic systems},\ }\href
  {https://doi.org/10.1103/PhysRevB.64.064412} {\bibfield  {journal} {\bibinfo
  {journal} {Phys. Rev. B}\ }\textbf {\bibinfo {volume} {64}},\ \bibinfo
  {pages} {064412} (\bibinfo {year} {2001})}\BibitemShut {NoStop}%
\bibitem [{\citenamefont {Jiang}\ and\ \citenamefont
  {Kivelson}(2021)}]{m:jiangHighTemperatureSuperconductivity2021}%
  \BibitemOpen
  \bibfield  {author} {\bibinfo {author} {\bibfnamefont {H.-C.}\ \bibnamefont
  {Jiang}}\ and\ \bibinfo {author} {\bibfnamefont {S.~A.}\ \bibnamefont
  {Kivelson}},\ }\bibfield  {title} {\bibinfo {title} {High {{Temperature
  Superconductivity}} in a {{Lightly Doped Quantum Spin Liquid}}},\ }\href
  {https://doi.org/10.1103/PhysRevLett.127.097002} {\bibfield  {journal}
  {\bibinfo  {journal} {Phys. Rev. Lett.}\ }\textbf {\bibinfo {volume} {127}},\
  \bibinfo {pages} {097002} (\bibinfo {year} {2021})}\BibitemShut {NoStop}%
\bibitem [{\citenamefont {Gall}\ \emph {et~al.}(2021)\citenamefont {Gall},
  \citenamefont {Wurz}, \citenamefont {Samland}, \citenamefont {Chan},\ and\
  \citenamefont {K{\"o}hl}}]{m:gallCompetingMagneticOrders2021}%
  \BibitemOpen
  \bibfield  {author} {\bibinfo {author} {\bibfnamefont {M.}~\bibnamefont
  {Gall}}, \bibinfo {author} {\bibfnamefont {N.}~\bibnamefont {Wurz}}, \bibinfo
  {author} {\bibfnamefont {J.}~\bibnamefont {Samland}}, \bibinfo {author}
  {\bibfnamefont {C.~F.}\ \bibnamefont {Chan}},\ and\ \bibinfo {author}
  {\bibfnamefont {M.}~\bibnamefont {K{\"o}hl}},\ }\bibfield  {title} {\bibinfo
  {title} {Competing magnetic orders in a bilayer {{Hubbard}} model with
  ultracold atoms},\ }\href {https://doi.org/10.1038/s41586-020-03058-x}
  {\bibfield  {journal} {\bibinfo  {journal} {Nature}\ }\textbf {\bibinfo
  {volume} {589}},\ \bibinfo {pages} {40} (\bibinfo {year} {2021})}\BibitemShut
  {NoStop}%
\bibitem [{Note2()}]{Note2}%
  \BibitemOpen
  \bibinfo {note} {We can obtain any reduced density matrix of the BCS states
  by following the procedure in Ref.~\cite
  {m:chungDensitymatrixSpectraSolvable2001}. Note that Eq.~(16) in Ref.~\cite
  {m:chungDensitymatrixSpectraSolvable2001} should read $2\alpha =
  a^{11}-ca^{22}c^T$.}\BibitemShut {Stop}%
\bibitem [{\citenamefont {Petrescu}\ \emph {et~al.}(2017)\citenamefont
  {Petrescu}, \citenamefont {Piraud}, \citenamefont {Roux}, \citenamefont
  {McCulloch},\ and\ \citenamefont
  {Le~Hur}}]{m:petrescuPrecursorLaughlinState2017}%
  \BibitemOpen
  \bibfield  {author} {\bibinfo {author} {\bibfnamefont {A.}~\bibnamefont
  {Petrescu}}, \bibinfo {author} {\bibfnamefont {M.}~\bibnamefont {Piraud}},
  \bibinfo {author} {\bibfnamefont {G.}~\bibnamefont {Roux}}, \bibinfo {author}
  {\bibfnamefont {I.~P.}\ \bibnamefont {McCulloch}},\ and\ \bibinfo {author}
  {\bibfnamefont {K.}~\bibnamefont {Le~Hur}},\ }\bibfield  {title} {\bibinfo
  {title} {Precursor of the {{Laughlin}} state of hard-core bosons on a two-leg
  ladder},\ }\href {https://doi.org/10.1103/PhysRevB.96.014524} {\bibfield
  {journal} {\bibinfo  {journal} {Phys. Rev. B}\ }\textbf {\bibinfo {volume}
  {96}},\ \bibinfo {pages} {014524} (\bibinfo {year} {2017})}\BibitemShut
  {NoStop}%
\bibitem [{\citenamefont {Dehghani}\ \emph {et~al.}(2021)\citenamefont
  {Dehghani}, \citenamefont {Cian}, \citenamefont {Hafezi},\ and\ \citenamefont
  {Barkeshli}}]{m:dehghaniExtractionManybodyChern2021}%
  \BibitemOpen
  \bibfield  {author} {\bibinfo {author} {\bibfnamefont {H.}~\bibnamefont
  {Dehghani}}, \bibinfo {author} {\bibfnamefont {Z.-P.}\ \bibnamefont {Cian}},
  \bibinfo {author} {\bibfnamefont {M.}~\bibnamefont {Hafezi}},\ and\ \bibinfo
  {author} {\bibfnamefont {M.}~\bibnamefont {Barkeshli}},\ }\bibfield  {title}
  {\bibinfo {title} {Extraction of the many-body {{Chern}} number from a single
  wave function},\ }\href {https://doi.org/10.1103/PhysRevB.103.075102}
  {\bibfield  {journal} {\bibinfo  {journal} {Phys. Rev. B}\ }\textbf {\bibinfo
  {volume} {103}},\ \bibinfo {pages} {075102} (\bibinfo {year}
  {2021})}\BibitemShut {NoStop}%
\bibitem [{\citenamefont {Niu}\ \emph {et~al.}(1985)\citenamefont {Niu},
  \citenamefont {Thouless},\ and\ \citenamefont
  {Wu}}]{m:niuQuantizedHallConductance1985}%
  \BibitemOpen
  \bibfield  {author} {\bibinfo {author} {\bibfnamefont {Q.}~\bibnamefont
  {Niu}}, \bibinfo {author} {\bibfnamefont {D.~J.}\ \bibnamefont {Thouless}},\
  and\ \bibinfo {author} {\bibfnamefont {Y.-S.}\ \bibnamefont {Wu}},\
  }\bibfield  {title} {\bibinfo {title} {Quantized {{Hall}} conductance as a
  topological invariant},\ }\href {https://doi.org/10.1103/PhysRevB.31.3372}
  {\bibfield  {journal} {\bibinfo  {journal} {Phys. Rev. B}\ }\textbf {\bibinfo
  {volume} {31}},\ \bibinfo {pages} {3372} (\bibinfo {year}
  {1985})}\BibitemShut {NoStop}%
\bibitem [{\citenamefont {Motruk}\ and\ \citenamefont
  {Na}(2020)}]{m:motrukDetectingFractionalChern2020}%
  \BibitemOpen
  \bibfield  {author} {\bibinfo {author} {\bibfnamefont {J.}~\bibnamefont
  {Motruk}}\ and\ \bibinfo {author} {\bibfnamefont {I.}~\bibnamefont {Na}},\
  }\bibfield  {title} {\bibinfo {title} {Detecting {{Fractional Chern
  Insulators}} in {{Optical Lattices}} through {{Quantized Displacement}}},\
  }\href {https://doi.org/10.1103/PhysRevLett.125.236401} {\bibfield  {journal}
  {\bibinfo  {journal} {Phys. Rev. Lett.}\ }\textbf {\bibinfo {volume} {125}},\
  \bibinfo {pages} {236401} (\bibinfo {year} {2020})}\BibitemShut {NoStop}%
\bibitem [{\citenamefont {Repellin}\ \emph {et~al.}(2020)\citenamefont
  {Repellin}, \citenamefont {L{\'e}onard},\ and\ \citenamefont
  {Goldman}}]{m:repellinFractionalChernInsulators2020}%
  \BibitemOpen
  \bibfield  {author} {\bibinfo {author} {\bibfnamefont {C.}~\bibnamefont
  {Repellin}}, \bibinfo {author} {\bibfnamefont {J.}~\bibnamefont
  {L{\'e}onard}},\ and\ \bibinfo {author} {\bibfnamefont {N.}~\bibnamefont
  {Goldman}},\ }\bibfield  {title} {\bibinfo {title} {Fractional {{Chern}}
  insulators of few bosons in a box: {{Hall}} plateaus from center-of-mass
  drifts and density profiles},\ }\href
  {https://doi.org/10.1103/PhysRevA.102.063316} {\bibfield  {journal} {\bibinfo
   {journal} {Phys. Rev. A}\ }\textbf {\bibinfo {volume} {102}},\ \bibinfo
  {pages} {063316} (\bibinfo {year} {2020})}\BibitemShut {NoStop}%
\bibitem [{\citenamefont {Cian}\ \emph {et~al.}(2021)\citenamefont {Cian},
  \citenamefont {Dehghani}, \citenamefont {Elben}, \citenamefont {Vermersch},
  \citenamefont {Zhu}, \citenamefont {Barkeshli}, \citenamefont {Zoller},\ and\
  \citenamefont {Hafezi}}]{m:cianManyBodyChernNumber2021}%
  \BibitemOpen
  \bibfield  {author} {\bibinfo {author} {\bibfnamefont {Z.-P.}\ \bibnamefont
  {Cian}}, \bibinfo {author} {\bibfnamefont {H.}~\bibnamefont {Dehghani}},
  \bibinfo {author} {\bibfnamefont {A.}~\bibnamefont {Elben}}, \bibinfo
  {author} {\bibfnamefont {B.}~\bibnamefont {Vermersch}}, \bibinfo {author}
  {\bibfnamefont {G.}~\bibnamefont {Zhu}}, \bibinfo {author} {\bibfnamefont
  {M.}~\bibnamefont {Barkeshli}}, \bibinfo {author} {\bibfnamefont
  {P.}~\bibnamefont {Zoller}},\ and\ \bibinfo {author} {\bibfnamefont
  {M.}~\bibnamefont {Hafezi}},\ }\bibfield  {title} {\bibinfo {title}
  {Many-{{Body Chern Number}} from {{Statistical Correlations}} of {{Randomized
  Measurements}}},\ }\href {https://doi.org/10.1103/PhysRevLett.126.050501}
  {\bibfield  {journal} {\bibinfo  {journal} {Phys. Rev. Lett.}\ }\textbf
  {\bibinfo {volume} {126}},\ \bibinfo {pages} {050501} (\bibinfo {year}
  {2021})}\BibitemShut {NoStop}%
\bibitem [{Note3()}]{Note3}%
  \BibitemOpen
  \bibinfo {note} {For numerical tractability, we disregarded all instances of
  having three bosons on $R_1 \cup R_3$ (since these are extremely rare ($<0.1
  \%$) and computationally expensive to invert)}\BibitemShut {NoStop}%
\bibitem [{\citenamefont {Palm}\ \emph {et~al.}(2021)\citenamefont {Palm},
  \citenamefont {Buser}, \citenamefont {L{\'e}onard}, \citenamefont
  {Aidelsburger}, \citenamefont {Schollw{\"o}ck},\ and\ \citenamefont
  {Grusdt}}]{m:palmBosonicPfaffianState2021}%
  \BibitemOpen
  \bibfield  {author} {\bibinfo {author} {\bibfnamefont {F.~A.}\ \bibnamefont
  {Palm}}, \bibinfo {author} {\bibfnamefont {M.}~\bibnamefont {Buser}},
  \bibinfo {author} {\bibfnamefont {J.}~\bibnamefont {L{\'e}onard}}, \bibinfo
  {author} {\bibfnamefont {M.}~\bibnamefont {Aidelsburger}}, \bibinfo {author}
  {\bibfnamefont {U.}~\bibnamefont {Schollw{\"o}ck}},\ and\ \bibinfo {author}
  {\bibfnamefont {F.}~\bibnamefont {Grusdt}},\ }\bibfield  {title} {\bibinfo
  {title} {Bosonic {{Pfaffian}} state in the {{Hofstadter-Bose-Hubbard}}
  model},\ }\href {https://doi.org/10.1103/PhysRevB.103.L161101} {\bibfield
  {journal} {\bibinfo  {journal} {Phys. Rev. B}\ }\textbf {\bibinfo {volume}
  {103}},\ \bibinfo {pages} {L161101} (\bibinfo {year} {2021})}\BibitemShut
  {NoStop}%
\bibitem [{\citenamefont {Vermersch}\ \emph {et~al.}(2018)\citenamefont
  {Vermersch}, \citenamefont {Elben}, \citenamefont {Dalmonte}, \citenamefont
  {Cirac},\ and\ \citenamefont
  {Zoller}}]{mm:vermerschUnitaryDesignsRandom2018}%
  \BibitemOpen
  \bibfield  {author} {\bibinfo {author} {\bibfnamefont {B.}~\bibnamefont
  {Vermersch}}, \bibinfo {author} {\bibfnamefont {A.}~\bibnamefont {Elben}},
  \bibinfo {author} {\bibfnamefont {M.}~\bibnamefont {Dalmonte}}, \bibinfo
  {author} {\bibfnamefont {J.~I.}\ \bibnamefont {Cirac}},\ and\ \bibinfo
  {author} {\bibfnamefont {P.}~\bibnamefont {Zoller}},\ }\bibfield  {title}
  {\bibinfo {title} {Unitary $n$-designs via random quenches in atomic
  {{Hubbard}} and spin models: {{Application}} to the measurement of
  {{R}}\'enyi entropies},\ }\href {https://doi.org/10.1103/PhysRevA.97.023604}
  {\bibfield  {journal} {\bibinfo  {journal} {Phys. Rev. A}\ }\textbf {\bibinfo
  {volume} {97}},\ \bibinfo {pages} {023604} (\bibinfo {year}
  {2018})}\BibitemShut {NoStop}%
\bibitem [{\citenamefont {Petrescu}\ and\ \citenamefont
  {Le~Hur}(2013)}]{m:petrescuBosonicMottInsulator2013}%
  \BibitemOpen
  \bibfield  {author} {\bibinfo {author} {\bibfnamefont {A.}~\bibnamefont
  {Petrescu}}\ and\ \bibinfo {author} {\bibfnamefont {K.}~\bibnamefont
  {Le~Hur}},\ }\bibfield  {title} {\bibinfo {title} {Bosonic {{Mott Insulator}}
  with {{Meissner Currents}}},\ }\href
  {https://doi.org/10.1103/PhysRevLett.111.150601} {\bibfield  {journal}
  {\bibinfo  {journal} {Phys. Rev. Lett.}\ }\textbf {\bibinfo {volume} {111}},\
  \bibinfo {pages} {150601} (\bibinfo {year} {2013})}\BibitemShut {NoStop}%
\bibitem [{\citenamefont {Piraud}\ \emph {et~al.}(2015)\citenamefont {Piraud},
  \citenamefont {{Heidrich-Meisner}}, \citenamefont {McCulloch}, \citenamefont
  {Greschner}, \citenamefont {Vekua},\ and\ \citenamefont
  {Schollw{\"o}ck}}]{m:piraudVortexMeissnerPhases2015}%
  \BibitemOpen
  \bibfield  {author} {\bibinfo {author} {\bibfnamefont {M.}~\bibnamefont
  {Piraud}}, \bibinfo {author} {\bibfnamefont {F.}~\bibnamefont
  {{Heidrich-Meisner}}}, \bibinfo {author} {\bibfnamefont {I.~P.}\ \bibnamefont
  {McCulloch}}, \bibinfo {author} {\bibfnamefont {S.}~\bibnamefont
  {Greschner}}, \bibinfo {author} {\bibfnamefont {T.}~\bibnamefont {Vekua}},\
  and\ \bibinfo {author} {\bibfnamefont {U.}~\bibnamefont {Schollw{\"o}ck}},\
  }\bibfield  {title} {\bibinfo {title} {Vortex and {{Meissner}} phases of
  strongly interacting bosons on a two-leg ladder},\ }\href
  {https://doi.org/10.1103/PhysRevB.91.140406} {\bibfield  {journal} {\bibinfo
  {journal} {Phys. Rev. B}\ }\textbf {\bibinfo {volume} {91}},\ \bibinfo
  {pages} {140406} (\bibinfo {year} {2015})}\BibitemShut {NoStop}%
\bibitem [{\citenamefont {Greschner}\ \emph {et~al.}(2015)\citenamefont
  {Greschner}, \citenamefont {Piraud}, \citenamefont {Heidrich-Meisner},
  \citenamefont {McCulloch}, \citenamefont {Schollw\"ock},\ and\ \citenamefont
  {Vekua}}]{greschner2015spontaneous}%
  \BibitemOpen
  \bibfield  {author} {\bibinfo {author} {\bibfnamefont {S.}~\bibnamefont
  {Greschner}}, \bibinfo {author} {\bibfnamefont {M.}~\bibnamefont {Piraud}},
  \bibinfo {author} {\bibfnamefont {F.}~\bibnamefont {Heidrich-Meisner}},
  \bibinfo {author} {\bibfnamefont {I.~P.}\ \bibnamefont {McCulloch}}, \bibinfo
  {author} {\bibfnamefont {U.}~\bibnamefont {Schollw\"ock}},\ and\ \bibinfo
  {author} {\bibfnamefont {T.}~\bibnamefont {Vekua}},\ }\bibfield  {title}
  {\bibinfo {title} {Spontaneous increase of magnetic flux and chiral-current
  reversal in bosonic ladders: Swimming against the tide},\ }\href
  {https://link.aps.org/doi/10.1103/PhysRevLett.115.190402} {\bibfield
  {journal} {\bibinfo  {journal} {Phys. Rev. Lett.}\ }\textbf {\bibinfo
  {volume} {115}},\ \bibinfo {pages} {190402} (\bibinfo {year}
  {2015})}\BibitemShut {NoStop}%
\bibitem [{\citenamefont {Cornfeld}\ and\ \citenamefont
  {Sela}(2015)}]{m:cornfeldChiralCurrentsOnedimensional2015}%
  \BibitemOpen
  \bibfield  {author} {\bibinfo {author} {\bibfnamefont {E.}~\bibnamefont
  {Cornfeld}}\ and\ \bibinfo {author} {\bibfnamefont {E.}~\bibnamefont
  {Sela}},\ }\bibfield  {title} {\bibinfo {title} {Chiral currents in
  one-dimensional fractional quantum {{Hall}} states},\ }\href
  {https://doi.org/10.1103/PhysRevB.92.115446} {\bibfield  {journal} {\bibinfo
  {journal} {Phys. Rev. B}\ }\textbf {\bibinfo {volume} {92}},\ \bibinfo
  {pages} {115446} (\bibinfo {year} {2015})}\BibitemShut {NoStop}%
\bibitem [{\citenamefont {Senthil}\ and\ \citenamefont
  {Levin}(2013)}]{m:senthilIntegerQuantumHall2013}%
  \BibitemOpen
  \bibfield  {author} {\bibinfo {author} {\bibfnamefont {T.}~\bibnamefont
  {Senthil}}\ and\ \bibinfo {author} {\bibfnamefont {M.}~\bibnamefont
  {Levin}},\ }\bibfield  {title} {\bibinfo {title} {Integer {{Quantum Hall
  Effect}} for {{Bosons}}},\ }\href
  {https://doi.org/10.1103/PhysRevLett.110.046801} {\bibfield  {journal}
  {\bibinfo  {journal} {Phys. Rev. Lett.}\ }\textbf {\bibinfo {volume} {110}},\
  \bibinfo {pages} {046801} (\bibinfo {year} {2013})}\BibitemShut {NoStop}%
\bibitem [{\citenamefont {Grier}\ \emph {et~al.}()\citenamefont {Grier},
  \citenamefont {Pashayan},\ and\ \citenamefont
  {Schaeffer}}]{m:grierSampleoptimalClassicalShadows2022}%
  \BibitemOpen
  \bibfield  {author} {\bibinfo {author} {\bibfnamefont {D.}~\bibnamefont
  {Grier}}, \bibinfo {author} {\bibfnamefont {H.}~\bibnamefont {Pashayan}},\
  and\ \bibinfo {author} {\bibfnamefont {L.}~\bibnamefont {Schaeffer}},\
  }\bibfield  {title} {\bibinfo {title} {Sample-optimal classical shadows for
  pure states},\ }\href {http://arxiv.org/abs/2211.11810} {\ }\Eprint
  {https://arxiv.org/abs/2211.11810} {arXiv:2211.11810} \BibitemShut {NoStop}%
\bibitem [{\citenamefont {Wan}\ \emph {et~al.}()\citenamefont {Wan},
  \citenamefont {Huggins}, \citenamefont {Lee},\ and\ \citenamefont
  {Babbush}}]{m:wanMatchgateShadowsFermionic2022}%
  \BibitemOpen
  \bibfield  {author} {\bibinfo {author} {\bibfnamefont {K.}~\bibnamefont
  {Wan}}, \bibinfo {author} {\bibfnamefont {W.~J.}\ \bibnamefont {Huggins}},
  \bibinfo {author} {\bibfnamefont {J.}~\bibnamefont {Lee}},\ and\ \bibinfo
  {author} {\bibfnamefont {R.}~\bibnamefont {Babbush}},\ }\bibfield  {title}
  {\bibinfo {title} {Matchgate {{Shadows}} for {{Fermionic Quantum
  Simulation}}},\ }\href {http://arxiv.org/abs/2207.13723} {\ }\Eprint
  {https://arxiv.org/abs/2207.13723} {arXiv:2207.13723} \BibitemShut {NoStop}%
\bibitem [{\citenamefont {Kunjummen}\ \emph {et~al.}()\citenamefont
  {Kunjummen}, \citenamefont {Tran}, \citenamefont {Carney},\ and\
  \citenamefont {Taylor}}]{m:kunjummenShadowProcessTomography2022}%
  \BibitemOpen
  \bibfield  {author} {\bibinfo {author} {\bibfnamefont {J.}~\bibnamefont
  {Kunjummen}}, \bibinfo {author} {\bibfnamefont {M.~C.}\ \bibnamefont {Tran}},
  \bibinfo {author} {\bibfnamefont {D.}~\bibnamefont {Carney}},\ and\ \bibinfo
  {author} {\bibfnamefont {J.~M.}\ \bibnamefont {Taylor}},\ }\bibfield  {title}
  {\bibinfo {title} {Shadow process tomography of quantum channels},\ }\href
  {http://arxiv.org/abs/2110.03629} {\ }\Eprint
  {https://arxiv.org/abs/2110.03629} {arXiv:2110.03629} \BibitemShut {NoStop}%
\bibitem [{\citenamefont {Levy}\ \emph {et~al.}()\citenamefont {Levy},
  \citenamefont {Luo},\ and\ \citenamefont
  {Clark}}]{m:levyClassicalShadowsQuantum2021}%
  \BibitemOpen
  \bibfield  {author} {\bibinfo {author} {\bibfnamefont {R.}~\bibnamefont
  {Levy}}, \bibinfo {author} {\bibfnamefont {D.}~\bibnamefont {Luo}},\ and\
  \bibinfo {author} {\bibfnamefont {B.~K.}\ \bibnamefont {Clark}},\ }\bibfield
  {title} {\bibinfo {title} {Classical {{Shadows}} for {{Quantum Process
  Tomography}} on {{Near-term Quantum Computers}}},\ }\href
  {http://arxiv.org/abs/2110.02965} {\ }\Eprint
  {https://arxiv.org/abs/2110.02965} {arXiv:2110.02965} \BibitemShut {NoStop}%
\bibitem [{\citenamefont {Huang}\ \emph {et~al.}()\citenamefont {Huang},
  \citenamefont {Chen},\ and\ \citenamefont
  {Preskill}}]{m:huangLearningPredictArbitrary2022}%
  \BibitemOpen
  \bibfield  {author} {\bibinfo {author} {\bibfnamefont {H.-Y.}\ \bibnamefont
  {Huang}}, \bibinfo {author} {\bibfnamefont {S.}~\bibnamefont {Chen}},\ and\
  \bibinfo {author} {\bibfnamefont {J.}~\bibnamefont {Preskill}},\ }\bibfield
  {title} {\bibinfo {title} {Learning to predict arbitrary quantum processes},\
  }\href {http://arxiv.org/abs/2210.14894} {\ }\Eprint
  {https://arxiv.org/abs/2210.14894} {arXiv:2210.14894} \BibitemShut {NoStop}%
\bibitem [{\citenamefont {Huang}\ \emph {et~al.}(2021)\citenamefont {Huang},
  \citenamefont {Kueng},\ and\ \citenamefont
  {Preskill}}]{m:huangEfficientEstimationPauli2021}%
  \BibitemOpen
  \bibfield  {author} {\bibinfo {author} {\bibfnamefont {H.-Y.}\ \bibnamefont
  {Huang}}, \bibinfo {author} {\bibfnamefont {R.}~\bibnamefont {Kueng}},\ and\
  \bibinfo {author} {\bibfnamefont {J.}~\bibnamefont {Preskill}},\ }\bibfield
  {title} {\bibinfo {title} {Efficient {{Estimation}} of {{Pauli Observables}}
  by {{Derandomization}}},\ }\href
  {https://doi.org/10.1103/PhysRevLett.127.030503} {\bibfield  {journal}
  {\bibinfo  {journal} {Phys. Rev. Lett.}\ }\textbf {\bibinfo {volume} {127}},\
  \bibinfo {pages} {030503} (\bibinfo {year} {2021})}\BibitemShut {NoStop}%
\bibitem [{\citenamefont {Rath}\ \emph {et~al.}(2021)\citenamefont {Rath},
  \citenamefont {{van Bijnen}}, \citenamefont {Elben}, \citenamefont {Zoller},\
  and\ \citenamefont {Vermersch}}]{m:rathImportanceSamplingRandomized2021}%
  \BibitemOpen
  \bibfield  {author} {\bibinfo {author} {\bibfnamefont {A.}~\bibnamefont
  {Rath}}, \bibinfo {author} {\bibfnamefont {R.}~\bibnamefont {{van Bijnen}}},
  \bibinfo {author} {\bibfnamefont {A.}~\bibnamefont {Elben}}, \bibinfo
  {author} {\bibfnamefont {P.}~\bibnamefont {Zoller}},\ and\ \bibinfo {author}
  {\bibfnamefont {B.}~\bibnamefont {Vermersch}},\ }\bibfield  {title} {\bibinfo
  {title} {Importance {{Sampling}} of {{Randomized Measurements}} for {{Probing
  Entanglement}}},\ }\href {https://doi.org/10.1103/PhysRevLett.127.200503}
  {\bibfield  {journal} {\bibinfo  {journal} {Phys. Rev. Lett.}\ }\textbf
  {\bibinfo {volume} {127}},\ \bibinfo {pages} {200503} (\bibinfo {year}
  {2021})}\BibitemShut {NoStop}%
\bibitem [{\citenamefont {Akhtar}\ \emph {et~al.}()\citenamefont {Akhtar},
  \citenamefont {Hu},\ and\ \citenamefont
  {You}}]{m:akhtarScalableFlexibleClassical2022}%
  \BibitemOpen
  \bibfield  {author} {\bibinfo {author} {\bibfnamefont {A.~A.}\ \bibnamefont
  {Akhtar}}, \bibinfo {author} {\bibfnamefont {H.-Y.}\ \bibnamefont {Hu}},\
  and\ \bibinfo {author} {\bibfnamefont {Y.-Z.}\ \bibnamefont {You}},\
  }\bibfield  {title} {\bibinfo {title} {Scalable and {{Flexible Classical
  Shadow Tomography}} with {{Tensor Networks}}},\ }\href
  {http://arxiv.org/abs/2209.02093} {\ }\Eprint
  {https://arxiv.org/abs/2209.02093} {arXiv:2209.02093} \BibitemShut {NoStop}%
\bibitem [{\citenamefont {Bertoni}\ \emph {et~al.}()\citenamefont {Bertoni},
  \citenamefont {Haferkamp}, \citenamefont {Hinsche}, \citenamefont {Ioannou},
  \citenamefont {Eisert},\ and\ \citenamefont
  {Pashayan}}]{m:bertoniShallowShadowsExpectation2022}%
  \BibitemOpen
  \bibfield  {author} {\bibinfo {author} {\bibfnamefont {C.}~\bibnamefont
  {Bertoni}}, \bibinfo {author} {\bibfnamefont {J.}~\bibnamefont {Haferkamp}},
  \bibinfo {author} {\bibfnamefont {M.}~\bibnamefont {Hinsche}}, \bibinfo
  {author} {\bibfnamefont {M.}~\bibnamefont {Ioannou}}, \bibinfo {author}
  {\bibfnamefont {J.}~\bibnamefont {Eisert}},\ and\ \bibinfo {author}
  {\bibfnamefont {H.}~\bibnamefont {Pashayan}},\ }\bibfield  {title} {\bibinfo
  {title} {Shallow shadows: {{Expectation}} estimation using low-depth random
  {{Clifford}} circuits},\ }\href {http://arxiv.org/abs/2209.12924} {\ }\Eprint
  {https://arxiv.org/abs/2209.12924} {arXiv:2209.12924} \BibitemShut {NoStop}%
\bibitem [{\citenamefont {Arienzo}\ \emph {et~al.}()\citenamefont {Arienzo},
  \citenamefont {Heinrich}, \citenamefont {Roth},\ and\ \citenamefont
  {Kliesch}}]{m:arienzoClosedformAnalyticExpressions2022a}%
  \BibitemOpen
  \bibfield  {author} {\bibinfo {author} {\bibfnamefont {M.}~\bibnamefont
  {Arienzo}}, \bibinfo {author} {\bibfnamefont {M.}~\bibnamefont {Heinrich}},
  \bibinfo {author} {\bibfnamefont {I.}~\bibnamefont {Roth}},\ and\ \bibinfo
  {author} {\bibfnamefont {M.}~\bibnamefont {Kliesch}},\ }\bibfield  {title}
  {\bibinfo {title} {Closed-form analytic expressions for shadow estimation
  with brickwork circuits},\ }\href {http://arxiv.org/abs/2211.09835} {\
  }\Eprint {https://arxiv.org/abs/2211.09835} {arXiv:2211.09835} \BibitemShut
  {NoStop}%
\bibitem [{\citenamefont {McGinley}\ and\ \citenamefont
  {Fava}()}]{accompanying}%
  \BibitemOpen
  \bibfield  {author} {\bibinfo {author} {\bibfnamefont {M.}~\bibnamefont
  {McGinley}}\ and\ \bibinfo {author} {\bibfnamefont {M.}~\bibnamefont
  {Fava}},\ }\bibfield  {title} {\bibinfo {title} {Shadow tomography from
  emergent state designs in analog quantum simulators},\ }\href
  {https://arxiv.org/abs/2212.02543} {\ }\Eprint
  {https://arxiv.org/abs/2212.02543} {2212.02543} \BibitemShut {NoStop}%
\bibitem [{Note4()}]{Note4}%
  \BibitemOpen
  \bibinfo {note} {Formally, the Haar ensemble is defined as the unique
  ensemble of unitaries that is invariant under any unitary transformation:
  $\forall V \in U(D), P(U)dU = P(UV) dU = P(VU)dU$.}\BibitemShut {Stop}%
\bibitem [{\citenamefont {Ambainis}\ and\ \citenamefont
  {Emerson}()}]{m:ambainisQuantumTdesignsTwise2007}%
  \BibitemOpen
  \bibfield  {author} {\bibinfo {author} {\bibfnamefont {A.}~\bibnamefont
  {Ambainis}}\ and\ \bibinfo {author} {\bibfnamefont {J.}~\bibnamefont
  {Emerson}},\ }\bibfield  {title} {\bibinfo {title} {Quantum t-designs: T-wise
  independence in the quantum world},\ }\href
  {http://arxiv.org/abs/quant-ph/0701126} {\ }\Eprint
  {https://arxiv.org/abs/quant-ph/0701126} {arXiv:quant-ph/0701126}
  \BibitemShut {NoStop}%
\bibitem [{\citenamefont {Dankert}\ \emph {et~al.}(2009)\citenamefont
  {Dankert}, \citenamefont {Cleve}, \citenamefont {Emerson},\ and\
  \citenamefont {Livine}}]{m:dankertExactApproximateUnitary2009}%
  \BibitemOpen
  \bibfield  {author} {\bibinfo {author} {\bibfnamefont {C.}~\bibnamefont
  {Dankert}}, \bibinfo {author} {\bibfnamefont {R.}~\bibnamefont {Cleve}},
  \bibinfo {author} {\bibfnamefont {J.}~\bibnamefont {Emerson}},\ and\ \bibinfo
  {author} {\bibfnamefont {E.}~\bibnamefont {Livine}},\ }\bibfield  {title}
  {\bibinfo {title} {Exact and approximate unitary 2-designs and their
  application to fidelity estimation},\ }\href
  {https://doi.org/10.1103/PhysRevA.80.012304} {\bibfield  {journal} {\bibinfo
  {journal} {Phys. Rev. A}\ }\textbf {\bibinfo {volume} {80}},\ \bibinfo
  {pages} {012304} (\bibinfo {year} {2009})}\BibitemShut {NoStop}%
\bibitem [{\citenamefont
  {Gottesman}()}]{m:gottesmanHeisenbergRepresentationQuantum1998}%
  \BibitemOpen
  \bibfield  {author} {\bibinfo {author} {\bibfnamefont {D.}~\bibnamefont
  {Gottesman}},\ }\bibfield  {title} {\bibinfo {title} {The {{Heisenberg
  Representation}} of {{Quantum Computers}}},\ }\href
  {http://arxiv.org/abs/quant-ph/9807006} {\ }\Eprint
  {https://arxiv.org/abs/quant-ph/9807006} {arXiv:quant-ph/9807006}
  \BibitemShut {NoStop}%
\end{thebibliography}%
\begin{widetext}
	\newpage
\end{widetext}

\appendix

\section{Review of classical shadow tomography}
\label{app:classical_shadow_tomography}

In this section we review the operating principles of classical shadow tomography. Consider a $n$-qubit state $\rho$ on a quantum device. Classical shadow tomography consists of applying random unitaries $U_j$ (drawn from an ensemble $\mathcal E$) to $\rho$, followed by measuring the state in the standard basis, obtaining a measurement outcome $s_j$ (Here denoted $s$ to denote that they are measurement outcomes on the system degrees of freedom).
Repeating the procedure by drawing a different $U$ for $m$ times, we obtain a set of classical data $\mathcal{D} = \{(U_1,s_1),\dots,(U_m,s_m)\}$.
By a judicious choice of the ensemble of unitaries $\mathcal E$, all information about $\rho$ can be extracted by large enough data.

Classical shadow tomography can be also recast into the formalism of \cref{sec:framework}.
Suppose the ensemble of unitaries $\mathcal E$ is finite and consists of $\mu$ unitaries $U_1,\dots,U_\mu$. 
To coherently draw an unitary uniformly at random from the ensemble, we can i) prepare an ancillary $\mu$-level system in the uniform superposition of basis states: $\ket{\phi} = \frac{1}{\sqrt{\mu}} \sum_a\ket{a}$ and ii) apply the controlled unitary
\begin{align} 
    U \equiv \sum_a U_a \otimes \ketbra{a} \label{eq:controlled-Uj}
\end{align}
on the extended system $\rho\otimes \ketbra{\phi}$.
Measurement of the ancilla in the standard basis will collapse it to one of the basis state $\ket{a}$ and effectively apply the corresponding $U_a$ on the target system with probability $1/\mu$.
In this way, we can view classical shadow tomography as a special case of our protocol where the scrambling quench is given by the controlled unitary $U$ in \cref{eq:controlled-Uj}.

Given an unitary ensemble $\mathcal E$, we can compute the corresponding recovery map $R$ for classical shadow tomography. 
In particular, the two choices of ensembles $\mathcal E$ considered in Ref.~\cite{m:huangPredictingManyProperties2020} have analytic expressions of the recovery map $R$: 
\begin{enumerate}
    \item $\mathcal E = \text{Haar}(2^n)$ is the Haar ensemble---the uniformly random ensemble of unitaries over the entire Hilbert space~\footnote{Formally, the Haar ensemble is defined as the unique ensemble of unitaries that is invariant under any unitary transformation: $\forall V \in U(D), P(U)dU = P(UV) dU = P(VU)dU$.}.
    In this case, the state recovery formula is:
    \begin{align}
        \rho = \lim_{m\rightarrow \infty } \frac{1}{m}\sum_{j=1}^m \left[ (2^n+1) U_{j}^\dagger \ketbra{z_j}{z_j} U_{j} - \mathbb{I} \right]~.~\label{eq:shadow_recovery}
    \end{align}
    Equivalently, in the formalism of \cref{sec:framework}, we can write \cref{eq:shadow_recovery} as:
    \begin{align}
        |\rho) = \lim_{m\rightarrow \infty} \frac{1}{m}\sum_{j=1}^m R|s_j,a_j \rangle , \label{eq:shadow_recovery2}
    \end{align}
    where $R = \big(S^\dagger S \big)^{-1} S^\dagger$ and $S$ is defined as in \cref{eq:linear_map_Q} with $U$ in \cref{eq:controlled-Uj}, with explicit expression
    \begin{align} 
    \left[\big(S^\dagger S \big)^{-1}\right]_{(b,c),(d,e)} = (2^n+1) \delta_{b,d}\delta_{c,e} - \delta_{b,c} \delta_{d,e}~, 
    \end{align}
    and $a_j$ labels the unitary $U_{a_j}$ applied to the system.

    In fact, the assumption $\mathcal E = \text{Haar}(2^n)$ is not necessary for this formula.
    For each $s$, we define 
    \begin{align} 
         \mathcal E_s = \{(p_U,U^\dag \ketbra{s} U ):(p_U,U)\in\mathcal E\}
    \end{align} 
    to be an ensemble of shadows corresponding to the same measurement outcome $s$ from applying different unitaries $U$.
    \Cref{eq:shadow_recovery2} holds as long as the resulting state ensemble $\mathcal E_s$ forms a \textit{projective two-design} for all $s$. A projective, or state, two-design is an ensemble of quantum states that has the same statistical properties as an ensemble of Haar-random states up to the second moment~\cite{m:ambainisQuantumTdesignsTwise2007}.
    Notably, projective two-designs can be generated by deep random Clifford circuits, which are 
    computationally efficient to classically simulate~\cite{m:dankertExactApproximateUnitary2009,m:gottesmanHeisenbergRepresentationQuantum1998}.
    
    \item $\mathcal E$ is the ensemble of \textit{local random unitaries}: each unitary $U_j \sim \mathcal E$ is the product of independently chosen on-site unitaries $U_j^{(k)}$: 
    $U_j = \bigotimes_{k=1}^n U_j^{(k)}$, where $U_j^{(k)} \sim \text{Haar}(2)$. 
    In this case, the state recovery formula is:
    \begin{align}
        \rho = \lim_{m\rightarrow \infty }\frac{1}{m}\sum_{j=1}^m  \bigotimes_{k=1}^N \left[3 U_j^{(k)\dagger} \ketbra*{s_j^{(k)}} U_j^{(k)} - \mathbb{I}_k \right],\nonumber
    \end{align}
    where $s_j^{(k)}$ is the $k$-th bit of the measurement outcome $s_j$. This can be interpreted as a product of recovery maps in \cref{eq:shadow_recovery2} on each qubit.
    As with the first case, the condition that $U_j^{(k)}\sim \text{Haar}(2)$ can be relaxed to the local state ensemble $\{(p(U^{(k)}),U^{(k)}{}^\dag \ketbra{s^{(k)}}U^{(k)})\}$ forming local state two-designs.
\end{enumerate}

In each case, Ref.~\cite{m:huangPredictingManyProperties2020} provided an upper bound for the  sample complexity when the ensemble of unitaries forms a global or local unitary three-design. 
In the first case, the upper bound on the sample complexity of estimating $\langle O \rangle$ is proportional to $\Tr(O^2)$. This means that a globally random unitary is best suited to estimate the expectation values of \textit{low-rank}, global observables, such as state fidelity and entanglement witnesses, which have small values of $\Tr(O^2)$. For example, the state fidelity to a target state $\ket{\Psi}$ is a rank-1 projector: $O = \ketbra{\Psi}{\Psi}$ and can be efficiently estimated with a constant number of measurements, independent of total system size. 
However, if $O$ is a local observable such as an onsite Pauli matrix, $\Tr(O^2) = 2^n$ is exponentially large. 
Therefore, case 1.~above (using 
globally random unitaries) does not efficiently measure the expectation values of local observables.

In contrast, case 2.~efficiently measures the expectation values of local observables (and more generally, $k$-point observables) but is inefficiently for global observables. Specifically, the upper bound on the sample complexity of estimating $\langle O \rangle$ is proportional to $4^{\text{locality}(O)} \Vert O \Vert^2_\infty$, where $\text{locality}(O)$ is the number of qubits that $O$ acts non-trivially on, and $\Vert O \Vert_\infty$ is its spectral norm.

We use these results to guide our protocol design in \cref{sec:quench-setups}. We design quenches that are ergodic within the entire extended Hilbert space to efficiently estimate global observables. Meanwhile, to efficiently estimate local observables, we divide the target system into smaller ``patches," and independently quench evolve each patch with separate ancillae (\cref{fig:schemes}).

\section{Invertibility of quench protocol: the second no-resonance condition}
\label{app:no_resonance}

Here we provide the technical details leading to \cref{eq:invertibility_main_eq1}.
We assume that the scrambling map $S$ is tomographically incomplete for almost every $t$. Reproducing \cref{eq:assumption1}, the assumption can be formulated as:
\begin{align}
    &\exists \delta\rho \text{~s.t.~for~a.e.~}t,~\forall z,\nonumber \\
    &\delta P_z(t) = \bra{z} U_t \left(\delta \rho \otimes \ketbra{\phi_\ancs}\right) U_t^\dagger \ket{z} = 0~, \tag{\ref{eq:assumption1}}\label{eq:assumption_app}
\end{align}
with $\delta\rho \equiv \rho-\sigma$.

We first integrate Eq.~\eqref{eq:assumption_app} over all $t$ to obtain:
\begin{align}
    0 &=\lim_{T\rightarrow \infty}\frac{1}{T} \int_0^T dt \bra{z} U_t \left(\delta \rho \otimes \ketbra{\phi_\ancs} \right) U_t^\dagger \ket{z} \nonumber\\
    &= \sum_{E,E'} \lim_{T\rightarrow \infty}\frac{1}{T} \int_0^T dt \exp[i(E'-E)t] \label{eq:one_copy_integral_intermediate}\\
    &~\times \braket{z}{E} \bra{E}\left(\delta \rho \otimes \ketbra{\phi_\ancs} \right)\ket{E'} \braket{E'}{z} \nonumber\\
    &=\sum_E \abs{\braket{z}{E}}^2  \bra{E}\left(\delta \rho \otimes\ketbra{\phi_\ancs} \right)\ket{E}~, \label{eq:one_copy_integral}
\end{align}
where the integral over time in \cref{eq:one_copy_integral_intermediate} evaluates to $\delta(E-E')$, giving \cref{eq:one_copy_integral}, which is \cref{eq:invertibility_secondary_eq1}. We then integrate the squared expression:
\begin{align}
    &0 = \lim_{T\rightarrow \infty}\frac{1}{T} \int_0^T dt \left(\bra{z} U_t \left(\delta \rho \otimes \ketbra{0}{0} \right) U_t^\dagger \ket{z}\right)^2 \nonumber\\
    &= \sum_{E_1,E_2,E_3,E_4}\lim_{T\rightarrow \infty}\frac{1}{T} \int_0^T dt e^{i(E_4+E_3-E_2-E_1)t} \nonumber\\
    &~\quad\times \braket{z}{E_1} \braket{z}{E_2} \bra{E_1}\left(\delta \rho \otimes \ketbra{\phi_\ancs} \right)\ket{E_3} \nonumber\\
    &~\quad\times \bra{E_2}\left(\delta \rho \otimes \ketbra{\phi_\ancs} \right)\ket{E_4} \braket{E_3}{z} \braket{E_4}{z} \label{eq:two_copy_integral_intermediate}\\
    &=\underbrace{\left[\sum_E \abs{\braket{z}{E}}^2  \bra{E}\left(\delta \rho \otimes \ketbra{\phi_\ancs} \right)\ket{E}\right]^2}_{E_1=E_3, E_2=E_4} \label{eq:two_copy_integral}\nonumber \\
    &~+ \underbrace{\sum_{E\neq E'} \abs{\braket{z}{E}}^2 \abs{\braket{z}{E'}}^2  \left|\bra{E}\left(\delta \rho \otimes\ketbra{\phi_\ancs} \right)\ket{E'}\right|^2}_{E_1=E_4, E_2=E_3, E_1 \neq E_2}~,
\end{align}
where we have assumed the second no-resonance condition on the spectrum of $H$ to perform the computation: that $E_1+E_2-E_3-E_4 = 0$ if and only if $E_1=E_3, E_2=E_4$ or $E_1=E_4, E_2=E_3$, and taking care not to double count the case $E_1=E_2=E_3=E_4$~\cite{m:markBenchmarkingQuantumSimulators2022}. The no-resonance condition is commonly assumed in literature on thermalization~\cite{m:goldsteinDistributionWaveFunction2006,m:reimannFoundationStatisticalMechanics2008,m:lindenQuantumMechanicalEvolution2009,m:kanekoCharacterizingComplexityManybody2020,m:huangExtensiveEntropyUnitary2021} and is generally considered a mild assumption on non-integrable Hamiltonians: if a given non-integrable Hamiltonian does not satisfy the no-resonance condition, an infinitesimally small perturbation will generically satisfy the no-resonance condition, without changing any other physical properties~\cite{m:kanekoCharacterizingComplexityManybody2020}. 

Combining \eqref{eq:one_copy_integral} and \eqref{eq:two_copy_integral} implies that the second term is 0, giving \cref{eq:invertibility_main_eq1}:

\begin{equation}
    \sum_{E\neq E'} \abs{\braket{z}{E}}^2 \abs{\braket{z}{E'}}^2\left|\bra{E}\left(\delta \rho \otimes\ketbra{\phi_\ancs} \right)\ket{E'}\right|^2 = 0~. \tag{\ref{eq:invertibility_main_eq1}}
\end{equation}

\section{Quench evolution time: Numerical evidence}
\label{app:lieb_robinson}
The sample complexity of our protocol---the number of samples required to extract a given observable---depends on details of the quench evolution, including the time of the quench. In \cref{sec:quench_time}, we argued that the requisite quench evolution time is set by the Lieb-Robinson bound. Here, we present numerical evidence to support this claim.

\begin{figure*}
    \centering
    \includegraphics[width=0.8\textwidth]{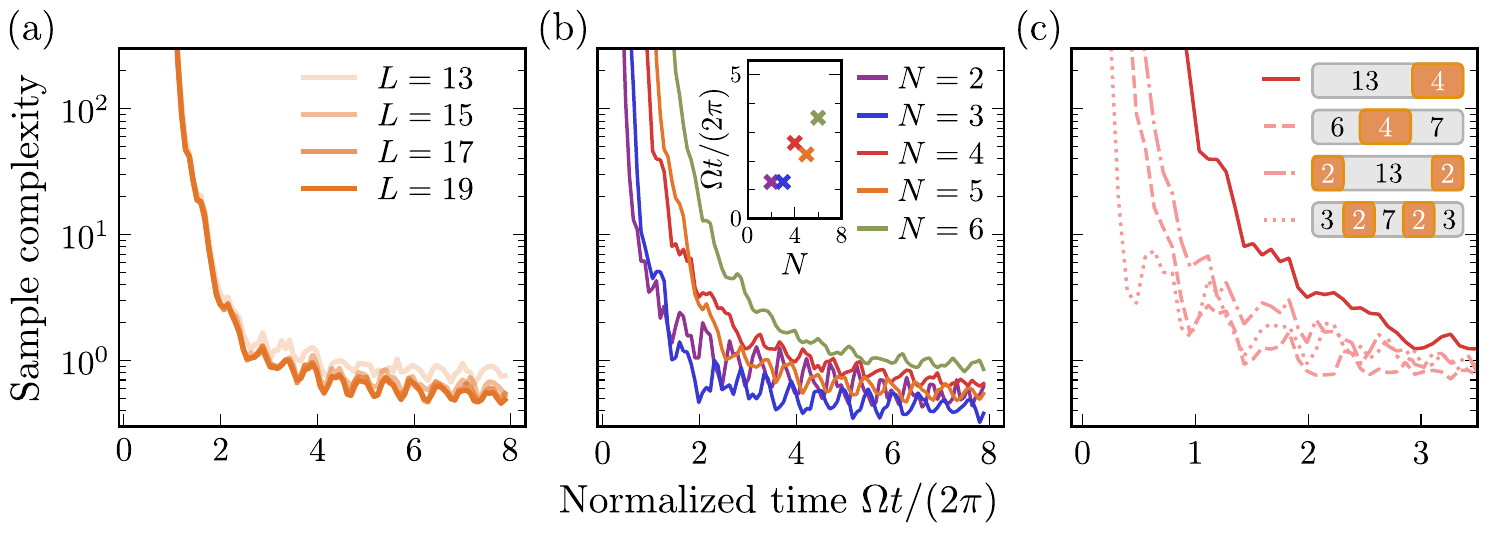}
    \caption{Effects of the Lieb-Robinson bound on the sample complexity of our scheme. In all examples, we simulate quench evolving a one-dimensional chain of system and ancilla Rydberg atoms, and we plot the sample complexity of estimating the state fidelity with our protocol.
    (a) We simulate the quench evolution of a one-dimensional chain of total length $L$, and $N=5$ Rydberg atoms as the system of interest. Increasing the length of the chain does not significantly change the sample complexity of our protocol against time: at early times, the sample complexity is high since most of the information about the state is contained in the Lieb-Robinson light cone. At late times, information has spread into enough ancillary degrees of freedom and the sample complexity does not significantly improve from spreading into more ancillae. (b) Larger system sizes $N$ require longer times for the sample complexity to saturate: we attribute this is due to the time taken for the Lieb-Robinson light cone to encompass a sufficient number of (approximately $N$) ancillary atoms. In the inset we observe an approximately linear scaling by plotting the earliest time that the sample complexity is below 1.5, as a function of system size $N$. Here, we fix the total system size to be $L=17$. (c) The rate of information spreading can be varied by changing the position of system atoms: the sample complexity saturates at different times, decreasing when the number of boundaries between system and ancilla atoms are increased from one to four (inset).}
    \label{fig:sample_complexity_time}
\end{figure*}

We numerically simulate our protocol in a system of Rydberg atoms, with quench evolution by the Hamiltonian in \cref{eq:H_ryd}. For simplicity, we arrange our system and ancilla Rydberg atoms in a one-dimensional chain of total length $L$, illustrated in \cref{fig:sample_complexity_time}. We plot the sample complexity required for the task in \cref{fig:rydberg-all}(a): measuring the fidelity between the experimental state and an ideal state which is the ground state of Eq.~\eqref{eq:H_ryd} defined on the initial system. For simplicity we set the experimental and target states to be equal.

As illustrated in \cref{fig:sample_complexity_time}(a), for a fixed number $N$ of system atoms, we find that increasing the number of ancilla atoms does not change the behavior of the sample complexity against time. This is because the sample complexity depends on the number of degrees of freedom that the system information spreads to. The sample complexity saturates as long as the information is (approximately) uniformly spread among at least $N$ additional ancilla atoms. At early times, the sample complexity is high since most of the information about the state is contained within the Lieb-Robinson light cone, which is smaller than the requisite number of atoms. Because of exponentially small tails, some information is spread into these degrees of freedom, making the quench map invertible, albeit with an exponentially large sample complexity. At late times, information has spread into enough ancillary degrees of freedom and the sample complexity does not significantly improve from spreading into more ancillae. In contrast, a larger number $N$ of system atoms requires a longer time for the sample complexity to saturate. In \cref{fig:sample_complexity_time}(b) we observe an approximately linear dependence of this saturation time $t^*$ with $N$, supporting our hypothesis. Finally, the rate of information spread can be varied by changing the position of the system atoms in the overall chain, which changes the number of boundaries between the system and ancilla atoms. This is reflected in \cref{fig:sample_complexity_time}(c), in which the sample complexity decreases more quickly when the number of such boundaries are increased from one to four.

\section{Optimal data post-processing with frame theory}\label{app:frame-theory}
In this section, we use \textit{frame theory} to derive the classical data processing scheme that minimizes the sample complexity in our protocol, as presented in \cref{sec:frame}. Our presentation largely follows Ref.~\cite{Daubechiestenlectures}.

We first define a \textit{frame}. For simplicity, we restrict our discussion to operators on the finite dimensional vector space $\mathbb{C}^d$. 
\begin{definition}[Frame]
A family of operators $\{|S_z)\}\subseteq L(\mathbb{C}^d)$ is an \emph{operator frame} if there exist constants $0<a\leq b<\infty$ such that
\begin{align}
    a\left(O|O\right) \leq \sum_z \left(O|S_z\right)\left(S_z|O\right) \leq b \left(O|O\right), \label{eq:frame_def}
\end{align}
for all operators $O\in L(\mathbb{C}^d)$, where $(A|B) = \tr(A^\dagger B)$.
A frame is \textit{tight} if $a=b$.  
\end{definition} 
For a finite frame, the right inequality is always satisfied. Meanwhile, the left inequality is satisfied when no operator $O$ is trace-orthogonal to all operators $S_z$ in the frame.
Therefore, the left inequality is equivalent to the invertibility of the superoperator $\sum_z |S_z)(S_z|$ corresponding to the frame. In our context, one can verify that as long as the scrambling map is invertible, the POVM $\{\ketbra{S_z}{S_z}\} \equiv \{\oket{S_z}\}$ is an operator frame. 

We next define the \textit{dual frame} of a frame.
\begin{definition}[Dual frame]
	A \textit{dual frame} $\{|R_z)\}$ of the frame $\{|S_z)\}$ is one such that
\begin{align}
    \sum_z |S_z)(R_z| = \mathbb{I}~, \label{eq:dual_frame}
\end{align}
where $\mathbb{I}$ is the identity superoperator. Therefore, for any frame $\{|S_z)\}$, an operator $\oket{O}$ has representation:
\begin{align}
    \oket{O} = \sum_z \oket{S_z}\obraket{R_z}{O}\equiv \sum_{z} o_z \oket{S_z}.
\end{align}
\end{definition} 
The dual frame is not unique if there are more than $d^2$ elements in the frame. One particular choice of the dual frame is the \textit{canonical dual frame}.
\begin{definition}[Canonical dual frame]
The \textit{canonical dual frame} of $\{|S_z)\}$ is
$\{|R^\can_z)\}$, defined as
\begin{align}
    |R^\can_z) \equiv \left[\sum_z |S_z)(S_z| \right]^{-1}  |S_z)~.
\end{align}	
\end{definition}

If the dual frame is not unique, neither is the representation $\oket{O} = \sum_{z} o_z \oket{S_z}$. The canonical dual frame is optimal in the following sense:
\begin{theorem}[Ref. \cite{Daubechiestenlectures}, Prop.~3.2.4]\label{thm:opt-frame}
Let $o_z^\can = \obraket{R^\can_z}{O}$ be the coefficients of $O$ corresponding to the canonical dual frame of $\{|S_z)\}$.
For all dual frames $\{|R_z)\}$ of $\{|S_z)\}$, we have
\begin{align} 
	\sum_{z}\abs{o_z}^2 \geq \sum_{z} \abs{o_z^\can}^2, 
\end{align}
where $o_z = \obraket{R_z}{O}$.	
\end{theorem}

In our context, the frame $\mathbf{F} = \{\oket{S_z}\}$ are the columns of the Hermitian conjugate of the \scrambler~$S$ defined in \cref{sec:framework}:
\begin{align} 
	S^\dag = \big[\oket{S_1},\dots,\oket{S_z},\dots,\oket{ S_{d_\ext}}\big] .
\end{align}
Similarly, a dual frame $\tilde{\mathbf{F}}= \{\oket{R_z}\}$ is given by the columns of the recovery map $R$:
\begin{align} 
	R = \left[\oket{R_1},\dots,\oket{R_z},\dots,\oket{R_{d_\ext}}\right]~.
\end{align}
Given $m$ experimental snapshots $z_1,z_2,\dots,z_m$, our protocol uses the mean
\begin{align} 
	\bar o_{(m)} = \frac{1}{m} \sum_{j = 1}^m o_{z_j}
\end{align}
as an estimator for $\Tr(O\rho)$, where $o_z \equiv (O|R_z)$. From \cref{eq:sample_complexity}, the sample complexity is proportional to the variance of $o_z$.
\begin{align} 
	\Var[o_z] = \sum_z P_z \abs{o_z}^2  - \big | \sum_{z} P_z o_z \big |^2~. \label{eq:var_oz}
\end{align}

While the second term is always $|\Tr(O\rho)|^2$, the first term, and hence the sample complexity, may vary with different choices of dual frame $\mathbf{\tilde F}$.
The optimal procedure for estimating $\Tr(\rho O)$ is then equivalent to constructing an optimal dual frame $\tilde{\mathbf{F}}$.
We use \cref{thm:opt-frame} to construct the optimal dual frame that minimizes the variance $\Var[o_z]$ and, hence, minimizes the sample complexity of estimating $\Tr(O\rho)$. In order to do so, we use the following Lemma:
\begin{lemma}
\label{lemma:rescale_frame}
Given a frame $\mathbf F = \{\oket{S_z}\}$, a rescaling by any set of positive numbers $\{c_z\}$ gives a \textit{rescaled frame} $\mathbf{F}_c = \{c_z \oket{S_z}\}$. Similarly, any dual frame $\tilde{\mathbf F} = \{\oket{R_z}\}$ of a frame $\mathbf F$ may be rescaled into a dual frame of the rescaled frame $\mathbf{F}_c$, given by $\tilde{\mathbf{F}}_c = \{c_z^{-1}\oket{R_z}\}$.
\end{lemma}
\begin{proof}
To show that $\mathbf{F}_c$ is a frame, note that since $(O|S_z)(S_z|O) \geq 0$, given that $\mathbf{F}$ is a frame, $\mathbf{F}_c$ also satisfies the inequalities in \cref{eq:frame_def} with lower and upper bounds $\text{min}(c_z)a$ and $\text{max}(c_z)b$ and hence is also a frame.

We can immediately verify that $\tilde{\mathbf{F}}_c$ is a dual frame of $\mathbf{F}_c$:
\begin{equation}
    \sum_z c_z \oket{S_z}   \frac{1}{c_z} \obra{R_z} =\sum_z \oket{S_z}\obra{R_z} = \mathbb{I}~.
\end{equation}
\end{proof}

To relate the variance to \cref{thm:opt-frame}, we rescale the original frame into the \textit{sample complexity frame} $\mathbf{F}_\text{samp} = \{\oket{S^{(s)}_z}\}$, where $\oket{S^{(s)}_z} = \oket{S_z}/\sqrt{P_z}$ for all $z$. Using \cref{lemma:rescale_frame}, any dual frame $\tilde{\mathbf F} = \{\oket{R_z}\}$ of $\mathbf{F}$ can also be rescaled to obtain a dual frame $\tilde{\mathbf F}_\text{samp} = \{\oket{R^{(s)}_z} = \sqrt{P_z} \oket{R_z}\}$ of $\mathbf{F}_\text{samp}$ and vice versa. $\Tr(O\rho)$ remains correctly estimated in this rescaled frame:
\begin{align} 
	\Tr(O\rho) =  \sum_z \underbrace{(O^\dagger|R^{(s)}_z)}_{= \sqrt{P_z} o_z \equiv o^{(s)}_z} \underbrace{(S^{(s)}_z|\rho)}_{=\sqrt{P_z}} = \sum_z P_z o_z~,
\end{align}
The advantage of $\mathbf{F}_\text{samp}$, however, is that the sample complexity can be related to \cref{thm:opt-frame}.
\begin{align}
    \text{Var}[o_z] = \sum_z P_z \abs{o_z}^2 = \sum_z \big | o_z^{(s)} \big|^2
\end{align}

By \cref{thm:opt-frame}, $\sum_z\big| o^{(s)}_z \big |^2$ is minimized if $\tilde{\mathbf{F}}_\text{samp}$ is the canonical dual frame of  $\mathbf{F}_\text{samp}$. By rescaling the canonical dual frame of $\mathbf{F}_\text{samp}$, we get the optimal dual frame $\tilde{\mathbf{F}}$ to the original frame $\mathbf{F}$:
\begin{align} 
	\oket{R_z} &= \frac{1}{\sqrt{P_z}}\oket{\tilde {S}^{(s)}_z} =\frac{1}{\sqrt{P_z}} \left[\sum_z|S^{(s)}_z)(S^{(s)}_z|\right]^{-1}|S^{(s)}_z) \nonumber\\
	&=\frac{1}{P_z}  	 
	\left[\sum_z \frac{1}{P_z}|S_z)(S_z| \right]^{-1}|S_z)~, \label{eq:optimal-dual-frame_app}
\end{align}
which is \cref{eq:optimal-QL} in the main text, with the optimal $o_z$ given by \cref{eq:optimal-oz}:
\begin{align}
    o_z = \frac{1}{P_z} \sum_{ijkl} O_{j,i} \big(A^{-1}\big)_{(i,j),(k,l)} S_{z,(k,l)}~.
\end{align}
In \cref{fig:rydberg-all}(a), we illustrate the improvement in sample complexity between the inverse map $R$ obtained from the Moore-Penrose pseudoinverse and \cref{eq:optimal-dual-frame_app}.

\subsection{QR decomposition algorithm to solve linear equations} \label{app:QR_decomp_linear_equations}
As discussed in \cref{sec:frame}, it is often computationally favorable to   directly solve the linear equation $S^\dagger \ket{o} = \oket{O}$ for $|o\rangle$
using a QR decomposition on $S$ and then Gaussian elimination, instead of computing the left-inverse $R$ of $S$ and applying it to $|O)$. 
To elaborate, a QR decomposition is first done on the $d_\ext \times d_\sys^2$ rectangular matrix $S$:
\begin{equation}
    S = Q \begin{bmatrix} R_1 \\ \pmb{0} \end{bmatrix}~,
\end{equation}
where $Q$ is a $d_\ext \times d_\ext$ unitary matrix and $R_1$ is a $d_\sys^2 \times d_\sys^2$ upper triangular matrix. One may verify that the solution
\begin{equation}
    \ket{o} = Q \begin{bmatrix} (R_1^\dagger)^{-1} \oket{O} \\ \pmb{0} \end{bmatrix}
\end{equation}
is equal to the Moore-Penrose solution $R_\text{MP}^\dagger \oket{O} = S (S^\dagger S)^{-1} \oket{O}$. The solution $(R_1^\dagger)^{-1} \oket{O}$ can be computed by forward substitution, which is more numerically stable than first computing the inverse $(R_1^\dagger)^{-1}$, then $(R_1^\dagger)^{-1} \oket{O}$. 

Similarly, the optimal estimator for $O$, $\{o_z\} = \ket{o} = \Gamma S (S^\dagger \Gamma S)^{-1} \oket{O}$ can be obtained by instead solving the linear equation $S^\dagger \Gamma^{1/2} \left[\Gamma^{-1/2} \ket{o}\right] = \oket{O}$ to obtain $\Gamma^{-1/2} \ket{o}$ and in turn $\ket{o}$.

\section{Optimal dual frame for nonlinear observables}\label{app:frame-nonlinear}

In \cref{sec:analysis} and \cref{app:frame-theory}, we have constructed the optimal dual frame for estimating the expectation value of a linear observable. 
In this section, we prove that the same dual frame is also optimal in extracting nonlinear observables from the snapshots.
As an example, we discuss the estimation of quadratic quantities.

Suppose we would like to estimate a quadratic quantity: 
\begin{align}
 \langle O \rangle &= \tr(O (\rho \otimes \rho)) \approx \sum_{z,z'} (O|R_z R_{z'}) (S_z|\rho)(S_{z'}|\rho)\nonumber\\
 & =\sum_{z,z'} p(z) p(z') o_{z,z'}~,
\end{align}
for some dual frame $\{|R_z)\}$. 
With finite samples $\{z_i \}_{i=1}^m$, the ``U-statistics" estimator $u_m$ is an unbiased, minimal variance estimator~\cite{m:huangPredictingManyProperties2020,m:hoeffdingClassStatisticsAsymptotically1948}:
\begin{align}
    \sum_{z,z'} p(z) p(z') o_{z,z'} \approx \frac{1}{m(m-1)} \sum_{i \neq j} o_{z_i,z_j} \equiv u_m~.
\end{align}

Hoeffding's theorem gives us the variance of $u_m$:
\begin{align}
    \Var(u_m) &= \frac{2\left( 2(m-2) \Var (o_1(z_1)) + \Var(o(z_1,z_2)) \right)}{m(m-1)}\nonumber\\
    &  \approx \frac{1}{m}\Var(o_1(z_1))~, \label{eq:nonlinear_op_hoeffding}
\end{align}
where $o_1(z) = \sum_{z'} p(z') o_{z,z'} = (\tilde{F}_z|\tr_1\left[\left(\rho\otimes \mathbb{I}\right) O \right])$. Since the variance is dominated by the variance of the linear estimator of $\tr_1\left[\left(\rho\otimes \mathbb{I}\right) O\right]$, the dual frame $\{|R_z)\}$ in \cref{eq:optimal-QL} is also optimal for estimating observables that are quadratic in $\rho$.

For a general non-linear observable, Hoeffding's theorem states that the variance of a non-linear U-statistic is asymptotically dominated by the variance of a linear estimator [corresponding to the first term in \cref{eq:nonlinear_op_hoeffding}]. Therefore, the same dual frame is also asymptotically optimal for non-linear observables.

\section{Performance in the presence of noise}\label{app:noise}
In this section, we demonstrate two strategies for data processing in extracting information from a Rydberg array in the presence of noise.
The first strategy processes the data as if there were no noise in the scrambling quench, resulting in systematic errors in the extracted values. 
We numerically show that the systematic errors are well within our estimate in \cref{sec:noise}.
The second strategy assumes that we have a good description of the noisy channel so that we can numerically compute the corresponding inverse \scrambler~Q. 
We show that the increased sample complexity due to noise agrees well with our prediction in \cref{sec:noise}.

\begin{figure}[t]
\includegraphics[width=0.45\textwidth]{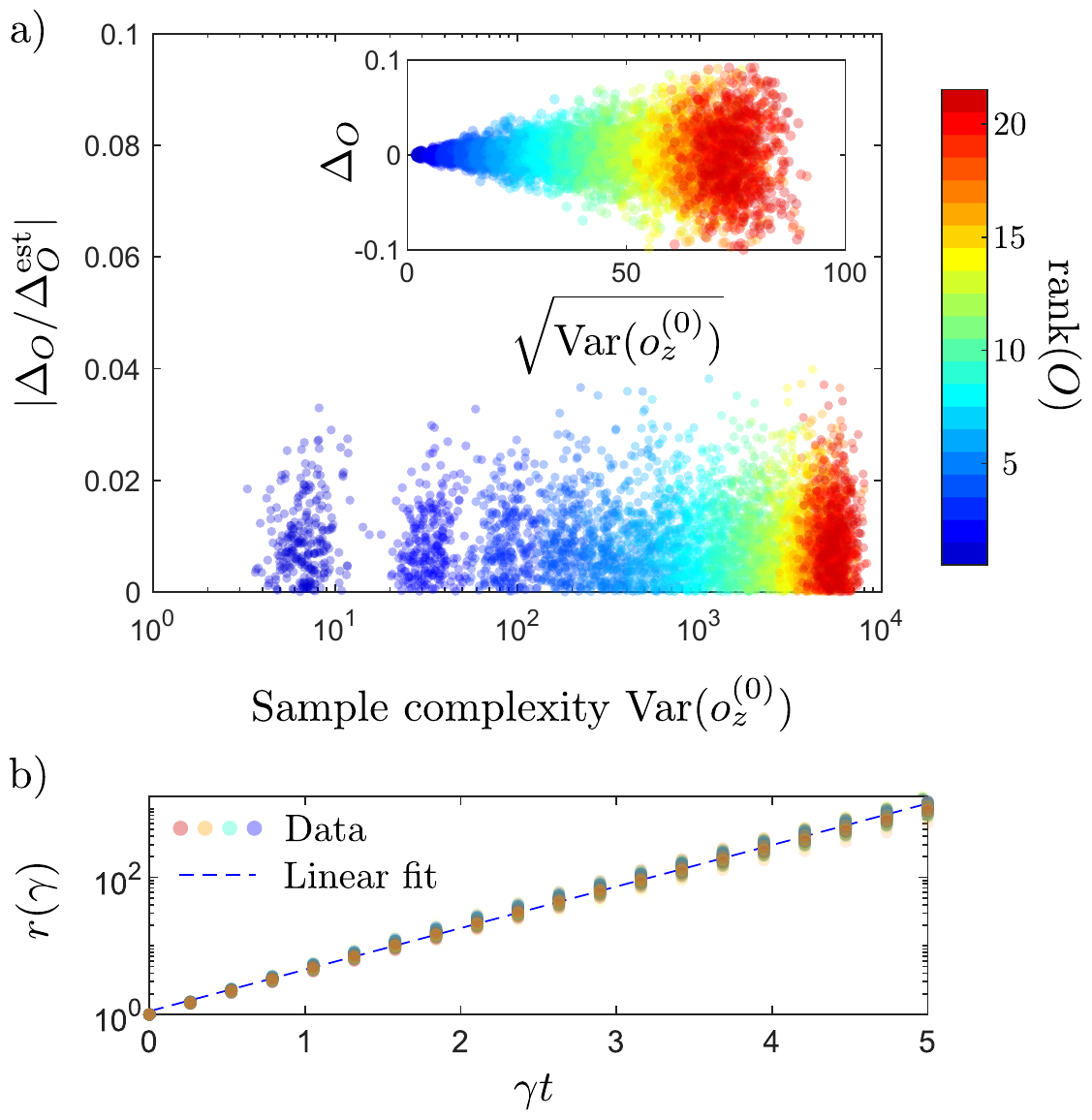}
\caption{The performance of our protocol in the extracting quantum state fidelity [\cref{fig:rydberg-all}(a)] in the presence of local dephasing noise.
a) In the first strategy, we process the data as if there were no experimental error. The main plot is the between the exact systematic error $\Delta_O$ and the bound $\Delta_O^\est$ in \cref{eq:back-of-envelope-bound} as a function of the sample complexity $\Var(o_z^{(\gamma)})$ (evaluated in a noiseless evolution).
Each scatter point corresponds to a randomly chosen rank-$k$ observables for $k = 1,\dots,21$.
The initial state is the ground state of $H_\ryd$ with $n_\sys = 6$, $\Delta = -\Omega$, $V_2 = 0$.
We extract the expectation values of the observables by evolving the extended system, including 7 ancillas in the global setup (\cref{fig:rydberg-all}a), for time $t = 12\pi/\Omega$ under the noisy channel described in \cref{eq:noisy-rydberg-with-dephasing} at $\gamma t = 0.02$, $\Delta = -\Omega$, and $V_2 = 0.2\Omega$.
The inset plots the systematic error versus the noiseless sample complexity $M_O$, confirming our prediction that low-sample-complexity observables are also more robust against noise. 
b) In the second strategy, we instead take the error model into account in processing the measurement snapshots. 
We plot the ratio $r(\gamma)$ between the noisy sample complexity $\Var(o_z^{(\gamma)})$ and its noiseless version $\Var(o_z^{(o)})$ as a function of $\gamma t$. 
The overlapping scatter points correspond to random observables of ranks between one and $d_\sys = 21$. 
The dashed line is a linear fit to the mean ratio as a function of $\gamma t$ in the log-linear scale.
}
\label{fig:noise-systematic}
\end{figure}

Specifically, we consider the extraction of information from a Rydberg array of $n_\sys = 6$ atoms using the global setup described in \cref{fig:rydberg-all}a, where the scrambling quench of the extended system is affected by local dephasing noise. 
To simulate the local dephasing noise, we Trotterize the noisy scrambling quench channel into $N$ time steps:
\begin{align} 
	\mathcal Q^{(\gamma)}_t \approx 
	 \left[\mathcal E_{\gamma t/N}\circ \mathcal Q^{(0)}_{t/N} \right]^{\circ N},\label{eq:noisy-rydberg-with-dephasing}
\end{align}
where $t$ is the scrambling quench time, $\gamma$ is a constant corresponding to the total noise rate, $Q^{(0)}_{t}$ is the noiseless scrambling quench channel of the extended system, and 
\begin{align} 
	 \mathcal E_{p}[\rho_\ext] = (1-p) \rho_\ext + \frac{p}{n_\ext} \sum_{i = 1}^{n_\ext} Z_i \rho_\ext Z_i
\end{align}
is the local dephasing channel with an error probability $p \in [0,1]$.
In our numerics, we choose $N$ large enough such that both $\Omega t/N\ll 1$ and $\gamma t/N\ll 1$.

As discussed in \cref{sec:noise}, we can process the measurement data as if there were no noise in the experiment. 
An advantage of this strategy is that we do not require a precise description of the error model.
However, in doing so, we introduce systematic errors to the extracted values. 
We provide in \cref{eq:back-of-envelope-bound} estimates of such systematic errors under reasonable assumptions. 
Here, we numerically compute the systematic errors in the above example and compare them to the bound in \cref{eq:back-of-envelope-bound}. 

To obtain statistics over a wide range of observables, we extract the expectation value of random rank-$k$ observables:
\begin{align} 
	 O^{(k)} = U_{\Haar}^\dag D^{(k)} U_{\Haar},
\end{align}
where $U_\Haar$ is a Haar-random unitary acting on the system and $D^{(k)}$ is a rank-$k$ diagonal matrix such that
$[D^{(k)}]_{i,i} = 1$ if $1\leq i \leq k$ and $[D^{(k)}]_{i,i} = 0$ otherwise.

In \cref{fig:noise-systematic}a, we compare the exact systematic error $\Delta_O$ [defined in \cref{eq:sys-error-def}] to our bound from \cref{eq:back-of-envelope-bound}:
\begin{align} 
	\Delta_O^{\text{bound}} \equiv  \gamma t\sqrt{\sum_{z}o_z^2/d_\ext}. 
\end{align}
\Cref{fig:noise-systematic} plots the ratio $\abs{\Delta_O/\Delta_O^{\est}}$ as a function of the square root of the sample complexity, $\sqrt{\Var(o_z)}$, evaluated at $\gamma = 0$.
The back-of-the-envelope bound $\Delta_O^{\est}$ appears to well constrain the empirical error $\Delta_O$. 
Additionally, the ratio $\abs{\Delta_O/\Delta_O^{\est}}$ does not increase when we increase the rank of the random observables from one to full-rank, suggesting that the bound in \cref{eq:back-of-envelope-bound} is applicable to a wide range of observables.
In the inset of \cref{fig:noise-systematic}, we also plot the systematic error as a function of the sample complexity.
The inset supports our interpretation in \cref{sec:noise} that the extraction of low-sample-complexity observables is generally robust against noise.

Next, we consider using a different strategy to process the measurement snapshots in the same experiment.
In this strategy, we numerically compute the inverse \scrambler~to the exact noisy evolution.
In contrast to the earlier strategy, there is no systematic error in the data processing.
However, as discussed in \cref{sec:noise}, we expect the noisy channel to leak information to the environment and, thus, we need to collect more samples to recover information to the same precision as in the noiseless scenario.
In \cref{sec:noise}, we argue that the sample complexity would increase exponentially with the noise rate in the presence of a global depolarizing noise.
In \cref{fig:noise-systematic}b, we plot the ratio between the noisy sample complexity ${\Var(o_{z}^{(\gamma)})}$ and the noiseless version for several random observables of different ranks.
The sample complexity increases exponentially with $\gamma t$, supporting our general argument in \cref{sec:noise}.

\section{Many-body Chern number measurement}
\label{app:MBCN}
Here, we provide details of the many-body Chern number (MBCN) operator $\mathcal{T}(\phi)$, proposed in Ref.~\cite{m:dehghaniExtractionManybodyChern2021,m:cianManyBodyChernNumber2021}.

The MBCN is given by the winding number of $\langle\mathcal{T}(\phi)\rangle$, as $\phi$ runs from $0$ to $2\pi$. The non-Hermitian operator $\mathcal{T}(\phi)$ is given by:
\begin{equation}
    \mathcal{T}(\phi) = W_{1}^\dagger(\phi) \mathbb{S}_{1,3} W_{1}^{\mathstrut}(\phi) V_1^s V_2^s~,
    \label{eq:MBCN_single_copy}
\end{equation}
measured on the state of interest $\ket{\Psi}$, taken to be the ground state of an interacting Hamiltonian. $R_{1,2,3}$ are three disjoint, rectangular subsystems of the lattice, and the \textit{polarization} and \textit{twist angle} operators are:~\cite{m:cianManyBodyChernNumber2021}
\begin{align}
 V_i &= \prod_{(x,y) \in R_i} \exp(i \frac{2\pi y} {l_y} n_{x,y})~,\\
 W_i(\phi) &= \prod_{(x,y)\in R_i} \exp(i n_{x,y}\phi)~. 
\end{align}
Lastly, the swap operator $\mathbb{S}_{1,3}$ exchanges the regions $R_1$ and $R_3$:
\begin{align}
     \mathbb{S}_{1,3} &= \prod_{\substack{(x,y) \in R_1\\ (x',y) \in R_3}} \text{SWAP}[(x,y),(x',y)]~.
\end{align}

Finally, $s$ is an integer which equals the expected ground state degeneracy, in our case $s=2$~\cite{m:dehghaniExtractionManybodyChern2021}.

Given $\langle \mathcal{T}(\phi)\rangle$, the MBCN is
\begin{equation}
    C = \frac{1}{2\pi i}\oint \frac{d \langle \mathcal{T}(\phi) \rangle}{\langle \mathcal{T}(\phi) \rangle}~.
    \label{eq:MBCN_formula}
\end{equation}
It is derived with arguments from topological quantum field theory (TQFT): the swap operator serves to engineer a manifold in space-time with non-contractible loops. The twist angle operator $W_1$ applies an artificial electric field along such a non-contractible loop, and the induced polarization $V_1 \otimes V_2$ is measured.

The relation \cref{eq:MBCN_formula} is, in principle, only applicable to systems with periodic or cylindrical boundary conditions. However, the numerical simulations in Ref.~\cite{m:dehghaniExtractionManybodyChern2021} reveal that \cref{eq:MBCN_formula} is also applicable to systems with open boundary conditions, which are more easily realized experimentally.

We also note two alternate formulae for estimating the MBCN, proposed in Refs.~\cite{m:dehghaniExtractionManybodyChern2021,m:cianManyBodyChernNumber2021}: Eqs.~(12) and (13) in Ref.~\cite{m:dehghaniExtractionManybodyChern2021}.
\begin{align}
    &\langle\mathcal{T}(\phi)\rangle = \label{eq:two_swap_chern}\\ 
    &\bra{\Psi}_{A,B}^{\otimes 2}  \big(W^\dagger_{R^A_1}(\phi) V^{s\dagger}_{R^B_2}\mathbb{S}_{R_1^A, R_1^B}\mathbb{S}_{R_3^A, R_3^B} W_{R^A_1}(\phi) V^{s}_{R^A_2}\big) \ket{\Psi}_{A,B}^{\otimes 2} \nonumber\\
    &\langle\mathcal{T}(\phi)\rangle = \label{eq:one_swap_chern}\\
    &\bra{\Psi}_{A,B}^{\otimes 2} \left(V^{s\dagger}_{R^A_1} W_{R^B_2}^\dagger(\phi) \mathbb{S}_{R_1^A, R_1^B}  W_{R^A_2}(\phi) V^{s}_{R^A_1}\right) \ket{\Psi}_{A,B}^{\otimes 2} \nonumber
\end{align}
Both formulae rely on SWAP operations between two identical copies of the wavefunction $\ket{\Psi}$ to engineer a non-trivial spacetime manifold.
In particular, Eq.~\eqref{eq:one_swap_chern} was utilized in Ref.~\cite{m:cianManyBodyChernNumber2021} to be estimated via a randomized measurement protocol: the SWAP on two copies can be done via classical post-processing of randomized measurements on single copies of $\ket{\Psi}$. Furthermore, random unitaries only have to be applied onto a small subsystem $R_1$ in order to estimate Eq.~\eqref{eq:one_swap_chern}---therefore this method can be scaled to large system sizes.

As stated in Ref.~\cite{m:dehghaniExtractionManybodyChern2021}, \cref{eq:two_swap_chern} is more stable than \cref{eq:one_swap_chern} to finite size effects. However, we find that neither Eqs.~\eqref{eq:two_swap_chern} and \eqref{eq:one_swap_chern} are numerically stable for the system sizes investigated.

Instead, we observe that the single-copy measurement \cref{eq:MBCN_single_copy} is a robust estimation formula---in particular, it is particularly robust for the Laughlin state near filling fraction $\nu = 1/2$. In \cref{fig:MBCN_stability}(a) we illustrate the currently-understood phase diagram~\cite{m:cooperFractionalQuantumHall2020} as a function of the filling fraction $\nu = N/N_\text{flux}$, where $N_\text{flux} = \alpha (L_x-1)(L_y-1)$ is the number of flux quanta threading the lattice. At various values of filling fraction, the ground state of the HBH model is believed to be described by states including the Laughlin state, Jain composite fermion states, Moore-Read Pfaffian state~\cite{m:palmBosonicPfaffianState2021}, Read-Rezayi states, and vortex lattice phases. At a critical filling $\nu \approx 2$, the ground state is well described by the Laughlin state---this is confirmed by conventional order parameters such as doublon density~\cite{m:palmBosonicPfaffianState2021}, which sharply drops at the transition into the Laughlin state [\cref{fig:MBCN_stability}(b,c,d)]. At this transition, the MBCN estimator robustly gives $\mathcal{C} = 1$ for many choices of subsystem size and total system size [\cref{fig:MBCN_stability}(e,f,g)]. In other regions of the phase diagram, the MBCN estimator is less robust and shows some sensitivity to the choice of subsystems $R_1,R_2$ and $R_3$.

\begin{figure*}
    \centering
    \includegraphics[width=0.9\textwidth]{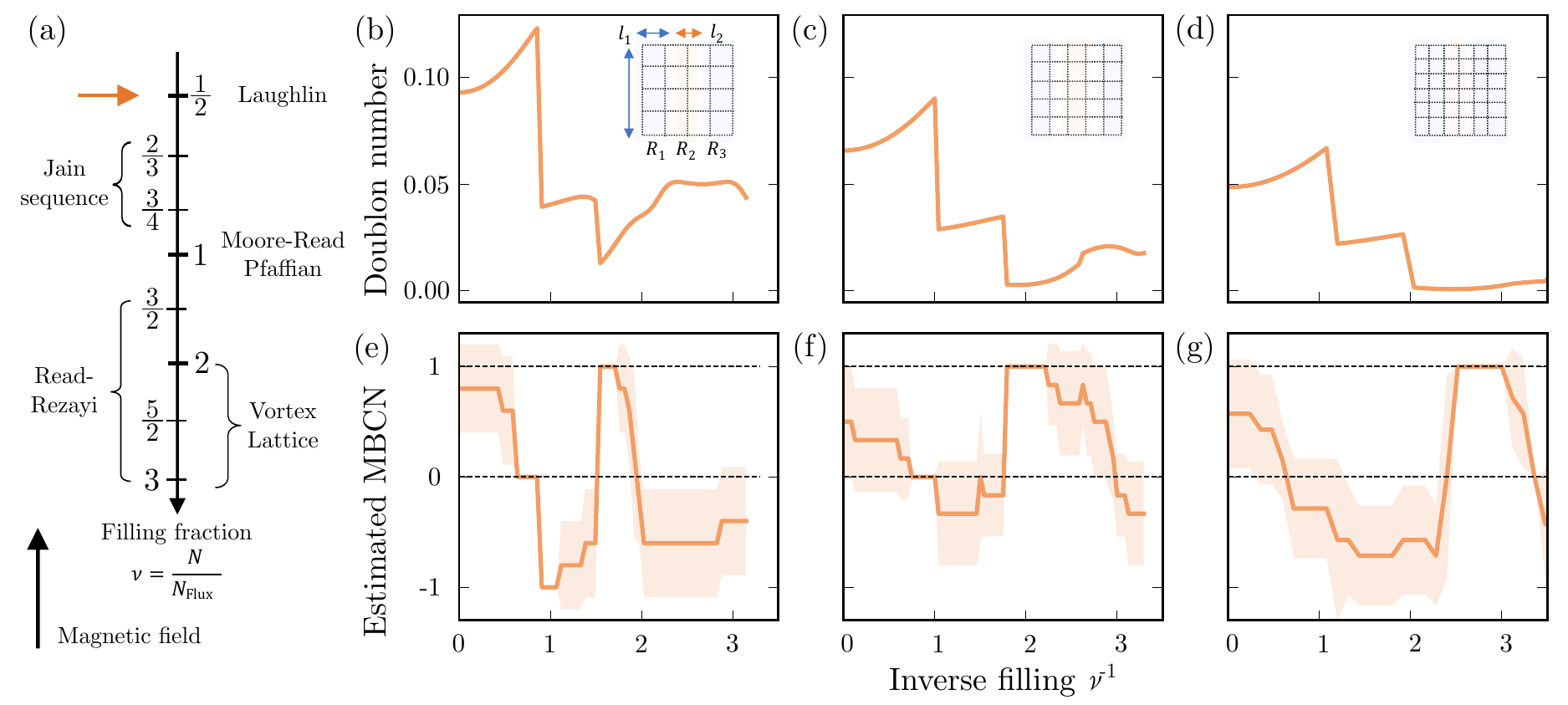}
    \caption{(a) Phase diagram of the HBH model conjectured in Ref.~\cite{m:cooperFractionalQuantumHall2020} as a function of filling fraction $\nu = N/N_\text{flux}$. We consider a Laughlin state: the ground state of the HBH model with three particles on 36 sites and a flux of $\alpha=0.25 \pi$ per plaquette, corresponding to a filling $\nu$ near 0.5. (b-g) Stability of MBCN estimate for the Laughlin state at $\nu = 1/2$. We plot the doublon number (top row) and estimated MBCN (bottom row) as a function of inverse filling $\nu^{-1} = \alpha (L_x-1)(L_y-1)/N$. In each column we indicate (b) $L_x=L_y=5$, (c) $L_x=L_y=6$, and (d) $L_x=L_y=7$, in each case with $N=3$. At a inverse filling $\nu^{-1}=2$, the doublon number abruptly drops, indicating the transition to the Laughlin state. This is accompanied by the estimated MBCN~\cref{eq:MBCN_single_copy}, with $s=2$ jumping to $\mathcal{C}=1$ (e-g). The shaded regions indicate the standard deviation of the estimated MBCN over different choices of subsystems $R_1,R_2$ and $R_3$ used to evaluate \cref{eq:MBCN_single_copy} indicated in the insets of the top row. We choose the parameter values $l_y \in \{L_y-1,L_y\}$, (b) $(l_1,l_2) \in \{(1,2),(1,3),(2,1)\}$, (c) $(l_1,l_2) \in \{(1,3),(1,4),(2,2)\}$, and (d) $(l_1,l_2) \in \{(2,2),(2,3),(3,1)\}$. We note that the Laughlin state is particularly numerically stable, with no variation over the parameters $l_y,l_1$, and $l_2$.}
    \label{fig:MBCN_stability}
\end{figure*}

\end{document}